\documentclass[12pt]{article}
\usepackage{amsmath}
\usepackage{amsthm}
\usepackage{graphicx,psfrag,epsf}
\usepackage{enumerate}
\usepackage{natbib}
\usepackage{hyperref} 

\theoremstyle{definition}
\newtheorem{theorem}{Theorem}
\newtheorem{lemma}[theorem]{Lemma}

\newtheorem{assumption}{Assumption}
\newtheorem{remark}{Remark}
\newtheorem{remark*}{Remark*}

\newcommand{\blind}{1}

\usepackage[margin=1in]{geometry}

\usepackage{amsfonts} 

\usepackage{booktabs}
\usepackage{array}
\usepackage{multirow}
\usepackage{graphicx} 

\usepackage[usenames, dvipsnames]{color}
\definecolor{mypurple}{RGB}{203, 66, 244}
\definecolor{mygray}{RGB}{145,145,145}
\definecolor{mygreen}{RGB}{0,128,0}

\usepackage{xcolor}

\newcommand{\revisionadded}[1]{{#1}}
\newcommand{\revisionaddedtwo}[1]{{#1}}

\usepackage{dsfont}
\newcommand{\indic}{{\mathds 1}}

\newcommand{\eff}{\text{eff}}
\newcommand{\expit}{\text{expit}}
\newcommand{\prob}{\mbox{Pr}}
\newcommand{\PP}{\mathbb{P}}
\newcommand{\var}{\mbox{Var}}

\newcommand{\ba}{\bar{a}}
\newcommand{\bA}{\bar{A}}

\newcommand{\bX}{\bar{X}}

\begin{document}

\def\spacingset#1{\renewcommand{\baselinestretch}%
{#1}\small\normalsize} \spacingset{1}


\if1\blind
{
  \title{\bf Estimating Time-Varying Causal Excursion Effect in Mobile Health with Binary Outcomes}
  \author{
  Tianchen Qian\thanks{Department of Statistics, University of California Irvine. tiancq1@uci.edu},
  Hyesun Yoo\thanks{Department of Statistics, University of Michigan},
  Predrag Klasnja\thanks{School of Information, University of Michigan},
  Daniel Almirall\thanks{Department of Statistics, University of Michigan},\\and
  Susan A. Murphy\thanks{Department of Statistics, Harvard University}}
  \maketitle
} \fi

\if0\blind
{
  \bigskip
  \bigskip
  \bigskip
  \begin{center}
    {\LARGE\bf Title}
\end{center}
  \medskip
} \fi

\bigskip

\begin{abstract}
Advances in wearables and digital technology now make it possible to deliver behavioral mobile health interventions to individuals in their everyday life. The micro-randomized trial (MRT) is increasingly used to provide data to inform the construction of these interventions. \revisionadded{In an MRT, each individual is repeatedly randomized among multiple intervention options, often hundreds or even thousands of times, over the course of the trial.} This work is motivated by multiple MRTs that have been conducted, or are currently in the field, \revisionadded{in which the primary outcome is a longitudinal binary outcome. The primary aim of such MRTs is to examine whether a particular time-varying intervention has an effect on the longitudinal binary outcome, often marginally over all but a small subset of the individual's data. We propose the definition of causal excursion effect that can be used in such primary aim analysis for MRTs with binary outcomes.} Under rather restrictive assumptions one can, based on existing literature, derive a semiparametric, locally efficient estimator of the causal effect. We, starting from this estimator, develop an estimator that can be used as the basis of a primary aim analysis under more plausible assumptions. Simulation studies are conducted to compare the estimators. We illustrate the developed methods using data from the MRT, BariFit. In BariFit, the goal is to support weight maintenance for individuals who received bariatric surgery.
\end{abstract}

\noindent%
\vfill

\newpage
\spacingset{1.45} 

\section{Introduction}
\label{sec:introduction}


In mobile health (mHealth), mobile devices (including smart phones and wearable devices) are used to deliver interventions intended to promote healthy behaviors and health-related behavioral change \citep{free2013effectiveness}. Treatments include prompts to self-monitor, cognitive interventions to promote reflection and goal setting as well as suggestions of ways to enact healthy behavior changes. These treatments are delivered to the individual via the individual's phone or a wearable. 
An increasingly common trial, called the micro-randomized trial (MRT), is being used to inform the development of mHealth interventions \citep{liao2016sample,  klasnja2015, law2016use, klasnja2018efficacy, kramer2019investigating, Tate2019trial}. 
In an MRT, each individual is repeatedly randomized among the multiple options for a treatment, often hundreds or even thousands of times, over the course of the trial. In all cases the randomization probabilities are determined as part of the design of the trial and are thus known.
Between randomizations, covariate data is collected on the individual's current/recent context via sensors and/or self-report, and after each randomization a ``proximal,'' near-time outcome is collected. 
The time-varying treatments and covariates as well as the time-varying proximal outcome comprise the longitudinal data for use in assessing if a treatment has an effect on the proximal outcome and/or in which contexts this effect may be greater or smaller. 

This paper is motivated by our involvement in a number of MRTs in which the primary longitudinal outcome (i.e., the time-varying proximal outcome) is binary. Schematics of these trials can be found at the website of the Methodology Center at the Pennsylvania State University \citep{MethodologyCenter}. \revisionadded{For example, in the Substance Abuse Research Assistance study \citep{rabbi2018sara}, the primary outcome is whether or not the user completed a daily survey. In three other MRTs---namely, the Smart Weight Loss Management MRT \citep{SmartWeightLoss} which is currently in the field, a previous MRT conducted by JOOL Health \citep{bidargaddi2018prompt}, and in the BariFit \citep{BariFit} MRT which we describe in more detail below---the primary outcome is whether or not the user self-monitored on a daily basis.}


\revisionadded{An important second feature of the MRTs that further motivates this paper is that the primary and secondary aims of the trials---i.e., the comparisons that are of primary and secondary interest to the domain scientists---are often marginal comparisons. Here, marginal means, in part, that the treatment comparisons do not condition on the individual's full history of data (or the full history of prior treatment). (In Section \ref{sec:definition} we provide a precise definition.) For example, in the BariFit MRT, the primary aim is to examine the effect of a daily text message reminder on whether or not the user self-monitored their food intake via a food log by the end of each day. Secondary aims in BariFit also focuses on marginal comparisons, such as whether the effect of a daily text message reminder differs depending on whether the user self-monitored on the previous day. In each of these two examples, the domain scientists are interested in a comparison that is marginal over some aspect of the individual's full history of data (including prior treatments).}

\revisionadded{For the domain scientists, estimates of these marginal comparisons are crucial for informing decisions regarding whether to include the treatment in an mHealth intervention package, as are more complex (conditional) treatment effect comparisons that garner understanding of the contexts in which the treatment might be more (or less) effective. Thus, it is important to develop estimators that enable domain scientists to answer any number of scientific questions, from those that are marginal to those that condition on the user's full history.}


In this paper, we consider estimation and inference for causal effects that can be used as the basis of these primary and secondary aim comparisons. \revisionadded{One possible causal effect, which is akin to Robins' treatment ``blip'' in structural nested mean models \citep{robins1994snmm, robins1997snmm}, is one that conditions on the individual's full history. Our first contribution is that we propose a definition of marginal generalizations of this effect, which we call 
causal excursion effects. Such effects are marginal over all but a subset of the individual's prior data (with the aforementioned subset chosen by domain scientists), which is well-suited for primary and secondary aim comparisons. Furthermore, such effects can be viewed as ``excursions'' as they represent a causal effect of a time-varying treatment occurring over an interval of time extending into the future. In this case the definition of the excursion effect involves rules for how further treatments, if any, would occur during this interval of time. This can be used to answer questions that naturally arise in MRTs such as ``what is the effect of delivering a treatment now then not delivering any treatment for the next $m$ time points''. Lastly, these causal effects, as they are marginal over prior treatment assignments, can be interpreted as contrasts between excursions from the treatment protocol as specified by the micro-randomization. This informs how the current treatment protocol might be improved via moderation analysis on how these causal effects differ by individual's contexts.}

We, based on \citet{robins1994snmm}, provide a semiparametric, locally efficient estimator for the causal effect that is conditional on the full history. \revisionadded{Our second contribution is that we, starting from this estimator, develop an estimator that consistently estimates the causal excursion effect, which can be conditional on an arbitrary subset of the history.} The estimator is robust in the sense that, for consistency, it does not require that the model for the proximal outcome under no treatment to be correctly specified. We propose to use this estimator as the basis of primary and secondary aim comparisons for MRTs with binary outcomes.

\section{Preliminaries}
\label{sec:prelim}


\subsection{Micro-randomized trials and BariFit}
\label{subsec:MRT-BariFit}

As introduced in Section \ref{sec:introduction}, micro-randomized trials (MRTs) provide longitudinal data for use in developing mHealth interventions \citep{liao2016sample, dempsey2015, klasnja2015}.  BariFit, for example, is an MRT that was conducted to aid in the process of developing an mHealth intervention for promoting weight maintenance among individuals who received bariatric surgery \citep{BariFit}.
In this study a daily text reminder might be sent to encourage the participant to self-monitor their food intake via a food log; we will refer to this daily text reminder as the food tracking reminder.

In an MRT each participant is randomized, with known probabilities, between the treatment options at predetermined time points. In BariFit, the food tracking reminder is randomized with probability 0.5 between deliver versus do not deliver every morning for 112 days. In general, the randomization probability can vary depending on the individual's data observed up to that time. 

 In BariFit, the proximal outcome for the food tracking reminder is whether the participant completes their food log on that day. The analysis method developed here focuses on this proximal outcome.  However, it is conjectured that these reminders will assist the individual in building up healthy habits, so that longer term effects are desired. Thus, in defining the causal effects below, we do not assume that longer term effects are absent.


Because treatments are delivered to individuals during their everyday life, there may be unethical or unsafe times at which it is inappropriate or deemed excessively burdensome to deliver a treatment. For example, if the treatment is a smartphone notification that audibly pings and makes the phone light up, it is inappropriate to deliver the smartphone notification when the individual might be operating a motor vehicle \citep{klasnja2018efficacy}. \revisionadded{In these moments the individual is deemed ``unavailable'' for treatment.} Randomization occurs only at available time points, and the causal effect is conditional on the available times \citep{boruvka2018}. 
Due to the fact that many MRTs involve considerations of availability, the methods developed below accommodate this. However, in the case of the BariFit food tracking reminders, they were sent, if at all, early in the morning and, as text messages remain in the phone, the participant is able to read them at a time s/he deems convenient. Thus in the BariFit study, lack of availability is not a consideration.

\subsection{Related literature and our contribution}
\label{subsec:literature}




Because data from mHealth studies are often longitudinal, generalized estimating equations \citep{liang1986longitudinal} and random effects models \citep{laird1982random} are the most commonly used methods for modeling the time-varying association between two or more variables in mHealth studies \citep{schwartz2007analysis, bolger2013intensive}. However, in the presence of time-varying treatments or time-varying covariates, it is well known that the use of these methods can result in biased causal effect estimates without strong and often unrealistic assumptions \citep{pepeanderson1994, schildcrout2005regression}.

Structural nested mean models (SNMMs) and marginal structural models (MSMs) are two classes of models that facilitate estimation of causal effects of a time-varying treatment on a time-varying outcome, where the treatment assignment mechanism may depend on history variables \citep{robins1994snmm, robins1997snmm, robins2000marginal, robinsetal2000marginal}.
In an SNMM, the effect of sequentially removing an amount of treatment on future outcomes, after having removed all future treatments, is modeled. This effect is conditional on all the history information up to that time.
In an MSM, the expectation of the time-varying outcome under a fixed treatment trajectory (possibly conditional on a subset of baseline covariates) is modeled as a function of the treatment trajectory (and the subset of baseline covariates).

The causal excursion effect we considered can be conditional on an arbitrary subset of the history (with the aforementioned subset chosen by domain scientists). Unlike MSM, our approach allows estimation of causal effect modification by time-varying covariates. Unlike SNMM, our causal excursion effect is marginalized over variables not in the subset of the history, i.e., possibly marginal over a large part of the history variables. This makes the estimand coherent with the goal of primary and secondary aim comparisons, and avoids modeling the relationship between the time-varying outcome and the past history in MRTs, where the number of time points can be numerous and the history can be high-dimensional. \revisionadded{\citet[Section 7]{robins2004optimal} considered a related marginalization idea for SNMM, the so-called ``marginal SNMM''. The difference between our causal excursion effect and the effect in Robins' marginal SNMM is that a marginal SNMM considers a conditional set that grows over time, and the goal of the marginal SNMM is to estimate optimal treatment regime with respect to the conditional set that may not be the full history. Our consideration of the marginalization is for primary and secondary analyses, and we consider general conditional sets that might not be nested.} A related marginalization idea was also considered by \citet{neugebauer2007hrmsm} in the ``history-restricted'' extension of MSM. Furthermore, the causal excursion effect can be defined as a contrast between two treatment excursions (i.e., time-varying treatments over an interval of time) extended into the future; this excursion aspect was not considered in either SNMM or MSM.

In the previous work on data analytic methods for MRTs, \citet{boruvka2018} and \citet{ dempsey2017stratified} considered estimation of causal effects of mHealth interventions, where the outcome is continuous.
In this paper we consider binary outcome, and we address the unique challenges in the binary outcome by considering a log relative risk model for the causal excursion effect and by developing a novel estimator.


\section{Definition and assumptions}
\label{sec:definition}

\subsection{Notation and observed data}

Suppose that for each individual, there are $T$ time points at which the treatment can be delivered ($T$ doesn't need to be the same for each individual). For simplicity we assume that there are two treatment options which we will call treatment and no treatment. Thus,   the treatment assignment at time $t$, $A_t$, is binary, where $1$ means treatment and $0$ means no treatment.
Denote by $X_t$ the vector of observations collected after time $t-1$ and up to time $t$; $X_1$ includes baseline covariates. $X_t$ contains the availability indicator, $I_t$: $I_t = 1$ if the individual is available for treatment at time $t$ and $I_t = 0$ otherwise. If $I_t = 0$, randomization will not occur at time $t$ and $A_t = 0$.
We use overbar to denote a sequence of variables up to a time point; for example $\bA_t = (A_1, \ldots, A_t)$. Information accrued up to time $t$ is represented by the history $H_t = (X_1, A_1, X_2, A_2, \ldots, X_{t-1}, A_{t-1}, X_t) = (\bX_t, \bA_{t-1})$. The randomization probability for $A_t$ can depend on $H_t$, and is denoted by $p_t(H_t) = P(A_t=1|H_t)$; $p_t(\cdot)$ is known by the MRT design.
The observed data on a generic individual, ordered in time, is $O= (X_1, A_1, \ldots, X_T, A_T, X_{T+1})$.
We assume that the data from different individuals are independent and identically distributed draws from an unknown distribution $P_0$. Unless noted otherwise, all expectations are taken with respect to $P_0$.

The proximal outcome, $Y_{t,\Delta}$, following the treatment assignment at time $t$, is a known function of the individual's data within a subsequent window of length $\Delta$, where $\Delta \geq 1$ is a positive integer; i.e., $Y_{t,\Delta} = y(X_{t+1}, A_{t+1}, \ldots, X_{t+\Delta-1}, A_{t+\Delta-1}, X_{t+\Delta})$ for some known function $y(\cdot)$. In this paper $Y_{t,\Delta}$ is binary. For example, in a smoking cessation study where the treatment is a push notification that reminds the individual to practice stress-reduction exercises \citep{Sense2Stop}, the treatment is randomized every minute (albeit with very low probability of sending a push notification at any given minute), and the proximal outcome is whether the individual experiences a stress episode during the 120-minute window following a treatment. In this example, $t$ is every minute, and $\Delta = 120$. A simpler setting with $\Delta = 1$ is where the proximal outcome cannot depend on future treatment and is given by $Y_{t,1}=y(X_{t+1})$; an example is the BariFit MRT described in Section \ref{subsec:MRT-BariFit}, where the randomization occurs once a day, and the proximal outcome is measured in the same day. The estimator we propose in Section \ref{sec:marginal-estimator} allows for general $\Delta$.

For an arbitrary function $f(\cdot)$ of the generic observed data $O$, denote by $\mathbb{P}_n f(O)$ the sample average $\frac{1}{n}\sum_{i=1}^n f(O_i)$ where $O_i$ denotes the $i$th individual's observed data. We omit the subscript $i$ for the $i$th individual throughout the paper unless necessary. We use $\indic$ to denote the indicator function. 

\subsection{Potential outcomes and causal excursion effect}


To define treatment effects, we use the potential outcomes framework \citep{rubin1974estimating,robins1986new}. For an individual, let $X_t(\ba_{t-1})$ and $A_t(\ba_{t-1})$ be the observation that would have been observed and the $t$th treatment that would have been assigned, respectively, if s/he were assigned the treatment sequence $\ba_{t-1}$. Then the potential outcomes are defined as
\begin{equation}
\{X_1, A_1, X_2(a_1), A_2(a_1), X_3(\ba_2),\ldots, A_T(\ba_{T-1}), X_{T+1}(\ba_{T}) \mbox{ for all } \ba_T \in \{0,1\}^{\otimes T} \}, \label{eq:potential-outcome}
\end{equation}
where $\otimes$ denotes the Cartesian product. The potential outcome for the proximal outcome is $Y_{t,\Delta}(\ba_{t+\Delta - 1})$. The treatment at time $t$ in \eqref{eq:potential-outcome} is indexed by past treatments because in an MRT the randomization probabilities can depend on the participant's past treatment. However, for notational simplicity, which will be further justified by Assumption \ref{assumption:consistency} in Section \ref{subsec:identifiability}, henceforth denote $A_2(A_1)$ by $A_2$ and so on with $A_t(\bA_{t-1})$ by $A_t$. The potential history under the observed treatment sequence at time $t$ is $H_t(\bA_{t-1}) = (X_1, A_1, X_2(A_1), A_2, X_3(\bA_2),\ldots, X_t(\bA_{t-1}))$.

We define the causal effect of $A_t$ on $Y_{t,\Delta}$ using the log relative risk scale:
\begin{equation}
\beta_M\{t, S_t(\bA_{t-1})\} = \log \frac{ E \{ Y_{t,\Delta}(\bA_{t-1}, 1, \bar{0}) \mid S_t(\bA_{t-1}), I_t(\bA_{t-1}) = 1 \} }{E \{ Y_{t,\Delta}(\bA_{t-1}, 0, \bar{0}) \mid S_t(\bA_{t-1}), I_t(\bA_{t-1}) = 1 \}}, \label{eq:def-effect-marginal}
\end{equation}
where $S_t(\bA_{t-1})$ is a vector of summary variables formed from $H_t(\bA_{t-1})$, and $\bar{0}$ is a vector of length $\Delta - 1$.
\revisionaddedtwo{We omit the notational dependence of $\beta_M$ on $\Delta$ for simplicity.}
\revisionadded{Expression \eqref{eq:def-effect-marginal} denotes the contrast of the expected outcome under two ``excursions'' from the current treatment protocol: treatment at time $t$ and no treatment for the next $\Delta - 1$ time points, versus no treatment at time $t$ and no treatment for the next $\Delta - 1$ time points. In both excursions the treatment assignment up to time $t$ (i.e., $\bA_{t-1}$) is stochastic and follows the current treatment protocol of the MRT; i.e., the way the treatments are sequentially randomized with randomization probability $p_t(H_t)$.} We call $\beta_M\{t, S_t(\bA_{t-1})\}$ a causal excursion effect. The expectation in \eqref{eq:def-effect-marginal} marginalizes over the randomization distribution of $\bA_{t-1}$ that are not included in $S_t(\bA_{t-1})$. In other words, the meaning of an excursion is relative to how treatments were assigned in the past: at time $t$, we are considering excursions from the current protocol of assigning treatment. The methods developed below generalize to other types of excursions, such as excursions that specify a decision rule at each time between time $t$ and time $t+\Delta-1$.


\revisionaddedtwo{$\beta_M\{t, S_t(\bA_{t-1})\}$ is a marginal generalization of the treatment ``blips''  in a structural nested mean model \citep{robins1994snmm,robins1997snmm}, where $S_t(\bA_{t-1})$ is set to be $H_t(\bA_{t-1})$;} hence the subscript ``M'' in $\beta_M\{t, S_t(\bA_{t-1})\}$. 
We are particularly interested in the marginal effect, because the primary, pre-specified analysis in an MRT usually aim to assess whether a particular intervention component has a marginal effect on the proximal outcome. For such an analysis, one would set $S_t(\bA_{t-1}) = 1$; i.e., the treatment effect is fully marginal. Subsequent analyses usually have a hierarchy of increasingly complex $S_t(\bA_{t-1})$ which includes variables that may modify the treatment effect. In this paper we sometimes call $\beta_M\{t, S_t(\bA_{t-1})\}$ a ``marginal excursion effect'' to emphasize its marginal aspect. \revisionadded{We discuss in more detail the advantage and limitation of our choice of the causal effect definition and why it is well-suited for mHealth and MRTs in Section \ref{sec:discussion}.}

\revisionaddedtwo{A special case of $\beta_M\{t, S_t(\bA_{t-1})\}$ is when $\Delta = 1$ and $S_t(\bA_{t-1})$ is set to $H_t(\bA_{t-1})$:
\begin{equation}
\beta_C\{t, H_t(\bA_{t-1})\} = \log \frac{ E \{ Y_{t,1}(\bA_{t-1}, 1) \mid H_t(\bA_{t-1}), I_t(\bA_{t-1}) = 1 \} }{E \{ Y_{t,1}(\bA_{t-1}, 0) \mid H_t(\bA_{t-1}), I_t(\bA_{t-1}) = 1 \}}. \label{eq:def-effect-conditional}
\end{equation}
As we will see in Section \ref{sec:semiparametric}, an estimator for this fully conditional case can be derived based on the semiparametric efficiency literature, which motivates our proposed estimator for the general $\beta_M\{t, S_t(\bA_{t-1})\}$.}


There has been much debate over the choice of association measure for binary outcomes in the literature, and reasons to prefer relative risk over odds ratio include its interpretability and conditions for collapsibility \citep{greenland1987,lumley2006}. \revisionadded{A drawback of using the relative risk as opposed to odds ratio is that in general the estimated probability of success is not guaranteed to lie in the interval $[0,1]$, unless some alternative parameterization such as \citet{richardson2017} is used.} Nonetheless, we chose to define \eqref{eq:def-effect-marginal} on the relative risk scale, both for interpretability and modeling ease.  See Section \ref{sec:discussion} for further discussion concerning this modeling choice.

\subsection{Identification of parameters}
\label{subsec:identifiability}

To express the causal excursion effect in terms of the observed data, we make the following assumptions.

\begin{assumption}[Consistency]
\label{assumption:consistency}
The observed data equals the potential outcome under observed treatment assignment. In particular, $X_2 = X_2(A_1)$, $A_2 = A_2(A_1)$, and for each subsequent $t \leq T$, $X_t = X_t(\bA_{t-1})$, $A_t = A_t(\bA_{t-1})$, and lastly, $X_{T+1} = X_{T+1}(\bA_T)$. This implies $Y_{t,\Delta} = Y_{t,\Delta}(\bA_{t+\Delta-1})$.
\end{assumption}

\begin{assumption}[Positivity]
\label{assumption:positivity}
If $\prob(H_t = h_t, I_t = 1) > 0$, then $\prob(A_t = a \mid H_t = h_t, I_t = 1) > 0$ for $a \in \{0,1\}$.
\end{assumption}

\begin{assumption}[Sequential ignorability]
\label{assumption:seqignorability}
For $1\leq t \leq T$, the potential outcomes $\{ X_{t+1}(\ba_t),$ $A_{t+1}(\ba_t), \ldots, X_{T+1}(\ba_T): \ba_T \in \{0,1\}^{\otimes T} \}$ are independent of $A_t$ conditional on $H_t$.
\end{assumption}

In an MRT, because the treatment is sequentially randomized with known probabilities bounded away from $0$ and $1$, Assumptions \ref{assumption:positivity} and \ref{assumption:seqignorability} are satisfied by design. Assumption \ref{assumption:consistency} may fail to hold if there is peer influence or social interaction between individuals; for example, in mHealth interventions with social media components, one individual's proximal outcome may be dependent on another individual's treatment assignment, which violates Assumption \ref{assumption:consistency}.
In those cases, a causal inference framework that incorporates interference needs to be used \citep{hong2006evaluating,hudgens2008toward}.
To maintain the focus of this paper we do not consider such settings here. 


We show in Appendix \ref{appen:identifiability} that under Assumptions \ref{assumption:consistency} - \ref{assumption:seqignorability},
the causal excursion effect \eqref{eq:def-effect-marginal} can be written in terms of the observed data distribution:
\begin{align}
\beta_M\{t, S_t(\bA_{t-1})\} = \log \frac{
E \left[ E \left\{ \prod_{j=t+1}^{t+\Delta-1} \frac{\indic(A_j = 0)}{1 - p_j(H_j)} Y_{t,\Delta} \Big| A_t = 1, H_t, I_t = 1 \right\} \Big| S_t, I_t = 1 \right]
}{
E \left[ E \left\{ \prod_{j=t+1}^{t+\Delta-1} \frac{\indic(A_j = 0)}{1 - p_j(H_j)} Y_{t,\Delta} \Big| A_t = 0, H_t , I_t = 1 \right\} \Big| S_t, I_t = 1 \right]
}, \label{eq:identifiability}
\end{align}
where we define $\prod_{j=t+1}^{t+\Delta-1} \frac{\indic(A_j = 0)}{1 - p_j(H_j)} = 1$ if $\Delta = 1$.
With a slight abuse of notation, we denote the right hand side of \eqref{eq:identifiability} by $\beta_M(t, S_t)$. Similarly the treatment effect conditional on full history with $\Delta = 1$ given in (\ref{eq:def-effect-conditional}) can be written as
\begin{align}
\beta_C(t, H_t) = \log \frac{E ( Y_{t,1} \mid A_t = 1, H_t, I_t=1 )}{E ( Y_{t,1} \mid A_t = 0, H_t , I_t=1 )}. \label{eq:blipidentifiability}
\end{align}

\section{A semiparametric, locally efficient estimator}
\label{sec:semiparametric}



To motivate the estimator for the marginal excursion effect $\beta_M(t, S_t)$, we first consider the special case where the treatment effect is conditional on the full history $H_t$ and the proximal outcome is defined with $\Delta = 1$; that is, consider (\ref{eq:blipidentifiability}). 
Using techniques in \cite{robins1994snmm}, the semiparametric efficient score \citep{newey1990} can be derived; a proof is provided in Appendix \ref{appen:proof-efficient-score}. 

\begin{theorem} \label{thm:efficient-score}
Suppose $f(\cdot)$ is a known deterministic function such that for $1 \leq t \leq T$,
\begin{align}
  \beta_C(t, H_t) =\log \frac{E ( Y_{t,1} \mid A_t = 1, H_t, I_t=1 )}{E ( Y_{t,1} \mid A_t = 0, H_t , I_t=1 )} = f(H_t)^T \psi \label{eq:assumption-conditional}
\end{align}
holds for some unknown $p$-dimensional parameter $\psi$, where the expectation is taken over the true distribution of the data, $P_0$. In the semiparametric model characterized by \eqref{eq:assumption-conditional} and Assumptions \ref{assumption:consistency}, \ref{assumption:positivity} and \ref{assumption:seqignorability}, the efficient score for $\psi$ is
\begin{align}
  S_{\eff}(\psi) = \sum_{t=1}^{T} I_t e^{-A_t f(H_t)^T \psi} \{Y_{t,1}-e^{\mu(H_t)+A_t f(H_t)^T \psi } \} 
  K_t\{A_t-p_t(H_t)\} f(H_t), \label{eq:eff-score-binaryA}
\end{align}
where
\begin{align*}
\mu(H_t) &= \log E(Y_{t,1} \mid H_t, A_t = 0, I_t = 1), \\
 K_t & =\frac{e^{f(H_t)^T \psi}}
{ e^{f(H_t)^T \psi} \{ 1-e^{\mu(H_t)} \} p_t(H_t)+ \{ 1-e^{\mu(H_t)+f(H_t)^T \psi} \} \{1-p_t(H_t)\} }.
\end{align*}
\end{theorem}

It follows from semiparametric efficiency theory that the solution $\hat\psi$ to $\PP_n{S_{\eff}(\psi)}=0$ achieves the semiparametric efficiency bound; i.e., it has the smallest asymptotic variance among all semiparametric regular and asymptotically linear estimators for $\psi$ \citep{newey1990, tsiatis2007semiparametric}. Of course this estimator is not practical because $S_{\eff}$ depends on an unknown quantity $\mu(H_t)$. In practice, one can replace $\mu(H_t)$ with a parametric working model and solve for the estimating equation \revisionadded{or construct a two-step estimator where in the first step $\mu(H_t)$ is estimated and its estimate is plugged into $S_{\eff}$ in the second step to form an estimating equation for $\psi$}. $S_{\eff}(\psi)$ is robust to misspecified model for $\mu(H_t)$ (i.e., it has expectation 0 if one replaces $\mu(H_t)$ by an arbitrary function of $H_t$), and the resulting estimator is semiparametric locally efficient, in the sense that it is consistent and when the working model for $\mu(H_t)$ (with all $1\leq t \leq T$) is correctly specified it attains the semiparametric efficiency bound.  

Here we describe a particular implementation of this efficient score; this implementation mainly serves to motivate the proposed method in Section \ref{sec:marginal-estimator}, where we consider estimation of a causal excursion effect in which $\Delta \geq 1$ and for which the causal excursion effect is marginal. 
Let the working model for $\mu(H_t)$ be $g(H_t)^T \alpha$, where $g(H_t)$ is a vector of features constructed from $H_t$ and $\alpha$ is a finite-dimensional parameter that is variationally independent of $\beta$. We combine the resulting estimating function from (\ref{eq:eff-score-binaryA})
with an estimating function for $\alpha$ in the working model to obtain
\begin{align}
m_C(\alpha, \psi) = & \sum_{t=1}^{T} I_t e^{-A_t f(H_t)^T \psi}\{Y_{t,1} - e^{g(H_t)^T \alpha + A_t f(H_t)^T \psi}\} \tilde{K}_t
\begin{bmatrix}
g(H_t) \\
\{ A_t - p_t(H_t) \} f(H_t)
\end{bmatrix}, \label{eq:ee-conditional}
\end{align}
where 
\begin{align*}
 \tilde{K}_t = \frac{e^{f(H_t)^T \psi}}
{ e^{f(H_t)^T \psi}\{ 1-e^{g(H_t)^T\alpha} \} p_t(H_t)+ \{1-e^{g(H_t)^T\alpha+f(H_t)^T \psi} \} \{1-p_t(H_t)\} }.
\end{align*}
\revisionadded{Equation \eqref{eq:ee-conditional} is similar to a g-estimating equation \citep{robins1994snmm}, except that in \eqref{eq:ee-conditional} we are only estimating the immediate effect of a time-varying treatment (rather than estimating the effect of all past treatments) and that we are estimating the nuisance parameter $\alpha$ simultaneously (rather than estimating $\alpha$ in the first step of a two-step estimator.} In Appendix \ref{appen:proof-asymptotics-conditional} we prove the following result.

\begin{theorem} \label{thm:asymptotics-conditional}
Suppose \eqref{eq:assumption-conditional} and Assumptions \ref{assumption:consistency}, \ref{assumption:positivity}, and \ref{assumption:seqignorability} hold, and that the randomization probability $p_t(H_t)$ is known. Let $\dot{m}_C$ be the derivative of $m_C(\alpha, \psi)$ with respect to $(\alpha, \psi)$. Let $(\hat\alpha, \hat\psi)$ be a solution to $\PP_n m_C(\alpha,\psi) = 0$. Suppose $\psi^*$ is the value of $\psi$ corresponding to the data generating distribution, $P_0$.
Under regularity conditions, $\sqrt{n}(\hat\psi - \psi^*)$ is asymptotically normal with mean zero and variance-covariance matrix $\Sigma_C$. A consistent estimator for $\Sigma_C$ is the lower block diagonal $(p\times p)$ entry of the matrix
$\{\PP_n \dot{m}_C(\hat\alpha, \hat\psi)\}^{-1}
\{\PP_n m_C(\hat\alpha, \hat\psi) m_C(\hat\alpha, \hat\psi)^T\}$
$\{\PP_n \dot{m}_C(\hat\alpha, \hat\psi)\}^{-1^T}$.
Furthermore, when $g(H_t)^T \alpha$ is a correct model for $\mu(H_t)$ in the sense that there exists $\alpha^*$ such that $g(H_t)^T \alpha^* = \log E(Y_{t,1} \mid H_t, A_t = 0, I_t = 1)$, $\hat\psi$ achieves the semiparametric efficiency bound of the semiparametric model defined in Theorem \ref{thm:efficient-score}.
\end{theorem}

\begin{remark} \label{rem:robust-m-C}
The resulting estimator $\hat\psi$ is robust in the sense that it is consistent even if $\exp\{g(H_t)^T \alpha\}$ is a misspecified model for $E(Y_{t,1} \mid H_t, I_t = 1, A_t = 0)$. This robustness results from the orthogonality between the so-called ``blipped-down outcome'' \citep{robins1997snmm}, $\exp\{-A_t f(H_t)^T \psi^*\}Y_{t,1}$, and the centered treatment indicator, $A_t - p_t(H_t)$---i.e.,\\ $E[\exp\{-A_t f(H_t)^T \psi^*\}Y_{t,1} \{A_t - p_t(H_t)\} \mid H_t] = 0$---which follows from an important property of the blipped-down outcome: $E[ \exp\{-A_t f(H_t)^T \psi^*\}Y_{t,1} \mid H_t, A_t] = E\{ Y_{t,1}(\bA_{t-1}, 0)\mid H_t, A_t\}$.
This property plays a key role in the robustness of both the estimator in Theorem \ref{thm:asymptotics-conditional} and the estimator we develop in Section \ref{sec:marginal-estimator}.
\end{remark}

\begin{remark} \label{rem:projection-conditional}
\revisionadded{Consistency of the estimator $\hat\psi$ requires that \eqref{eq:assumption-conditional} holds; in other words, it requires that the analysis model $f(H_t)^T\psi$ is a correctly specified model for the true conditional treatment effect $\beta_C(t,H_t)$. When $f(H_t)^T\psi$ is an incorrect model for the true $\beta_C(t,H_t)$, $\hat\psi$ converges in probability to some $\psi'$ whose form is given in Appendix \ref{appen:projection-conditional}.}
\end{remark}

\begin{remark} \label{rem:misspecified-conditional}
\revisionadded{If $g(H_t)$ includes unbounded covariates, then the working model $g(H_t)^T \alpha$ may always be a misspecified model for $\log E(Y_{t,1} \mid H_t, A_t = 0, I_t = 1)$, whose range is $(-\infty, 0]$. One solution is to transform the unbounded covariates before including them in the working model.}
\end{remark}

\section{Estimator for the marginal excursion effect}
\label{sec:marginal-estimator}


Now we focus on estimation of $\beta_M(t, S_t)$ where $S_t$ is an arbitrary subset of $H_t$. Suppose $\Delta \geq 1$ is a positive integer. Recall that
\begin{align*}
  \beta_M(t, S_t) =\log \frac{
E \left[ E \left\{ \prod_{j=t+1}^{t+\Delta-1} \frac{\indic(A_j = 0)}{1 - p_j(H_j)} Y_{t,\Delta} \Big| A_t = 1, H_t, I_t = 1 \right\} \Big| S_t, I_t = 1 \right]
}{
E \left[ E \left\{ \prod_{j=t+1}^{t+\Delta-1} \frac{\indic(A_j = 0)}{1 - p_j(H_j)} Y_{t,\Delta} \Big| A_t = 0, H_t, I_t = 1 \right\} \Big| S_t, I_t = 1 \right]
}.
\end{align*}
\revisionadded{Note that unlike $\beta_C(t, H_t)$ whose definition is independent of the choice of $f(H_t)$, the definition of $\beta_M(t, S_t)$ is dependent on the choice of $S_t$.}
We make a parametric assumption on $\beta_M(t, S_t)$. Suppose that for $1 \leq t \leq T$,
\begin{equation}
\beta_M(t, S_t) = S_t^T \beta \label{eq:model-maginal-binary}
\end{equation}
holds for some $p$-dimensional parameter $\beta$.  Note this model allows for time-dependent effects; $S_t$ could include a vector of basis functions of $t$. The estimation method described below readily generalizes to situations where the parametric model has a known functional form that may be nonlinear; the use of a linear model here enhances presentation clarity.


We propose to use a marginal generalization of the estimating function \eqref{eq:ee-conditional} to estimate $\beta$. The proposed estimating function is
\begin{align}
m_M(\alpha,\beta) = \sum_{t=1}^{T + \Delta - 1} I_t e^{-A_t S_t^T \beta} \{ Y_{t,\Delta} - e^{g(H_t)^T \alpha + A_t S_t^T \beta} \} J_t
\begin{bmatrix}g(H_t)\\
\{A_t - \tilde{p}_t(S_t)\} S_t
\end{bmatrix}, \label{eq:ee-marginal}
\end{align}
where $\exp\{g(H_t)^T \alpha\}$ is a working model for $E\{Y_{t,\Delta}(\bA_{t-1}, 0, \bar{0}) \mid H_t, I_t = 1, A_t = 0\}$. 
Because the model is now on the marginal effect, we apply a weighting and centering technique similar to \citet{boruvka2018}. 
The weight at time $t$ is 
\begin{align}
J_t = \bigg\{ \frac{\tilde{p}_t(S_t)}{p_t(H_t)} \bigg\}^{A_t} \bigg\{ \frac{1 - \tilde{p}_t(S_t)}{1 - p_t(H_t)} \bigg\}^{1 - A_t} \times \prod_{j=t+1}^{t+\Delta-1} \frac{\indic(A_j = 0)}{1 - p_j(H_j)},
\end{align}
where $\tilde{p}_t(S_t) \in (0,1)$ is arbitrary as long as it does not depend on terms in $H_t$ other than $S_t$. The product, $\prod_{j=t+1}^{t+\Delta-1} \indic(A_j = 0)/\{1 - p_j( H_j)\}$, is standard inverse probability weighting for settings with $\Delta>1$. The ratio of probabilities, $\{\tilde{p}_t(S_t)/p_t(H_t) \}^{A_t} [ \{1 - \tilde{p}_t(S_t)\} / \{1 - p_t(H_t)\} ]^{1 - A_t}$, can be viewed as a change of probability: Intuitively, the ratio transforms the data distribution in which $A_t$ is randomized with probability ${p}_t(H_t)$ to a distribution acting as if $A_t$ were randomized with probability $\tilde{p}_t(S_t)$. We thus center $A_t$ with $\tilde{p}_t(S_t)$; this  centering results in orthogonality between the estimation of $\beta$ and the estimation of the nuisance parameter, $\alpha$. The weighting and centering, together with the factor $\exp(-A_t S_t^T \beta)$, makes the resulting estimator for $\beta$ consistent even when the working model $\exp\{g(H_t)^T \alpha\}$ is misspecified.


In Appendix \ref{appen:proof-asymptotics-marginal} we prove the following result.

\begin{theorem} \label{thm:asymptotics-marginal}
Suppose \eqref{eq:model-maginal-binary} and Assumptions \ref{assumption:consistency}, \ref{assumption:positivity}, and \ref{assumption:seqignorability} hold, and that the randomization probability $p_t(H_t)$ is known. Suppose $\beta^*$ is the value of $\beta$ corresponding to the data generating distribution, $P_0$.
Let $\dot{m}_M$ be the derivative of $m_M(\alpha, \beta)$ with respect to $(\alpha, \beta)$. Let $(\hat\alpha, \hat\beta)$ be a solution to $\PP_n m_M(\alpha,\beta) = 0$.
Under regularity conditions, $\sqrt{n}(\hat\beta - \beta^*)$ is asymptotically normal with mean zero and variance-covariance matrix $\Sigma_M$. A consistent estimator for $\Sigma_M$ is the lower block diagonal $(p\times p)$ entry of the matrix
$\{\PP_n \dot{m}_M(\hat\alpha, \hat\beta)\}^{-1}
\{\PP_n m_M(\hat\alpha, \hat\beta) m_M(\hat\alpha, \hat\beta)^T\}$
$\{\PP_n \dot{m}_M(\hat\alpha, \hat\beta)\}^{-1^T}$.
\end{theorem}

\begin{remark}
The consistency of $\hat \beta$ does not require the working model $\exp\{g(H_t)^T \alpha\}$ to be correctly specified. This robustness property is desirable because $H_t$ can be high dimensional in an MRT (where the total number of time points, $T$, can be in the hundreds or even thousands), which makes it difficult to model $E\{Y_{t,\Delta}(\bA_{t-1}, 0, \bar{0}) \mid H_t, I_t = 1, A_t = 0\}$ correctly.
\end{remark}

\begin{remark}
\label{rem:projection-marginal}
Under the assumptions in Theorem \ref{thm:asymptotics-marginal}, the choice of $\tilde{p}_t(S_t)$ doesn't affect the consistency of $\hat\beta$ as long as it depends at most on $S_t$ and it lies in $(0,1)$; it affects the asymptotic variance of $\hat\beta$. On the other hand, when the analysis model $S_t^T\beta$ is an incorrect model for the true $\beta_M(t, S_t)$, $\tilde{p}_t(S_t)$ determines the probability limit of $\hat\beta$. \revisionadded{For example, suppose the data analyst chooses $\Delta = 1$ and $S_t = \emptyset$, i.e., the analysis model is a constant over time, $\beta_0$. Suppose, however, that the true treatment effect $\beta_M(t, \emptyset)$ is not a constant over time. In this case, when we set $\tilde{p}_t(S_t)$ to be any constant in $(0,1)$,} $\hat\beta$ converges in probability to
\begin{equation}
\beta'=\log\frac{\sum_{t=1}^{T}E\{E(Y_{t,1}\mid H_{t}, A_{t}=1) \mid I_t = 1 \} E(I_t)}{\sum_{t=1}^{T}E\{E(Y_{t,1}\mid H_{t}, A_{t}=0) \mid I_t = 1\} E(I_t) }, \nonumber 
\end{equation}
which further simplifies to
\begin{equation*}
\log \frac{\sum_{t=1}^{T}E(Y_{t,1}\mid I_t = 1, A_{t}=1) E(I_t)}{\sum_{t=1}^{T}E(Y_{t,1}\mid I_t = 1, A_{t}=0) E(I_t)}
\end{equation*}
if the randomization probability $p_t(H_t)$ is constant. For general $\Delta$ and $S_t$, the form of the probability limit of $\hat\beta$, $\beta'$, is provided in Appendix \ref{appen:projection-marginal}.
\end{remark}

\begin{remark}
\revisionadded{The estimating equation \eqref{eq:ee-marginal} for $\beta$ is motivated by the locally efficient estimating equation \eqref{eq:ee-conditional} for $\psi$ in terms of the use of the ``blipping-down'' factor, $\exp(-A_t S_t^T \beta)$. Note that unlike $\hat\psi$ from \eqref{eq:ee-conditional}, the estimator $\hat\beta$ from \eqref{eq:ee-marginal} may not be semiparametric locally efficient. We were not able to derive the semiparametric efficiency bound for $\beta$ in \eqref{eq:model-maginal-binary}. We suspect that the semiparametric efficient estimating equation for $\beta$ would include an analogue of $\tilde{K}_t$, but it is not straightforward how to obtain the marginal analogue of $\tilde{K}_t$ for estimating $\beta$.}
\end{remark}

\begin{remark}
\revisionadded{In the definition of $\beta_M(t,S_t)$, $\Delta$ characterizes the length of the excursion (i.e., the number of time points) into the future. Since the excursion considered here specifies no treatment at the $\Delta-1$ time points following the current time point, consideration of a  large $\Delta$ may result in instability in the estimator based on \eqref{eq:ee-marginal}.  This would occur if  the randomization probability to treatment at any time point  is much greater than zero. To see this, note that the summand in the estimating equation contributes to the estimation (i.e., is not a constant zero) only if $J_t$ is nonzero, which holds only if $A_j = 0$ for all $ t+1 \leq j \leq t + \Delta - 1$. Therefore, for large $\Delta$, when the randomization probability is close to zero, many of the observed treatment trajectories will contribute to the estimation; when the randomization probability is much greater than zero, only very few treatment trajectories will contribute to the estimation, making the estimator unstable with large variance.}
\end{remark}

\section{Simulation}
\label{sec:simulation}

\subsection{Overview}
\label{subsec:simu-overview}



In the simulation we focus on the causal excursion effect with $\Delta = 1$, and we conduct two simulation studies to evaluate the proposed estimator of the marginal excursion effect (``EMEE'') in Section \ref{sec:marginal-estimator} and the semiparametric, locally efficient estimator of the conditional effect (``ECE'') described in Section \ref{sec:semiparametric}.

Because the sandwich estimator for the variance of EMEE in Theorem \ref{thm:asymptotics-marginal} can be anti-conservative when the sample size is small, we adopt the small sample correction technique in \citet{mancl2001} to modify the term $\PP_n m_M(\hat\alpha, \hat\beta)^{\otimes 2}$ in the variance estimator. In particular, we pre-multiply the vector of each individual's residual, $(Y_{t,1} - \exp\{g(H_t)^T \hat\alpha + A_t S_t^T \hat\beta\}: 1\leq t \leq T)$, by the inverse of the identity matrix minus the leverage for this individual. Also, as in \citet{liao2016sample}, we use critical values from a $t$ distribution. In particular, for a known $p$-dimensional vector $c$, to test the null hypothesis $c^T\beta = 0$ or to form two-sided confidence intervals, we use the critical value $t^{-1}_{n-p-q}(1 - \xi/2)$, where $p,q$ are the dimensions of $\beta,\alpha$, respectively, and $\xi$ is the significance level. Similar small sample corrections are applied to ECE as well.

The numerical algorithm that solves $\PP_n m_C(\alpha,\psi) = 0$ can be unstable when the denominator in $\tilde{K}_t$ gets close to 0. This is because $\exp\{g(H_t)^T \alpha\}$ and $\exp\{g(H_t)^T \alpha + f(H_t)^T \psi\}$ are not constrained within $(0,1)$. In our implementation of ECE, to improve the numerical stability we replace $\tilde{K}_t$ in \eqref{eq:ee-conditional} by
\begin{align}
 \frac{e^{f(H_t)^T \psi}}
{ e^{f(H_t)^T \psi}[ 1- \min\{e^{g(H_t)^T\alpha}, \lambda\} ] p_t(H_t)+ [1- \min\{e^{g(H_t)^T\alpha+f(H_t)^T \psi}, \lambda\} ] \{1-p_t(H_t)\} }, \label{eq:weight-mod}
\end{align}
with the truncation parameter value $\lambda = 0.95$.
\revisionadded{In Appendix \ref{appen:ECE-alternative-implementation} we also provide an alternative, two-step implementation of ECE that do not rely on truncation. The two implementations result in similar performance in terms of standard error and bias, so we conjecture that this truncation-based implementation with $\lambda = 0.95$ incurs at most negligible efficiency loss compared to the original ECE estimator (i.e., the one with $\lambda = \infty$).}


Throughout the simulations, we assume that all individuals are available at all time points, and we omit $I_t = 1$ in writing conditional expectations.

R code \citep{Rsoftware} to reproduce the simulation results can be downloaded at https://github.com/tqian/binary-outcome-mrt.

\subsection{Simulation on consistency}

\label{subsec:simulation-consistency}

\revisionadded{
We consider a generative model where an important moderator exists. We illustrate that EMEE consistently estimates both the marginal excursion effect when the moderator is not included in $S_t$ (hence the effect is averaged over the distribution of the moderator) and the causal excursion effect moderation when the moderator is included in $S_t$. On the other hand, ECE only consistently estimates the treatment effect conditional on the full history (hereafter referred to as the ``conditional treatment effect'') when the moderator is included in $f(H_t)$; incorrect use of ECE to estimate the marginal excursion effect by excluding the moderator in $f(H_t)$ results in inconsistent estimates.}


We use the following generative model. The time-varying covariate, $Z_t$, is independent of all variables observed before $Z_t$, and it takes three values $0,1,2$ with equal probability. The randomization probability is constant with $p_t(H_t) = 0.2$. The outcome $Y_{t,1}$ is generated from a Bernoulli distribution with
\begin{align*}
  E(Y_{t,1} \mid H_t, A_t) = \big\{0.2 \indic_{Z_t = 0} + 0.5 \indic_{Z_t = 1} + 0.4 \indic_{Z_t = 2} \big\} e^{A_t (0.1 + 0.3 Z_t)}.
\end{align*}
Here, $Z_t$ moderates the conditional treatment effect: The true conditional treatment effect $\beta_C(t,H_t)$ equals $0.1 + 0.3 Z_t$.

We first consider estimating the fully marginal excursion effect, which equals
\begin{align*}
\beta_0 = \log \frac{E\{ E(Y_{t,1} \mid H_t, A_t = 1) \}}{ E \{ E(Y_{t,1} \mid H_t, A_t = 0) \}} = 0.477.
\end{align*}
This is the setting of a typical primary analysis of MRT. \revisionadded{In order to estimate $\beta_0$, by Theorem \ref{thm:asymptotics-marginal} it is appropriate to use the EMEE estimator with $S_t = 1$. For illustration purpose, we also consider using the ECE estimator with $f(H_t) = 1$ in the simulation. Note that this choice of $f(H_t)$ corresponds to a misspecified model for $\beta_C(t, H_t)$, because the true conditional treatment effect is $\beta_C(t,H_t) = 0.1 + 0.3 Z_t$. (For this generative model, the ECE estimator is appropriate (i.e., its in-probability limit is easily interpretable) only if $Z_t$ is included in $f(H_t)$.)} For comparison, we also include the generalized estimating equations (GEE) estimator for binary outcome with log link in the simulation, because GEE is widely used in analyzing mHealth data \citep{schwartz2007analysis, bolger2013intensive}. We use independence (``GEE.ind'') and exchangeable (``GEE.exch'') as working correlation structures for GEE. We use the working model $g(H_t)^T \alpha = \alpha_0 + \alpha_1 Z_t$ for log of the expected outcome under no treatment, which is misspecified for all estimators.



The simulation result for estimating $\beta_0$ is given in Table \ref{tab:simulation1}; the total number of time points is $T=30$ for each individual. The bias, standard deviation (SD), root mean squared error (RMSE), 95\% confidence interval coverage probability before small sample correction (CP (unadj)) and after small sample correction (CP (adj)) are all computed based on 1000 replicates. As expected, EMEE consistently estimates $\beta_0$, and the incorrect use of ECE (due to misspecification of the conditional treatment effect model with $f(H_t) = 1$) results in an inconsistent estimator for $\beta_0$. The consistency of GEE generally requires the working model $g(H_t)^T \alpha$ to be correct; in other words, it does not have the robustness property as EMEE. The result shows that both GEE.ind and GEE.exch are inconsistent. We also see that small sample correction helps to improve the confidence interval coverage for EMEE.

\begin{table}[htbp]

\caption{\label{tab:simulation1} Performance of EMEE, ECE, GEE.ind, and GEE.exch for the marginal excursion effect $\beta_0$.}
\begin{center}
\begin{tabular}[t]{ccccccc}
\toprule
Estimator & Sample size & Bias & SD & RMSE & CP (unadj) & CP (adj)\\
\midrule
 & 30 & 0.000 & 0.077 & 0.077 & \textbf{0.93} & 0.94\\

 & 50 & 0.001 & 0.057 & 0.057 & 0.94 & 0.95\\

\multirow{-3}{*}{\centering\arraybackslash EMEE} & 100 & 0.000 & 0.041 & 0.041 & 0.95 & 0.95\\
\cmidrule{1-7}
 & 30 & \textbf{0.048} & 0.075 & 0.089 & \textbf{0.85} & \textbf{0.88}\\

 & 50 & \textbf{0.049} & 0.055 & 0.074 & \textbf{0.84} & \textbf{0.85}\\

\multirow{-3}{*}{\centering\arraybackslash ECE} & 100 & \textbf{0.048} & 0.040 & 0.063 & \textbf{0.75} & \textbf{0.76}\\
\cmidrule{1-7}


 & 30 & \textbf{0.041} & 0.073 & 0.084 & \textbf{0.88} & \textbf{0.89}\\

 & 50 & \textbf{0.042} & 0.054 & 0.069 & \textbf{0.86} & \textbf{0.87}\\

\multirow{-3}{*}{\centering\arraybackslash GEE.ind} & 100 & \textbf{0.041} & 0.039 & 0.056 & \textbf{0.80} & \textbf{0.81}\\
\cmidrule{1-7}
 & 30 & \textbf{0.041} & 0.073 & 0.084 & \textbf{0.87} & \textbf{0.89}\\

 & 50 & \textbf{0.042} & 0.054 & 0.069 & \textbf{0.86} & \textbf{0.88}\\

\multirow{-3}{*}{\centering\arraybackslash GEE.exch} & 100 & \textbf{0.041} & 0.039 & 0.056 & \textbf{0.80} & \textbf{0.81}\\
\bottomrule
\end{tabular}
\end{center}
\small{* EMEE: the estimator of the marginal excursion effect proposed in Section \ref{sec:marginal-estimator}. ECE: the semiparametric, locally efficient estimator of the conditional effect described in Section \ref{sec:semiparametric}. GEE.ind: GEE with independence working correlation structure. GEE.exch: GEE with exchangeable working correlation structure.
SD: standard deviation. RMSE: root mean squared error. CP: 95\% confidence interval coverage probability, before (unadj) and after (adj) small sample correction. Boldface indicates when Bias or CP are significantly different, at the 5\% level, from 0 or 0.95, respectively. Sample size refers to the number of individuals in each simulated trial.}
\end{table}

\subsection{Simulation on efficiency}

\label{subsec:simulation-efficiency}

\revisionadded{Using the same generative model as in Section \ref{subsec:simulation-consistency}, we now consider estimating the excursion effect moderation by $Z_t$, which can occur in a typical secondary analysis. We set $S_t = Z_t$ in EMEE and $f(H_t) = Z_t$ in ECE. Because the generative model implies that
\begin{align*}
    \log \frac{E\{E(Y_{t,1} \mid H_t, A_t = 1) \mid Z_t\} }{ E\{E(Y_{t,1} \mid H_t, A_t = 0) \mid Z_t\} } = \log \frac{E(Y_{t,1} \mid H_t, A_t = 1) }{ E(Y_{t,1} \mid H_t, A_t = 0) } = 0.1 + 0.3 Z_t,
\end{align*}
in this case the parameter value in the causal excursion effect and in the conditional treatment effect coincides. In other words, the analysis models for EMEE and ECE are both correct, and hence both estimators should be consistent for $\beta_0 = 0.1$ and $\beta_1 = 0.3$. }

\revisionadded{To assess the relative efficiency between ECE and EMEE, we consider two ways to specify the control variables: an incorrectly specified working model ($g(H_t)^T \alpha = \alpha_0 + \alpha_1 Z_t$), and a correctly specified working model ($g(H_t)^T \alpha = \alpha_0 + \alpha_1 Z_t + \alpha_2 \indic_{Z_t = 2}$). We will assess the relative efficiency between EMEE and ECE as they are both consistent for the same estimands. We also included GEE.ind and GEE.exch for comparison. Because consistency of GEE replies on correct specification of not only the treatment effect model but also the control variables, we expect GEE.ind and GEE.exch to be consistent only when $g(H_t)^T \alpha = \alpha_0 + \alpha_1 Z_t + \alpha_2 \indic_{Z_t = 2}$.}


\revisionadded{The results are given in Tables \ref{tab:simulation2} and \ref{tab:simulation2.5}. In both tables, the total number of time points is $30$ for each individual. The bias, standard deviation (SD), root mean squared error (RMSE), 95\% confidence interval coverage probability before small sample correction (CP (unadj)) and after small sample correction (CP (adj)) are all computed based on 1000 replicates. }

\revisionadded{Table \ref{tab:simulation2} shows the simulation result for EMEE with $S_t = Z_t$ and ECE with $f(H_t) = Z_t$, both with incorrectly specified control variables ($g(H_t)^T \alpha = \alpha_0 + \alpha_1 Z_t$). 
As expected, EMEE and ECE are consistent for $\beta_0$ and $\beta_1$, and GEE.ind and GEE.exch are inconsistent. We also see that ECE is slightly more efficient than EMEE. For example, the relative efficiency between ECE and EMEE for estimating $\beta_0$ with sample size 30 is $(0.20 / 0.18)^2 = 1.23$.}

\revisionadded{Table \ref{tab:simulation2.5} shows the simulation result for EMEE with $S_t = Z_t$ and ECE with $f(H_t) = Z_t$, both with correctly specified control variables ($g(H_t)^T \alpha = \alpha_0 + \alpha_1 Z_t + \alpha_2 \indic_{Z_t = 2}$). 
As expected, all four estimators are consistent for $\beta_0$ and $\beta_1$. In this case, ECE is much more efficient than EMEE due to ECE achieving the semiparametric efficiency bound under correctly specified control variables. For example, the relative efficiency between ECE and EMEE for estimating $\beta_0$ with sample size 30 is $(0.21 / 0.17)^2 = 1.53$.}

\revisionadded{The above results indicate that when the working model $g(H_t)^T\alpha$ is misspecified, there could be slight efficiency gain by using ECE over EMEE when both estimators are consistent for the causal excursion effect. When the working model is correctly specified, the efficiency gain by using ECE can be significant. Additional simulations in Appendix \ref{appen:simulation-all} under other generative models also support this conclusion. Thus if one had adequate data so as to consistently estimate the potentially complex, high dimensional $E(Y_{t,1} \mid H_t, A_t = 0)$ and one felt confident that there are no other covariates in $H_t$ than $S_t$ that interact with treatment (so that the conditional treatment effect and the causal excursion effect are equal with $f(H_t) = S_t$), then it could be worthwhile to use ECE to estimate the causal excursion effect.}

\begin{table}[htbp]

\caption{\label{tab:simulation2} \revisionadded{Performance of EMEE, ECE, GEE.ind, and GEE.exch for the treatment effect modification ($S_t = f(H_t) = Z_t$), when the working model is misspecified ($g(H_t)^T \alpha = \alpha_0 + \alpha_1 Z_t$).}}
\begin{center}
\resizebox{\linewidth}{!}{\begin{tabular}[t]{cccccccccccc}
\toprule
\multicolumn{1}{c}{ } & \multicolumn{1}{c}{ } & \multicolumn{5}{c}{$\beta_0$} & \multicolumn{5}{c}{$\beta_1$} \\
\cmidrule(l{2pt}r{2pt}){3-7} \cmidrule(l{2pt}r{2pt}){8-12}
Estimator & Sample size & Bias & RMSE & SD & CP (unadj) & CP (adj) & Bias & RMSE & SD & CP (unadj) & CP (adj)\\
\midrule
 & 30 & -0.02 & 0.20 & 0.20 & 0.94 & 0.95 & 0.01 & 0.13 & 0.13 & 0.94 & 0.95\\

 & 50 & -0.01 & 0.16 & 0.16 & 0.95 & 0.96 & 0.01 & 0.11 & 0.11 & 0.94 & 0.95\\

\multirow{-3}{*}{\centering\arraybackslash EMEE} & 100 & -0.01 & 0.11 & 0.11 & 0.96 & 0.96 & 0.01 & 0.07 & 0.07 & 0.95 & 0.96\\
\cmidrule{1-12}
 & 30 & -0.02 & 0.18 & 0.18 & 0.94 & 0.95 & 0.01 & 0.12 & 0.12 & \textbf{0.93} & 0.94\\

 & 50 & -0.01 & 0.15 & 0.15 & 0.94 & 0.95 & 0.00 & 0.09 & 0.09 & 0.94 & 0.94\\

\multirow{-3}{*}{\centering\arraybackslash ECE} & 100 & -0.01 & 0.10 & 0.10 & 0.96 & 0.96 & 0.01 & 0.06 & 0.06 & 0.94 & 0.95\\
\cmidrule{1-12}
 & 30 & \textbf{0.14} & 0.21 & 0.15 & \textbf{0.82} & \textbf{0.85} & \textbf{-0.12} & 0.15 & 0.08 & \textbf{0.75} & \textbf{0.78}\\

 & 50 & \textbf{0.15} & 0.19 & 0.12 & \textbf{0.75} & \textbf{0.77} & \textbf{-0.12} & 0.14 & 0.07 & \textbf{0.60} & \textbf{0.63}\\

\multirow{-3}{*}{\centering\arraybackslash GEE.ind} & 100 & \textbf{0.15} & 0.17 & 0.08 & \textbf{0.57} & \textbf{0.58} & \textbf{-0.12} & 0.13 & 0.05 & \textbf{0.33} & \textbf{0.34}\\
\cmidrule{1-12}
 & 30 & \textbf{0.14} & 0.21 & 0.15 & \textbf{0.82} & \textbf{0.85} & \textbf{-0.12} & 0.15 & 0.08 & \textbf{0.75} & \textbf{0.77}\\

 & 50 & \textbf{0.15} & 0.19 & 0.12 & \textbf{0.75} & \textbf{0.77} & \textbf{-0.12} & 0.14 & 0.07 & \textbf{0.60} & \textbf{0.62}\\

\multirow{-3}{*}{\centering\arraybackslash GEE.exch} & 100 & \textbf{0.15} & 0.17 & 0.08 & \textbf{0.57} & \textbf{0.58} & \textbf{-0.12} & 0.13 & 0.05 & \textbf{0.33} & \textbf{0.34}\\
\bottomrule
\end{tabular}}
\end{center}
\small{* EMEE: the estimator of the marginal excursion effect proposed in Section \ref{sec:marginal-estimator}. ECE: the semiparametric, locally efficient estimator of the conditional effect described in Section \ref{sec:semiparametric}. GEE.ind: GEE with independence working correlation structure. GEE.exch: GEE with exchangeable working correlation structure.
SD: standard deviation. RMSE: root mean squared error. CP: 95\% confidence interval coverage probability, before (unadj) and after (adj) small sample correction. Boldface indicates when Bias or CP are significantly different, at the 5\% level, from 0 or 0.95, respectively. Sample size refers to the number of individuals in each simulated trial.}
\end{table}

\begin{table}[htbp]

\caption{\label{tab:simulation2.5} \revisionadded{Performance of EMEE, ECE, GEE.ind, and GEE.exch for the treatment effect modification ($S_t = f(H_t) = Z_t$), when the working model is correctly specified ($g(H_t)^T \alpha = \alpha_0 + \alpha_1 Z_t + \alpha_2 \indic_{Z_t = 2}$)}}
\begin{center}
\resizebox{\linewidth}{!}{
\begin{tabular}{cccccccccccc}
\toprule
\multicolumn{1}{c}{ } & \multicolumn{1}{c}{ } & \multicolumn{5}{c}{$\beta_0$} & \multicolumn{5}{c}{$\beta_1$} \\
\cmidrule(l{3pt}r{3pt}){3-7} \cmidrule(l{3pt}r{3pt}){8-12}
Estimator & Sample size & Bias & RMSE & SD & CP (unadj) & CP (adj) & Bias & RMSE & SD & CP (unadj) & CP (adj)\\
\midrule
 & 30 & 0.00 & 0.21 & 0.21 & \textbf{0.93} & 0.94 & 0.00 & 0.13 & 0.13 & 0.94 & 0.95\\

 & 50 & 0.00 & 0.16 & 0.16 & 0.94 & 0.95 & 0.00 & 0.11 & 0.11 & 0.94 & 0.95\\

\multirow{-3}{*}{\centering\arraybackslash EMEE} & 100 & 0.00 & 0.11 & 0.11 & 0.95 & 0.95 & 0.00 & 0.07 & 0.07 & 0.95 & 0.95\\
\cmidrule{1-12}
 & 30 & 0.01 & 0.17 & 0.17 & 0.94 & 0.95 & -0.01 & 0.12 & 0.12 & \textbf{0.93} & 0.94\\

 & 50 & 0.00 & 0.13 & 0.13 & 0.94 & 0.95 & 0.00 & 0.09 & 0.09 & 0.94 & 0.95\\

\multirow{-3}{*}{\centering\arraybackslash ECE} & 100 & 0.00 & 0.10 & 0.10 & 0.94 & 0.95 & 0.00 & 0.06 & 0.06 & 0.95 & 0.95\\
\cmidrule{1-12}
 & 30 & 0.01 & 0.17 & 0.17 & 0.94 & 0.95 & -0.01 & 0.12 & 0.12 & \textbf{0.93} & 0.94\\

 & 50 & 0.00 & 0.13 & 0.13 & 0.94 & 0.95 & 0.00 & 0.09 & 0.09 & 0.94 & 0.95\\

\multirow{-3}{*}{\centering\arraybackslash GEE.ind} & 100 & 0.00 & 0.10 & 0.10 & 0.94 & 0.94 & 0.00 & 0.06 & 0.06 & 0.94 & 0.95\\
\cmidrule{1-12}
 & 30 & 0.01 & 0.17 & 0.17 & 0.94 & 0.95 & -0.01 & 0.12 & 0.12 & \textbf{0.93} & 0.94\\

 & 50 & 0.00 & 0.13 & 0.13 & 0.94 & 0.95 & 0.00 & 0.09 & 0.09 & 0.94 & 0.95\\

\multirow{-3}{*}{\centering\arraybackslash GEE.exch} & 100 & 0.00 & 0.10 & 0.10 & 0.94 & 0.94 & 0.00 & 0.06 & 0.06 & 0.95 & 0.95\\
\bottomrule
\end{tabular}}
\end{center}
\small{* EMEE: the estimator of the marginal excursion effect proposed in Section \ref{sec:marginal-estimator}. ECE: the semiparametric, locally efficient estimator of the conditional effect described in Section \ref{sec:semiparametric}. GEE.ind: GEE with independence working correlation structure. GEE.exch: GEE with exchangeable working correlation structure.
SD: standard deviation. RMSE: root mean squared error. CP: 95\% confidence interval coverage probability, before (unadj) and after (adj) small sample correction.
Boldface indicates when Bias or CP are significantly different, at the 5\% level, from 0 or 0.95, respectively.
Sample size refers to the number of individuals in each simulated trial.}
\end{table}

\section{Application}
\label{sec:application}


BariFit is a 16-week MRT conducted in 2017 by Kaiser Permanente, which aimed to promote weight maintenance for those who went through Bariatric surgery \citep{BariFit}. In this section, we assess the effect of the food tracking reminder on individual's food log completion rate using estimation methods proposed in this paper. The data set contains 45 participants. The food tracking reminder was randomly delivered to each participant with probability 0.5 every morning as a text message. Because of the form of the intervention, all participants were available for this intervention throughout the study; i.e., $I_t \equiv 1$. The binary proximal outcome, food log completion, is coded as 1 for a day if a participant logged $>0$ calories in the Fitbit app on that day, and 0 otherwise. \revisionadded{The food log completion rate averaged over all participant-days where a reminder is delivered is 0.512, and the food log completion rate averaged over all participant-days where a reminder is not delivered is 0.5118. An exploratory analysis indicates that the effect of the reminder seems to vary greatly among participants, and although there is an overall decreasing trend of the food log completion rate with day-in-study, there is no obvious pattern in terms of how the effect of the reminder varies with day-in-study. Details of the exploratory analysis are in Appendix \ref{appen:eda-barifit}.}


\revisionadded{We used EMEE for assessing the marginal excursion effect as well as the effect moderation by certain baseline and time-varying covariates. For an individual, denote by ``$\text{Day}_t$'' the day-in-study of the $t$-th time point (coded as $0,1,\ldots,111$; i.e., $\text{Day}_t = t - 1$), ``$\text{Gen}$'' the gender of the individual, and $Y_{t,1}$ the indicator of the individual completing their food log for $\text{Day}_t$. For all the analyses in this section, we always included $\text{Day}_t$, $\text{Gen}$, and $Y_{t-1,1}$ (lag-1 outcome) in the control variables $g(H_t)$, as they are prognostic of $Y_{t,1}$ in a preliminary generalized estimating equation fit (results not presented here).}

\revisionadded{We analyze the marginal excursion effect $\beta_0$ of the food tracking reminder on food log completion by setting $S_t = 1$ with the analysis model
\begin{equation*}
    \log \frac{E\{ E(Y_{t,1} \mid H_t, A_t = 1) \}}{ E \{ E(Y_{t,1} \mid H_t, A_t = 0) \}} = \beta_0.
\end{equation*}
We analyze the effect moderation by $S_t$ with $S_t = \text{Day}_t$, $S_t = \text{Gen}$, and $S_t = Y_{t-1,1}$, respectively, with the analysis model
\begin{equation*}
    \log \frac{E\{ E(Y_{t,1} \mid H_t, A_t = 1) \mid S_t \}}{ E \{ E(Y_{t,1} \mid H_t, A_t = 0) \mid S_t \}} = \beta_0 + \beta_1 S_t.
\end{equation*}
Results of the analysis are presented in Table \ref{tab:data-analysis-full}. Neither the marginal excursion effect nor the effect moderation by any of the three moderators are significantly different from zero.}

\begin{table}[htbp]
\caption{\label{tab:data-analysis-full} Analysis result for marginal excursion effect and effect moderation for the effect of food tracking reminder on food log completion rate in BariFit MRT. Estimates reported are on the log relative risk scale.}
\begin{center}
\resizebox{\linewidth}{!}{\begin{tabular}[t]{lcccccccc}
\toprule
\multicolumn{1}{c}{ } & \multicolumn{4}{c}{$\beta_0$} & \multicolumn{4}{c}{$\beta_1$} \\
\cmidrule(l{2pt}r{2pt}){2-5} \cmidrule(l{2pt}r{2pt}){6-9}
Analysis model & Estimate & SE &  95\% CI & $p$-value & Estimate & SE &  95\% CI & $p$-value\\
\midrule
$S_t = 1$            & 0.014 & 0.021 & (-0.028, 0.056) & 0.50 & - & - & - & - \\
$S_t = \text{Day}_t$ & 0.035 & 0.031 & (-0.028, 0.098) & 0.27 & -0.0005 & 0.0007 & (-0.0018, 0.0009) & 0.49 \\
$S_t = \text{Gen}$   & -0.006 & 0.017 & (-0.041, 0.029) & 0.75 & 0.026 & 0.032 & (-0.039, 0.091) & 0.43 \\
$S_t = Y_{t-1,1}$    & 0.017 & 0.095 & (-0.175, 0.209) & 0.86 & -0.003 & 0.094 & (-0.193, 0.186) & 0.97 \\
\bottomrule
\end{tabular}}
\end{center}
\small{* SE: standard error. 95\% CI: 95\% confidence interval. SE, 95\% CI and $p$-value are based on small sample correction described in Section \ref{subsec:simu-overview}.}
\end{table}



The result indicates that no effect of the food tracking reminder is detectable from the data. There are two possible reasons for the result, which are interrelated. One is an insufficient sample size; this study was not sized to test this particular hypothesis (instead, it was sized to test for other intervention components).  The other reason is that the true effect may be small or zero. These findings may inform the next iteration of BariFit study in the following ways. If the researchers want to improve the effectiveness of the food tracking reminder, they may consider implementing it as a notification with a smartphone app. The current reminder is sent as text message, which cannot be tailored to the individual's current context such as location or weather. Such tailoring may improve effectiveness of the reminder. Alternatively, if the researchers no longer wish to investigate the proximal effect of the food tracking reminder, they may choose not to randomize it in the next iteration of BariFit. This might be done by either combining the food tracking reminder with other messages that will be sent in the morning, or to remove the food tracking reminder completely from the intervention. This can help to reduce the burden of the mHealth intervention on the individual.

\section{Discussion}
\label{sec:discussion}

The causal excursion effect $\beta_M(t, S_t)$ defined in this paper is different from the majority of the literature on causal inference in the longitudinal setting \citep{robins1994snmm, robins2000marginal, van2003unified}. Rather than a contrast of the expected outcome under two fixed treatment trajectories, the causal excursion effect is a contrast of two ``excursions'' from the current treatment protocol into the future. In the two excursions, the past treatments are stochastic (with randomization probability determined by the study design), and their distribution is usually integrated over in the marginalization. We argue the causal excursion effect is a suitable estimand for the primary and secondary analyses in MRT for the following reasons. 
\revisionaddedtwo{
First, the causal excursion effect approximates an effect in real world implementation of an mHealth intervention. Because designing an MRT incorporates considerations regarding how the time-varying treatment would be implemented in real life (via considerations of user burden, habituation, etc.), the treatment assignment protocol in an MRT should be a protocol that might be plausibly implemented. For example, domain scientists may aim to deliver around an average of one push notification every other day based on burden considerations; this then is reflected in the MRT treatment protocol via the choice of randomization probability. 
Second, the causal excursion effect provides an indication of effective ``deviations'' from the current treatment protocol, or how it might be improved. Because of its interpretation as ``excursions'' from the current treatment protocol, the causal excursion effect tells us whether a treatment is worth further consideration (via the analysis of the fully marginal effect) and gives an indication of  whether the treatment protocol should be further modified depending on time-varying covariates (via the analysis of effect modification). 
Third, it naturally extends the traditional analyses used in (screening) fractional factorial designs to include marginalization not only over other factors/covariates but also over time.
Lastly, the marginal aspect of the causal excursion effect allows for us to design trials with higher power to detect a meaningful effect with a practical sample size, than if we focus on the alternative meaningful effect conditional on the full history, $\beta_C(t, H_t)$.}

\revisionadded{When is the other effect considered in the paper, the treatment effect conditional on the full history $\beta_C(t, H_t)$, relevant in the analysis of MRT? Because consistent estimation for and interpretation of $\beta_C(t, H_t)$ rely on model \eqref{eq:assumption-conditional} being correct, it is usually not suitable for (often pre-specified) primary and secondary analyses. However, it can be useful for exploratory analysis and hypothesis generation when flexible models such as splines or machine learning algorithms may be used to analyze MRT data, especially if the sample size is large (e.g., when an MRT is deployed through publicly available smartphone applications with thousands of users or more).}

Throughout we treated the model for the proximal outcome under no treatment, $E\{Y_{t,\Delta}(\bA_{t-1},0,\bar{0}) \mid H_t, I_t = 1, A_t = 0\}$, as a nuisance parameter, and used a working model $\exp\{g(H_t)^T \alpha\}$ for this nuisance parameter to reduce noise. In a series of works considering modeling of the treatment effect on a binary outcome in both cross-sectional \citep{richardson2017} and longitudinal settings \citep{Wang2017arXiv}, those authors propose to instead use log odds-product as the nuisance parameter. This way the nuisance parameter is no longer constrained by the treatment effect model on the relative risk scale. (As discussed by these authors, the valid range of $E\{Y_{t,\Delta}(\bA_{t-1},0,\bar{0}) \mid H_t, I_t = 1, A_t = 0\}$ is constrained by the treatment effect model on the relative risk scale, because $E\{Y_{t,\Delta}(\bA_{t-1},1,\bar{0}) \mid H_t, I_t = 1, A_t = 1\}$ must be within $[0,1]$.) We agree that this congeniality issue is critical when prediction is the goal as the nuisance part of the model would then be of interest, or when the consistency of the estimator for the parameters in the treatment effect depends on the correct specification of the nuisance part of the model. In the analysis of MRT data, however, the nuisance part of the model is of minimal interest, and more importantly consistency of the estimation methods developed in this paper do not depend on the correct specification of the nuisance part of the model. Therefore, since the purpose of modeling the nuisance parameter is to reduce noise, we choose to treat $E\{Y_{t,\Delta}(\bA_{t-1},0,\bar{0}) \mid H_t, A_t = 0\}$ as a nuisance parameter, because the interpretability makes it easier for domain scientists to model. \revisionadded{Admittedly, the estimated probability exceeding $[0,1]$ can sometimes cause numerical instability in the semiparametric, locally efficient estimator described in Section \ref{sec:semiparametric}. We addressed this by using a truncation-based implementation in \eqref{eq:weight-mod} or a two-step implementation in Appendix \ref{appen:ECE-alternative-implementation}. Alternative solutions include using congeniel parametrization such as that in \citet{richardson2017}.}

\revisionadded{As pointed out by a reviewer, which we also found in our communication with domain scientists, that the consideration of availability may raise concern of generalizability of the estimated effect. In an MRT, availability, just as the time-varying proximal outcome, is a time-varying outcome. To check for generalizability, baseline variables that are thought to be related to the time-varying outcomes (including the proximal outcome and availability) will be collected in the MRT, and the distribution of those baseline variables will be compared to the distribution of them in the target population. This is similar to what one would do if there are concerns about the generalizability of results of a standard clinical trial to a particular population.}

There are a few directions for future research. First, we have assumed binary treatment in the paper. Extension to treatment with multiple levels could involve modeling the treatment effect (defined as contrast to a reference level) as a function of the treatment level. Second, we have focused on estimating the marginal excursion effect. An interesting extension is to introduce random effects to the excursion effect and allow person-specific predictions. \revisionadded{For example, the exploratory analysis in Appendix \ref{appen:eda-barifit} implies possible treatment effect heterogeneity among individuals in the BariFit data set, and random effects models can be one way to account for such heterogeneity.} With random effects it would be nontrivial to deal with both the nonlinear link function as well as the marginalization.
Third, since there are numerous variables that can be potentially included in $g(H_t)$ for noise reduction,  one could, because of the high dimensionality of $H_t$, consider penalization methods for model selection in building the working model $g(H_t)^T \alpha$.
\revisionadded{Fourth, it would be useful to formulate how to best use the analyses proposed here to inform decision making for future implementation of the mHealth intervention, such as forming a warm-start policy for the mHealth intervention that utilizes reinforcement learning. For example, suppose a moderation analysis finds that the treatment effect is higher on weekends than on weekdays, then one might use weekend/weekday as a tailoring variable of a decision rule.  An active area of current research is how to best use this type of data to form a treatment policy \citep{luckett2019estimating}.}
\revisionadded{Fifth, it would be useful to develop an easily available sample size calculator that provides the  sample size required to achieve a desired power to test for marginal excursion effects.}


\revisionadded{As suggested by a reviewer, we point out that the purpose of focusing on causal excursion effects is for the primary and secondary analyses. Absent additional assumptions, the analysis results cannot be directly used, in general, to form the optimal treatment strategy (``optimal'' in the sense of maximizing certain summary measure at the end of the trial such as the sum of the proximal outcomes over the time points) because the marginal effects do not condition on the full past history. Other methods such as reinforcement learning may be used to estimate the optimal treatment strategy, and this is an active area of research in mHealth.}

Finally, we note that we used an preliminary version of the estimator for the marginal excursion effect in analyzing the effect of push notification on user engagement in \citet{bidargaddi2018prompt}.


\section*{Acknowledgement}

Research reported in this paper was supported by National Institute on Alcohol Abuse and Alcoholism (NIAAA) of the National Institutes of Health under award number R01AA23187, National Institute on Drug Abuse (NIDA) of the National Institutes of Health under award numbers P50DA039838 and R01DA039901, National Institute of Biomedical Imaging and Bioengineering (NIBIB) of the National Institutes of Health under award number U54EB020404, National Cancer Institute (NCI) of the National Institutes of Health under award number U01CA229437, and National Heart, Lung, and Blood Institute (NHLBI) of the National Institutes of Health under award number R01HL125440. The content is solely the responsibility of the authors and does not necessarily represent the official views of the National Institutes of Health.

\section*{Supplementary material}

Supplementary material includes proof of the identifiability result \eqref{eq:identifiability} (Appendix \ref{appen:identifiability}), proof of Theorem \ref{thm:asymptotics-conditional} (Appendix \ref{appen:proof-asymptotics-conditional}), proof of Theorem \ref{thm:asymptotics-marginal} (Appendix \ref{appen:proof-asymptotics-marginal}), the form of the limit of $\hat\psi$ when the conditional treatment effect model for $\beta_C(t,H_t)$ is misspecified (Appendix \ref{appen:projection-conditional}, the form of the limit of $\hat\beta$ in Remark \ref{rem:projection-marginal} for general $\Delta$ (Appendix \ref{appen:projection-marginal}), additional simulation study results (Appendix \ref{appen:simulation-all}), an alternative, two-step implementation of the ECE estimator that does not rely on weight truncation (Appendix \ref{appen:ECE-alternative-implementation}), exploratory analysis result for the BariFit data set (Appendix \ref{appen:eda-barifit}), and proof of Theorem \ref{thm:efficient-score} (Appendix \ref{appen:proof-efficient-score}).

\bibliographystyle{biometrika}

\bibliography{binary-outcome-ref}

\newpage

\appendix

\section*{Appendix}

\numberwithin{equation}{section}
\numberwithin{theorem}{section}
\numberwithin{assumption}{section}
\numberwithin{remark}{section}
\numberwithin{table}{section}
\numberwithin{figure}{section}

\section{Proof of identifiability result \eqref{eq:identifiability}}
\label{appen:identifiability}

\begin{lemma} \label{lem:iden-proofuse}
For any $1 \leq k \leq \Delta$, we have
\begin{align}
& E\{Y_{t,\Delta}(\bA_{t-1},a,\bar{0})\mid H_t,A_t=a,I_t=1 \} \nonumber \\
= & E\bigg\{\prod_{j=t+1}^{t + k - 1}\frac{\indic(A_{j}=0)}{1-p_j(H_{j})}Y_{t,\Delta}(\bA_{t-1},a,\bar{0})\bigg|A_t=a,H_t,I_t=1\bigg\}. \label{eq:lem-proofuse}
\end{align}
\end{lemma}

\begin{proof}[of Lemma \ref{lem:iden-proofuse}]
For $k=1$, \eqref{eq:lem-proofuse} holds because we defined $\prod_{j=t+1}^{t} \frac{\indic(A_j = 0)}{1 - p_j(H_j)} = 1$. In the following we assume $\Delta \geq 2$, and we prove the lemma by induction on $k = 1, \ldots, \Delta$.

Suppose \eqref{eq:lem-proofuse} holds for $k = k_0$ for some $1 \leq k_0 \leq \Delta - 1$. 
Denote by $\zeta = \prod_{j=t+1}^{t + k_0 - 1}\frac{\indic(A_{j}=0)}{1-p_j(H_{j})}Y_{t,\Delta}(\bA_{t-1},a,\bar{0})$. We have
\begin{align}
    & E(\zeta \mid H_{t+k_0}, A_t = a, I_t = 1) \nonumber \\
    = & E(\zeta \mid H_{t+k_0}, A_t = a, I_t = 1) \frac{E\{\indic(A_{t+k_0} = 0) \mid H_{t+k_0}, A_t = a, I_t = 1 \} }{1 - p_{t+k_0}(H_{t+k_0}, A_t = a, I_t = 1)} \nonumber \\
    = & E\bigg\{ \zeta \times \frac{\indic(A_{t+k_0} = 0)}{1 - p_{t+k_0}(H_{t+k_0}, A_t = a, I_t = 1)} \bigg| H_{t+k_0}, A_t = a, I_t = 1 \bigg\} \label{proofuse-iden-11} \\
    = & E\bigg\{ \prod_{j=t+1}^{t + k_0}\frac{\indic(A_{j}=0)}{1-p_j(H_{j})}Y_{t,\Delta}(\bA_{t-1},a,\bar{0}) \bigg| H_{t+k_0}, A_t = a, I_t = 1 \bigg\}, \nonumber
\end{align}
where \eqref{proofuse-iden-11} follows from sequential ignorability (Assumption \ref{assumption:seqignorability}). Therefore, by the induction hypothesis and the law of iterated expectation we have
\begin{align}
    & E\{Y_{t,\Delta}(\bA_{t-1},a,\bar{0})\mid H_t,A_t=a,I_t=1 \} = E(\zeta \mid H_t, A_t = a, I_t = 1) \nonumber \\
    = & E\bigg\{ \prod_{j=t+1}^{t + k_0}\frac{\indic(A_{j}=0)}{1-p_j(H_{j})}Y_{t,\Delta}(\bA_{t-1},a,\bar{0}) \bigg| H_t, A_t = a, I_t = 1 \bigg\},
\end{align}
i.e., we showed that \eqref{eq:lem-proofuse} holds for $k = k_0 + 1$. This completes the proof.
\end{proof}

\begin{proof}[of identifiability result \eqref{eq:identifiability}]

It suffices to show that under Assumptions 1-3, we have
\begin{align}
& E\{Y_{t,\Delta}(\bA_{t-1},a,\bar{0})\mid S_t(\bA_{t-1}),I_t(\bA_{t-1})=1\} \nonumber \\
= & E\bigg[E\bigg\{\prod_{j=t+1}^{t+\Delta-1}\frac{1(A_{j}=0)}{1-p_j(H_{j})}Y_{t,\Delta}\bigg|A_t=a,H_t, I_t = 1\bigg\}\bigg|S_t,I_t=1\bigg]. \label{proofuse-iden-0}
\end{align}

We have the following sequence of equality:
\begin{align}
& E\{Y_{t,\Delta}(\bA_{t-1},a,\bar{0})\mid S_t(\bA_{t-1}),I_t(\bA_{t-1})=1\} \nonumber \\
= & E[E\{Y_{t,\Delta}(\bA_{t-1},a,\bar{0})\mid H_t(\bA_{t-1}), I_t(\bA_{t-1})=1\}\mid S_t(\bA_{t-1}),I_t(\bA_{t-1})=1]   \label{proofuse-iden-1} \\
= & E[E\{Y_{t,\Delta}(\bA_{t-1},a,\bar{0})\mid H_t,I_t=1\}\mid S_t,I_t=1]   \label{proofuse-iden-2} \\
= & E[E\{Y_{t,\Delta}(\bA_{t-1},a,\bar{0})\mid H_t,A_t=a,I_t=1\}\mid S_t,I_t=1]   \label{proofuse-iden-3} \\
= & E\bigg[E\bigg\{\prod_{j=t+1}^{t+\Delta-1}\frac{1(A_{j}=0)}{1-p_j(H_{j})}Y_{t,\Delta}\bigg|A_t=a,H_t, I_t = 1\bigg\}\bigg|S_t,I_t=1\bigg], \label{proofuse-iden-4}
\end{align}
where \eqref{proofuse-iden-1} follows from the law of iterated expectation, \eqref{proofuse-iden-2} follows from consistency (Assumption \ref{assumption:consistency}), \eqref{proofuse-iden-3} follows from sequential ignorability (Assumption \ref{assumption:seqignorability}), and \eqref{proofuse-iden-4} follows from Lemma \ref{lem:iden-proofuse}. This completes the proof.

\end{proof}

\section{Proof of Theorem \ref{thm:asymptotics-conditional}}
\label{appen:proof-asymptotics-conditional}

To establish Theorem \ref{thm:asymptotics-conditional}, we assume the following regularity conditions.

\begin{assumption} \label{assp:regularity-conditional}
Suppose $(\alpha, \psi) \in \Theta$, where $\Theta$ is a compact subset of a Euclidean space. Suppose there exists unique $(\alpha', \psi') \in \Theta$ such that $E\{m_C(\alpha', \psi')\} = 0$.
\end{assumption}

\begin{assumption} \label{assp:regularity-conditional-moment}
Suppose $f(H_t)$ and $g(H_t)$ are bounded for all $t$.
\end{assumption}


\begin{lemma} \label{lem:exp-0-conditional}
Suppose \eqref{eq:assumption-conditional} and Assumptions \ref{assumption:consistency}, \ref{assumption:positivity} and \ref{assumption:seqignorability} hold. Suppose $\psi^*$ is the value of $\psi$ corresponding to the data generating distribution, $P_0$. For an arbitrary $\alpha$, we have
\begin{align}
    E[I_t e^{-A_t f(H_t)^T \psi^*}\{ Y_{t,1} - e^{g(H_t)^T \alpha + A_t f(H_t)^T \psi^*}\} \tilde{K}_t \{ A_t - p_t(H_t) \} f(H_t)] = 0.
\end{align}
\end{lemma}

\begin{proof}[of Lemma \ref{lem:exp-0-conditional}]
By the law of iterated expectation we have
\begin{align}
    & E[I_t e^{-A_t f(H_t)^T \psi^*}\{ Y_{t,1} - e^{g(H_t)^T \alpha + A_t f(H_t)^T \psi^*}\} \tilde{K}_t \{ A_t - p_t(H_t) \} f(H_t)] \nonumber \\
    = & E\big( E[I_t e^{-A_t f(H_t)^T \psi^*}\{ Y_{t,1} - e^{g(H_t)^T \alpha + A_t f(H_t)^T \psi^*}\} \tilde{K}_t \{ A_t - p_t(H_t) \} f(H_t) \mid H_t] \big) \nonumber \\
    = & E\big(E[e^{-A_t f(H_t)^T \psi^*}\{ Y_{t,1} - e^{g(H_t)^T \alpha + A_t f(H_t)^T \psi^*}\} \{ A_t - p_t(H_t) \} \mid H_t, I_t = 1] I_t \tilde{K}_t f(H_t) \big) \nonumber \\
    = & E\big(E[e^{- f(H_t)^T \psi^*}\{ Y_{t,1} - e^{g(H_t)^T \alpha + f(H_t)^T \psi^*}\} \{ 1 - p_t(H_t) \} \mid H_t, I_t = 1, A_t = 1] I_t p_t(H_t) \tilde{K}_t f(H_t) \big) \nonumber \\
    & - E\big(E[\{ Y_{t,1} - e^{g(H_t)^T \alpha}\} p_t(H_t) \mid H_t, I_t = 1, A_t = 0]\{1 - p_t(H_t)\} I_t \tilde{K}_t f(H_t) \big) \nonumber \\
    = & E [ \{ e^{- f(H_t)^T \psi^*} E(  Y_{t,1} \mid H_t, I_t = 1, A_t = 1 ) - E( Y_{t,1} \mid H_t, I_t = 1, A_t = 0 )\} \nonumber \\
    & \times I_t p_t(H_t) \{1 - p_t(H_t)\} \tilde{K}_t f(H_t) ] \label{eq:projection-derivation-conditional} \\
    = & 0, \nonumber
\end{align}
where the last equality follows from \eqref{eq:assumption-conditional}. This completes the proof.
\end{proof}

\begin{proof}[of Theorem \ref{thm:asymptotics-conditional}]

Assumption \ref{assp:regularity-conditional} implies that $(\hat\alpha, \hat\psi)$ converges in probability to $(\alpha', \psi')$, by Theorem 5.9 and Problem 5.27 of \citet{van2000asymptotic}. Because $m_C(\alpha,\psi)$ is continuously differentiable and hence Lipschitz continuous, Theorem 5.21 of \citet{van2000asymptotic} implies that $\sqrt(n)\{(\hat\alpha, \hat\psi) - (\alpha', \psi')\}$ is asymptotically normal with mean zero and covariance matrix $[E \{\dot{m}_C(\alpha', \psi')\}]^{-1} E \{m_C(\alpha', \psi') m_C(\alpha', \psi')^T\} [E \{\dot{m}_C(\alpha', \psi')\}]^{-1^T}$. By the law of large numbers and Slutsky's theorem, this covariance matrix can be consistently estimated by $\{\PP_n \dot{m}_C(\hat\alpha, \hat\psi)\}^{-1} \{\PP_n m_C(\hat\alpha, \hat\psi) m_C(\hat\alpha, \hat\psi)^T\} \{\PP_n \dot{m}_C(\hat\alpha, \hat\psi)\}^{-1^T}$. Furthermore, Assumption \ref{assp:regularity-conditional} and Lemma \ref{lem:exp-0-conditional} imply that $\psi^* = \psi'$, so we proved the asymptotic normality of $\hat\psi$. When $g(H_t)^T \alpha$ is a correct model for $\mu(H_t)$, that $\hat\psi$ attains the semiparametric efficiency bound follows from Theorem \ref{thm:efficient-score}. This completes the proof.

\end{proof}

\section{Proof of Theorem \ref{thm:asymptotics-marginal}}
\label{appen:proof-asymptotics-marginal}

The proof of Theorem \ref{thm:asymptotics-marginal} is similar to the proof of Theorem \ref{thm:asymptotics-conditional}. To establish Theorem \ref{thm:asymptotics-marginal}, we assume the following regularity conditions.

\begin{assumption} \label{assp:regularity-marginal}
Suppose $(\alpha, \beta) \in \Theta$, where $\Theta$ is a compact subset of a Euclidean space. Suppose there exists unique $(\alpha', \beta') \in \Theta$ such that $E\{m_M(\alpha', \beta')\} = 0$.
\end{assumption}

\begin{assumption} \label{assp:regularity-marginal-moment}
Suppose $S_t$, $\exp(S_t)$, $g(H_t)$ and $\exp\{g(H_t)\}$ all have finite forth moment.
\end{assumption}

\begin{lemma} \label{lem:exp-0-marginal}
Suppose \eqref{eq:model-maginal-binary} and Assumptions \ref{assumption:consistency}, \ref{assumption:positivity} and \ref{assumption:seqignorability} hold. Suppose $\beta^*$ is the value of $\beta$ corresponding to the data generating distribution, $P_0$. For an arbitrary $\alpha$, we have
\begin{align}
    E[I_t e^{-A_t S_t^T \beta^*}\{ Y_{t,1} - e^{g(H_t)^T \alpha + A_t S_t^T \beta^*}\} J_t \{ A_t - \tilde{p}_t(S_t) \} S_t] = 0.
\end{align}
\end{lemma}

\begin{proof}[of Lemma \ref{lem:exp-0-marginal}]
By the law of iterated expectation we have
\begin{align}
    & E[I_t e^{-A_t S_t^T \beta^*}\{ Y_{t, \Delta} - e^{g(H_t)^T \alpha + A_t S_t^T \beta^*}\} J_t \{ A_t - \tilde{p}_t(S_t) \} S_t] \nonumber \\
    = & E\big( E[I_t e^{-A_t S_t^T \beta^*}\{ Y_{t, \Delta} - e^{g(H_t)^T \alpha + A_t S_t^T \beta^*}\} J_t \{ A_t - \tilde{p}_t(S_t) \} S_t \mid H_t ] \big) \nonumber \\
    = & E\big(E[e^{-A_t S_t^T \beta^*}\{ Y_{t, \Delta} - e^{g(H_t)^T \alpha + A_t S_t^T \beta^*}\} J_t \{ A_t - \tilde{p}_t(S_t) \} \mid H_t, I_t = 1] I_t S_t \big) \nonumber \\
    = & E\bigg(E \bigg[e^{- S_t^T \beta^*}\{ Y_{t, \Delta} - e^{g(H_t)^T \alpha + S_t^T \beta^*}\} \{ 1 - \tilde{p}_t(S_t) \} \prod_{j=t+1}^{t+\Delta-1} \frac{\indic(A_j = 0)}{1 - p_j(H_j)}  \bigg| H_t, I_t = 1, A_t = 1 \bigg] I_t \tilde{p}_t(S_t) S_t \bigg) \nonumber \\
    & - E\bigg(E \bigg[\{ Y_{t, \Delta} - e^{g(H_t)^T \alpha}\} \tilde{p}_t(S_t) \prod_{j=t+1}^{t+\Delta-1} \frac{\indic(A_j = 0)}{1 - p_j(H_j)}  \bigg| H_t, I_t = 1, A_t = 0 \bigg] I_t \{1 - \tilde{p}_t(S_t)\} S_t \bigg) \nonumber \\
    = & E \bigg[ \bigg\{ e^{- S_t^T \beta^*} E\bigg( \prod_{j=t+1}^{t+\Delta-1} \frac{\indic(A_j = 0)}{1 - p_j(H_j)} Y_{t, \Delta} \bigg| H_t, I_t = 1, A_t = 1 \bigg) \nonumber \\
    & - E\bigg( \prod_{j=t+1}^{t+\Delta-1} \frac{\indic(A_j = 0)}{1 - p_j(H_j)} Y_{t, \Delta} \bigg| H_t, I_t = 1, A_t = 0 \bigg)\bigg\}  I_t \tilde{p}_t(S_t) \{1 - \tilde{p}_t(S_t)\} S_t \bigg] \label{eq:projection-derivation} \\
    = & 0, \nonumber
\end{align}
where the last equality follows from \eqref{eq:model-maginal-binary}. This completes the proof.
\end{proof}

\begin{proof}[of Theorem \ref{thm:asymptotics-marginal}]

Assumption \ref{assp:regularity-marginal} implies that $(\hat\alpha, \hat\beta)$ converges in probability to $(\alpha', \beta')$, by Theorem 5.9 and Problem 5.27 of \citet{van2000asymptotic}. Because $m_M(\alpha,\beta)$ is continuously differentiable and hence Lipschitz continuous, Theorem 5.21 of \citet{van2000asymptotic} implies that $\sqrt(n)\{(\hat\alpha, \hat\beta) - (\alpha', \beta')\}$ is asymptotically normal with mean zero and covariance matrix $[E \{\dot{m}_M(\alpha', \beta')\}]^{-1} E \{m_M(\alpha', \beta') m_M(\alpha', \beta')^T\} [E \{\dot{m}_M(\alpha', \beta')\}]^{-1^T}$. By the law of large numbers and Slutsky's theorem, this covariance matrix can be consistently estimated by $\{\PP_n \dot{m}_M(\hat\alpha, \hat\beta)\}^{-1} \{\PP_n m_M(\hat\alpha, \hat\beta) m_M(\hat\alpha, \hat\beta)^T\} \{\PP_n \dot{m}_M(\hat\alpha, \hat\beta)\}^{-1^T}$. Furthermore, Assumption \ref{assp:regularity-conditional} and Lemma \ref{lem:exp-0-conditional} imply that $\beta^* = \beta'$. This completes the proof.

\end{proof}

\section{Limit of $\hat\psi$ when conditional treatment effect model \eqref{eq:assumption-conditional} is misspecified}
\label{appen:projection-conditional}

\revisionadded{
When \eqref{eq:assumption-conditional} is misspecifed, the limit of $\hat\psi$ is $\psi'$ that satisfies the following equation:
\begin{align}
& E [ \{ e^{- f(H_t)^T \psi' } E(  Y_{t,1} \mid H_t, I_t = 1, A_t = 1 ) - E( Y_{t,1} \mid H_t, I_t = 1, A_t = 0 )\} \nonumber \\
    & \times I_t p_t(H_t) \{1 - p_t(H_t)\} \tilde{K}_t f(H_t) ] = 0. \nonumber
\end{align}
This is derived in \eqref{eq:projection-derivation-conditional} in the proof of Lemma \ref{lem:exp-0-conditional}.
}

\section{Limit of $\hat\beta$ in Remark \ref{rem:projection-marginal} for general $\Delta$}
\label{appen:projection-marginal}

When \eqref{eq:model-maginal-binary} is misspecifed, the limit of $\hat\beta$ is $\beta'$ that satisfies the following equation:
\begin{align}
    & \sum_{t=1}^{T}E\bigg[\bigg(E\bigg\{ Y_{t,\Delta}\prod_{j=t+1}^{t+\Delta-1}\frac{\indic(A_j=0)}{1 - p_j(H_j)}\bigg| H_t, I_t = 1,A_t=1\bigg\} e^{-S_t^{T}\beta'} \nonumber \\
    & - E\bigg\{ Y_{t,\Delta}\prod_{j=t+1}^{t+\Delta-1}\frac{\indic(A_j=0)}{1 - p_j(H_j)}\bigg| H_t, I_t = 1,A_t=0\bigg\} \bigg)I_t \tilde{p}_t(S_t)\{1-\tilde{p}_t(S_t)\}S_t\bigg]=0.
\end{align}
This is derived in \eqref{eq:projection-derivation} in the proof of Lemma \ref{lem:exp-0-marginal}.

\section{Additional simulation results}
\label{appen:simulation-all}

\revisionadded{Here we present additional simulation results on the relative efficiency between EMEE and ECE when the marginal excursion effect equals the conditional effect, in which case both estimators are consistent. We consider realistic settings where the true model $E(Y_{t,1} \mid H_t, A_t = 0)$ is complex and the working model for it is misspecified.}

We use the following generative model. The time-varying covariate $Z_t$ is generated from an autoregressive process: $Z_t = 0.5 Z_{t-1} + \epsilon_t$, where $\epsilon_t \sim N(0,1)$ is independent of all the variables observed prior to $Z_t$. The randomization probability is given by $p_t(H_t) = \min[0.8,\max\{0.2,\expit(\eta Z_t)\}]$, where $\expit(x) = \{ 1 + \exp( - x) \}^{-1}$. 
The proximal outcome $Y_{t,1}$ depends on $(A_{t-1}, Y_{t-1,1}, Z_t, A_t)$ through
\begin{align*}
  E(Y_{t,1} \mid H_t, A_t) = q(Z_t, Y_{t-1,1}, A_{t-1};\gamma) e^{\beta_0 A_t}.
\end{align*}
We consider two different $q(Z_t, Y_{t-1,1}, A_{t-1};\gamma)$:
\begin{align*}
  q_{\exp} (Z_t, Y_{t-1,1}, A_{t-1};\gamma) &= \min\left[0.8,\max\left\{0.1,\exp(-0.4 + \gamma (Z_t - 3) + 0.2 Y_{t-1,1} + 0.2 A_{t-1})\right\}\right], \\
  \text{and} \qquad q_{\expit}(Z_t, Y_{t-1,1}, A_{t-1};\gamma) &= \min\left[0.8,\max\left\{0.1,\expit(-0.5 + \gamma Z_t + 0.2 Y_{t-1,1} + 0.2 A_{t-1})\right\}\right].
\end{align*}
 We fix $\beta_0 = 0.1$.

We consider estimation of $\beta_0$ under the class of generative models with $\eta = -0.5, 0, 0.5$ and $\gamma = 0.1, 0.3, 0.5$. The parameter $\eta$ encodes how the randomization probability depends on $Z_t$, and $\gamma$ encodes the impact of $Z_t$ on the proximal outcome $Y_{t,1}$. We set $f(H_t) = 1$ and $S_t = 1$ in the analysis models of ECE and EMEE, respectively. Because in the generative model $\beta_C(t, H_t) = \beta_M(t, S_t) = \beta_0$, both estimators are consistent for $\beta_0$. We use the working model $g(H_t)^T \alpha = \alpha_0 + \alpha_1 Z_t$, which is misspecified, for both estimators.

Figure \ref{fig:rel-eff} shows the relative efficiency under different combinations of $(\eta, \gamma)$ and the two choices of $q(\cdot)$ calculated from 1000 replicates. The sample size is 50, and the total number of time points for each individual is 20. The relative efficiency between the two estimators ranges between $1.11$ and $1.00$.

\begin{figure}[htbp]
\centering
\caption{Relative efficiency between ECE and EMEE, defined as $\var(\text{EMEE}) / \var(\text{ECE})$. \label{fig:rel-eff}}
\includegraphics[width = 0.7\textwidth]{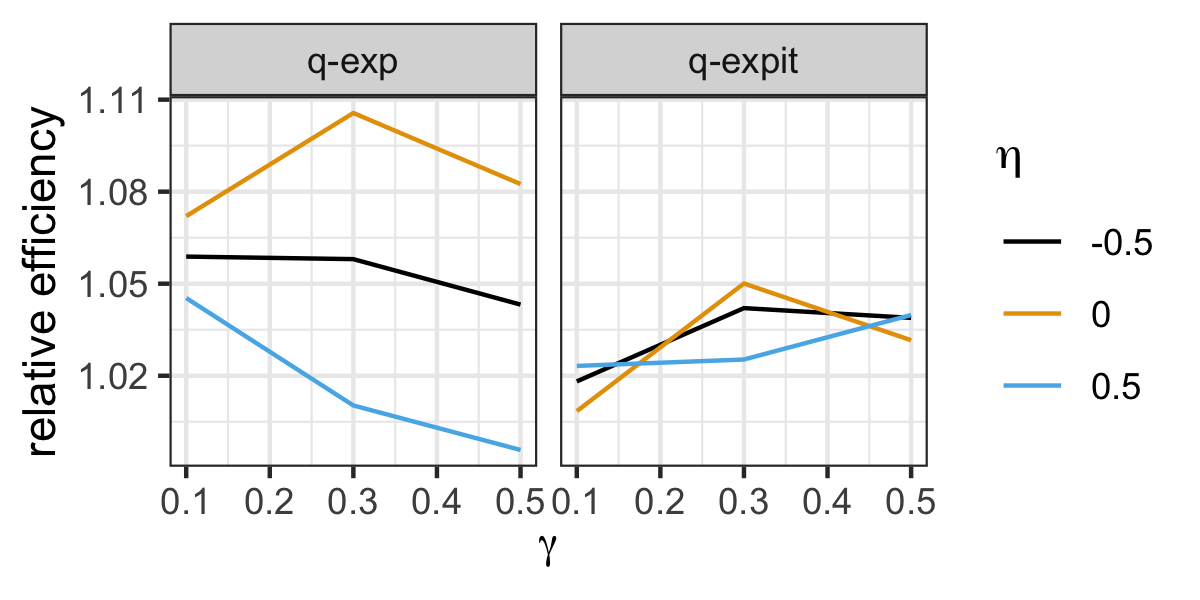}
\end{figure}

\section{An alternative implementation of ECE}
\label{appen:ECE-alternative-implementation}

\revisionadded{In Section \ref{sec:simulation}, we implemented ECE by truncating the terms in the numerator of the optimal weights, $\tilde{K}_t$, for improving stability. A reviewer raised the concern of possible efficiency loss from the truncation. To address this concern, we consider an alternative, two-step implementation of ECE that does not rely on truncating the optimal weights, and we compare its numerical performance with the truncation-based implementation. In this section we use ``ECE.trunc'' to refer to the truncation-based implementation of the ECE estimator (with truncation parameter $\lambda = 0.95$), and we use ``ECE.2step'' to refer to the alternative, two-step implementation. From the simulation results, we conclude that the two implementations of ECE have similar performance in terms of bias and standard deviation. Because ECE.2step does not rely on truncation and is semiparametrically efficient when working models are correct, and ECE.trunc has similar performance to ECE.2step in all settings including when working models are correct, it is likely that the efficiency loss from truncation in ECE.trunc is ignorable. Details are as follows.}

\revisionadded{Given an analysis model $f(H_t)^T \psi$ for the conditional treatment effect $\beta_C(t, H_t)$ and a working model $g(H_t)^T \alpha$ for the outcome under no treatment $E(Y_{t,1}\mid H_t, A_t=0, I_t = 1)$, ECE.2step is computed as follows. First, one computes the maximum likelihood estimator $(\hat\alpha, \hat\psi)$ under the working model $E(Y_{t,1} \mid H_t, A_t, I_t = 1) = \exp\{g(H_t)^T \alpha + A_t f(H_t)^T \psi\}$ and under the working assumption that $(\alpha,\psi)$ is variationally independent with the distribution of the covariate / treatment indicator. Precisely,
\begin{equation*}
    (\hat\alpha, \hat\psi) = \arg\max_{\alpha,\psi} \PP_n \sum_{t=1}^T I_t \bigg[Y_{t,1} \log E(Y_{t,1} \mid H_t, A_t, I_t = 1) + (1-Y_{t,1}) \log \{1 - E(Y_{t,1} \mid H_t, A_t, I_t = 1)\} \bigg].
\end{equation*}
$(\hat\alpha, \hat\psi)$ can be obtained from standard statistical software that implements generalized linear models. Then, one computes the optimal weights using $(\hat\alpha, \hat\psi)$:
\begin{equation*}
     \hat{K}_t = \frac{e^{f(H_t)^T \hat\psi}}
{ e^{f(H_t)^T \hat\psi}\{ 1-e^{g(H_t)^T\hat\alpha} \} p_t(H_t)+ \{1-e^{g(H_t)^T\hat\alpha+f(H_t)^T \hat\psi} \}\{1-p_t(H_t)\} }.
\end{equation*}
Consider the estimating function
\begin{equation*}
    m_{C2}(\psi; \hat\alpha, \hat\psi) = \sum_{t=1}^T I_t \{ e^{-A_t f(H_t)^T \psi} Y_{t,1} - e^{g(H_t)^T \hat\alpha} \} \hat{K}_t
\begin{bmatrix}
g(H_t) \\
\{ A_t - p_t(H_t) \} f(H_t)
\end{bmatrix},
\end{equation*}
where the varying $\psi$ only appears in the blipping-down factor, $e^{-A_t f(H_t)^T \psi}$.
Denote by $\hat\psi_2$ the solution of $\psi$ to the estimating equation $\PP_n m_{C2}(\psi; \hat\alpha, \hat\psi) = 0$. It can be shown under regularity conditions that $\hat\psi_2$ is consistent and asymptotically normal when \eqref{eq:assumption-conditional} and Assumptions \ref{assumption:consistency}, \ref{assumption:positivity}, and \ref{assumption:seqignorability} hold, and that the randomization probability $p_t(H_t)$ is known. Furthermore, $\hat\psi_2$ achieves the semiparametric efficiency bound of the semiparametric model defined in Theorem \ref{thm:efficient-score} when $E(Y_{t,1}\mid H_t, A_t=0, I_t = 1) = g(H_t)^T \alpha^*$ for some $\alpha^*$. The estimator $\hat\psi_2$ is the ECE.2step estimator.}

\revisionadded{To compare the performance of ECE.trunc and ECE.2step, we reuse the simulation settings in Section \ref{subsec:simulation-consistency} and Section \ref{subsec:simulation-efficiency}. Namely, we consider (i) estimating the marginal excursion effect $\beta_0 = 0.477$ with analysis model $f(H_t) = 1$ and control variables $g(H_t)^T\alpha = \alpha_0 + \alpha_1 Z_t$ as in Table \ref{tab:simulation1}, (ii) estimating the causal excursion effect moderation by $Z_t$ with $(\beta_0, \beta_1)^T = (0.1, 0.3)^T$ with analysis model $f(H_t) = (1, S_t)^T$ and control variables $g(H_t)^T\alpha = \alpha_0 + \alpha_1 Z_t$ as in Table \ref{tab:simulation2}, and (iii) estimating the causal excursion effect moderation by $Z_t$ with $(\beta_0, \beta_1)^T = (0.1, 0.3)^T$ with analysis model $f(H_t) = (1, S_t)^T$ and control variables $g(H_t)^T\alpha = \alpha_0 + \alpha_1 Z_t + \alpha_2 \indic_{Z_t = 2}$ as in Table \ref{tab:simulation2.5}. The simulation results are in Table \ref{tab:simulation-ece-2-implementation}, and all simulations are based on 1000 replicates with $T = 30$. We see that ECE.trunc and ECE.2step have almost identical standard deviation for all settings. They also have similar bias for all settings except for when we are estimating the marginal excursion effect $\beta_0$ as in Table \ref{tab:simulation1}, where ECE.trunc has slightly smaller bias than ECE.2step. Therefore, we conclude that the two implementations of ECE have similar performance in terms of bias and standard deviation.}

\begin{table}[htbp]

\caption{\label{tab:simulation-ece-2-implementation} \revisionadded{Comparison of the two implementations of the ECE estimators: ECE.trunc and ECE.2step. The estimand, the analysis model for the conditional treatment effect, and the control variables used in the comparison are those for the ECE estimator described in Tables \ref{tab:simulation1}, \ref{tab:simulation2} and \ref{tab:simulation2.5}.}}
\centering
\resizebox{0.74\linewidth}{!}{\begin{tabular}{cccccccc}
\toprule
Estimand & Sample size & Estimator & Bias & SD & RMSE & CP (unadj) & CP (adj)\\
\midrule
 &  & ECE.2step & 0.063 & 0.075 & 0.098 & 0.81 & 0.84\\

 & \multirow{-2}{*}{\centering\arraybackslash 30} & ECE.trunc & 0.050 & 0.073 & 0.089 & 0.86 & 0.88\\
 \cmidrule(l{2pt}r{2pt}){2-8}

 &  & ECE.2step & 0.060 & 0.058 & 0.083 & 0.79 & 0.80\\

 & \multirow{-2}{*}{\centering\arraybackslash 50} & ECE.trunc & 0.047 & 0.056 & 0.074 & 0.82 & 0.83\\
 \cmidrule(l{2pt}r{2pt}){2-8}

 &  & ECE.2step & 0.060 & 0.039 & 0.072 & 0.65 & 0.67\\

\multirow{-6}{*}{\centering\arraybackslash $\beta_0$ in Table \ref{tab:simulation1}} & \multirow{-2}{*}{\centering\arraybackslash 100} & ECE.trunc & 0.048 & 0.038 & 0.062 & 0.74 & 0.75\\
\cmidrule{1-8}
 &  & ECE.2step & 0.006 & 0.193 & 0.193 & 0.93 & 0.94\\

 & \multirow{-2}{*}{\centering\arraybackslash 30} & ECE.trunc & 0.007 & 0.192 & 0.192 & 0.93 & 0.94\\
 \cmidrule(l{2pt}r{2pt}){2-8}

 &  & ECE.2step & 0.002 & 0.147 & 0.147 & 0.93 & 0.94\\

 & \multirow{-2}{*}{\centering\arraybackslash 50} & ECE.trunc & 0.003 & 0.146 & 0.146 & 0.93 & 0.94\\
 \cmidrule(l{2pt}r{2pt}){2-8}

 &  & ECE.2step & -0.002 & 0.106 & 0.106 & 0.95 & 0.95\\

\multirow{-6}{*}{\centering\arraybackslash $\beta_0$ in Table \ref{tab:simulation2}} & \multirow{-2}{*}{\centering\arraybackslash 100} & ECE.trunc & -0.002 & 0.105 & 0.105 & 0.95 & 0.95\\
\cmidrule{1-8}
 &  & ECE.2step & -0.005 & 0.124 & 0.124 & 0.93 & 0.94\\

 & \multirow{-2}{*}{\centering\arraybackslash 30} & ECE.trunc & -0.006 & 0.122 & 0.123 & 0.93 & 0.94\\
 \cmidrule(l{2pt}r{2pt}){2-8}

 &  & ECE.2step & -0.002 & 0.095 & 0.095 & 0.94 & 0.95\\

 & \multirow{-2}{*}{\centering\arraybackslash 50} & ECE.trunc & -0.003 & 0.094 & 0.094 & 0.94 & 0.95\\
 \cmidrule(l{2pt}r{2pt}){2-8}

 &  & ECE.2step & 0.002 & 0.066 & 0.066 & 0.95 & 0.95\\

\multirow{-6}{*}{\centering\arraybackslash $\beta_1$ in Table \ref{tab:simulation2}} & \multirow{-2}{*}{\centering\arraybackslash 100} & ECE.trunc & 0.002 & 0.066 & 0.066 & 0.95 & 0.96\\
\cmidrule{1-8}
 &  & ECE.2step & 0.006 & 0.175 & 0.175 & 0.94 & 0.95\\

 & \multirow{-2}{*}{\centering\arraybackslash 30} & ECE.trunc & 0.005 & 0.174 & 0.174 & 0.94 & 0.95\\
 \cmidrule(l{2pt}r{2pt}){2-8}

 &  & ECE.2step & 0.000 & 0.134 & 0.134 & 0.94 & 0.95\\

 & \multirow{-2}{*}{\centering\arraybackslash 50} & ECE.trunc & -0.001 & 0.134 & 0.134 & 0.94 & 0.95\\
 \cmidrule(l{2pt}r{2pt}){2-8}

 &  & ECE.2step & 0.001 & 0.096 & 0.096 & 0.94 & 0.95\\

\multirow{-6}{*}{\centering\arraybackslash $\beta_0$ in Table \ref{tab:simulation2.5}} & \multirow{-2}{*}{\centering\arraybackslash 100} & ECE.trunc & 0.001 & 0.096 & 0.096 & 0.94 & 0.95\\
\cmidrule{1-8}
 &  & ECE.2step & -0.006 & 0.117 & 0.117 & 0.93 & 0.94\\

 & \multirow{-2}{*}{\centering\arraybackslash 30} & ECE.trunc & -0.006 & 0.116 & 0.116 & 0.93 & 0.94\\
 \cmidrule(l{2pt}r{2pt}){2-8}

 &  & ECE.2step & -0.001 & 0.090 & 0.090 & 0.94 & 0.95\\

 & \multirow{-2}{*}{\centering\arraybackslash 50} & ECE.trunc & -0.001 & 0.089 & 0.089 & 0.94 & 0.95\\
 \cmidrule(l{2pt}r{2pt}){2-8}

 &  & ECE.2step & 0.000 & 0.063 & 0.063 & 0.94 & 0.95\\

\multirow{-6}{*}{\centering\arraybackslash $\beta_1$ in Table \ref{tab:simulation2.5}} & \multirow{-2}{*}{\centering\arraybackslash 100} & ECE.trunc & 0.000 & 0.063 & 0.063 & 0.95 & 0.95\\
\bottomrule
\end{tabular}}
\end{table}

\section{Exploratory analysis result of BariFit data set}
\label{appen:eda-barifit}

\revisionadded{Here we present the exploratory analysis result of BariFit data, as referred to in Section \ref{sec:application}.}

\revisionadded{We first consider the change of food log completion rate over time. Figure \ref{fig:eda-barifit-by-day} shows the food log completion rate by day-in-study and by treatment. Circles (triangles) represent the food log completion rate of each day-in-study, averaged over all individuals who received treatment (no treatment) on that day. The length of the vertical line segments represents the magnitude of the difference between the circle and the triangle on each day, which is an empirical (crude) estimate for the treatment effect on that day. Solid (dashed) line segments represent positive (negative) treatment effect estimate. We see an overall decreasing trend of food log completion rate for both the treatment-days (circles) and the no-treatment-days (triangles). However, there is no obvious pattern in terms of the empirically estimated day-specific treatment effect, as solid and dashed line segments alternate irregularly.}

\revisionadded{We then consider variation of the food log completion rate across individuals. Figure \ref{fig:eda-barifit-by-user} shows the food log completion rate by individual and by treatment. Circles (triangles) represent the food log completion rate of each individual, averaged over all days when they received treatment (no treatment). The length of the vertical line segments represents the magnitude of the difference between the circle and the triangle for an individual, which is an empirical (crude) estimate for the treatment effect for that individual. Solid (dashed) line segments represent positive (negative) treatment effect estimate. The x-axis is sorted so that the individual-specific empirical treatment effect estimate is decreasing with larger participant id. We see that there is much heterogeneity among individuals in terms of both the completion rate with no treatment (the varying y-location of the triangles) and the empirically estimated treatment effect (the varying length and shape of the line segments).}

\begin{figure}[htp]
    \begin{center}
        \caption{\revisionadded{Food log completion rate by day and treatment.} \label{fig:eda-barifit-by-day}}
        \includegraphics[clip,width=\columnwidth]{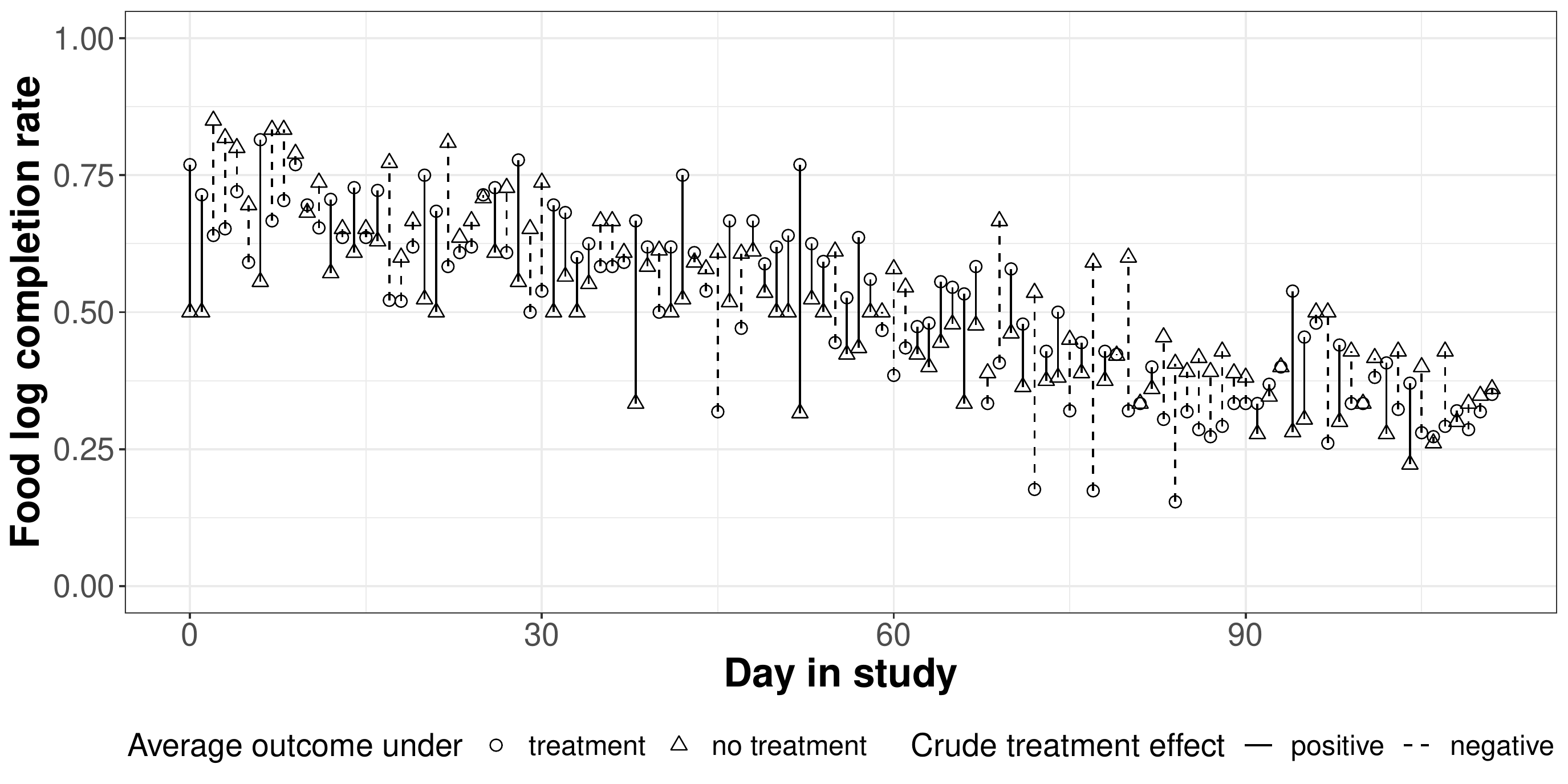}%
    \end{center}
    \small{* Circles (triangles) represent the food log completion rate of each day-in-study, averaged over all individuals who received treatment (no treatment) on that day. The length of the vertical line segments represents the magnitude of the difference between the circle and the triangle on each day, which is an empirical (crude) estimate for the treatment effect on that day. Solid (dashed) line segments represent positive (negative) treatment effect estimate.}
\end{figure}

\begin{figure}[htp]
    \begin{center}
        \caption{\revisionadded{Food log completion rate by individual and treatment.} \label{fig:eda-barifit-by-user}}
        \includegraphics[clip,width=\columnwidth]{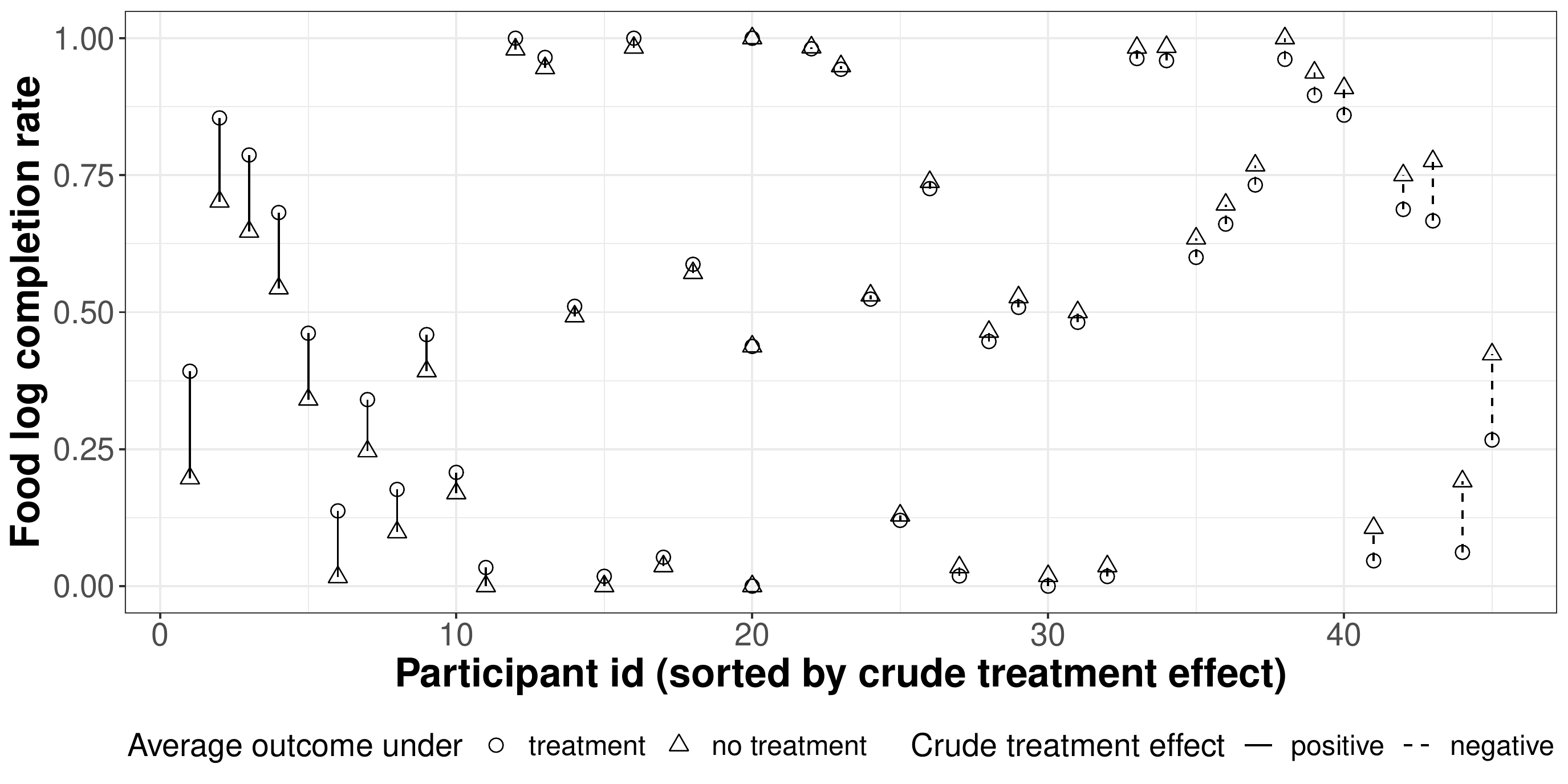}%
    \end{center}
    \small{* Circles (triangles) represent the food log completion rate of each individual, averaged over all days when they received treatment (no treatment). The length of the vertical line segments represents the magnitude of the difference between the circle and the triangle for an individual, which is an empirical (crude) estimate for the treatment effect for that individual. Solid (dashed) line segments represent positive (negative) treatment effect estimate. The x-axis is sorted so that the individual-specific empirical treatment effect estimate is decreasing with larger participant id.}
\end{figure}



\section{Proof of Theorem \ref{thm:efficient-score}}
\label{appen:proof-efficient-score}

\subsection{Overview}

In Section \ref{sec:proof-efficient-score-main}, we present the proof of Theorem \ref{thm:efficient-score} based on a general form of the efficient score using semiparametric efficiency theory developed in Section \ref{sec:main-lemmas}. In Section \ref{sec:asmp}, we give assumptions that characterize the semiparametric model, and we introduce  additional notation that will be used throughout the proof. In Section \ref{sec:main-lemmas}, we derive the general form of the efficient score using semiparametric efficiency theory. For ease of reading the proofs, the supporting technical lemmas that are used in deriving the general form of the efficient score are presented and proved in Section \ref{sec:supporting-lemmas}. For notation simplicity, this entire section is presented in the case where $I_t = 1$ for all $t$, and we omit the notation $I_t$ throughout.

The techniques in Section \ref{sec:main-lemmas} and Section \ref{sec:supporting-lemmas} follow mostly from Robin's derivation of the efficient score for structural nested mean models \citep{robins1994snmm}.

\subsection{Proof of Theorem \ref{thm:efficient-score} from a general form of efficient score}
\label{sec:proof-efficient-score-main}

We first present a useful lemma.

\begin{lemma} \label{lem:binary-action-centering}
Suppose that $B, C$ are two random variables, and that $B$ takes binary value $\{0,1\}$.
Suppose $E\{ S(B, C) \mid C \} = 0$ for some function $S(B, C)$. Then we have
\begin{align}
    S(B, C) = \{S(1, C) - S(0, C)\} \times \{ B - P(B = 1 \mid C) \}. \label{eq:lem:binary-action-centering}
\end{align}
\end{lemma}

\begin{proof}[of Lemma \ref{lem:binary-action-centering}]
Since $B$ takes binary value, we have
\begin{align}
    S(B, C) & = S(1, C) B + S(0, C) (1-B) \nonumber \\
    & = \{S(1, C) - S(0, C)\} B + S(0, C). \label{proofuse-semi-1001}
\end{align}
We also have
\begin{align}
    E\{ S(B, C) \mid C \} & = E\{ S(B, C) \mid C, B = 1 \} P(B = 1 \mid C) \nonumber \\
    & \phantom{=} + E\{ S(B, C) \mid C, B = 0 \} \times \{1 - P(B = 1 \mid C)\} \nonumber \\
    & = S(1, C) P(B = 1 \mid C) + S(0, C) \{1 - P(B = 1 \mid C)\} \nonumber \\
    & = \{ S(1, C) - S(0, C) \} P(B = 1 \mid C) + S(0, C). \label{proofuse-semi-1002}
\end{align}
Equation \eqref{proofuse-semi-1002} and $E\{ S(B, C) \mid C \} = 0$ imply
\begin{align}
    S(0, C) = - \{ S(1, C) - S(0, C) \} P(B = 1 \mid C). \label{proofuse-semi-1003}
\end{align}
Combining \eqref{proofuse-semi-1001} and \eqref{proofuse-semi-1003} yields \eqref{eq:lem:binary-action-centering}. This completes the proof.
\end{proof}

\begin{proof}[of Theorem \ref{thm:efficient-score}]
To connect Theorem \ref{thm:efficient-score} with the notation used in the rest of this section, let $\psi_0$ be the true value of the parameter $\psi$. Define $V_t = (H_t, A_t)$, $U_{t+1}(\psi) = Y_{t,1}e^{-A_t f(H_t)^T \psi}$, $\dot{U}_{t+1}(\psi) = U_{t+1}(\psi)-E\{U_{t+1}(\psi_{0})\mid H_t\}$, and $W_t = \var\{U_{t+1}(\psi_{0})\mid V_t\}^{-1}$.

By Lemma \ref{lem:EIF}, a general form of the efficient score is
\begin{align}
S_{\eff}(\psi_{0})=-\sum_{t=1}^{T}\rho_t \dot{U}_{t+1}(\psi_{0}), \label{proofuse-semi-1100.2}
\end{align}
where
\begin{align}
    \rho_t = \bigg[E\bigg\{\frac{\partial U_{t+1}(\psi_{0})}{\partial\psi}\mid V_t\bigg\}-E\bigg\{\frac{\partial U_{t+1}(\psi_{0})}{\partial\psi}W_t\mid H_t\bigg\} E(W_t\mid H_t)^{-1}\bigg]W_t.
\end{align}
Note that $E(\rho_t \mid H_t) = 0$; therefore, Lemma \ref{lem:binary-action-centering} implies
\begin{align}
    \rho_t = \{\rho_t(A_t = 1) - \rho_t(A_t = 0) \} \{A_t - p_t(H_t)\}, \label{proofuse-semi-1100.1}
\end{align}
where $\rho_t(A_t = a)$ denotes $\rho_t$ (as a function of $H_t$ and $A_t$) evaluated at $A_t = a$. In the following we calculate corresponding terms in the context of Theorem \ref{thm:efficient-score}.

First, we have
\begin{align}
\frac{\partial U_{t+1} (\psi_0)}{ \partial \psi} = - U_{t+1}(\psi_0) A_t f(H_t), \label{proofuse-semi-1100}
\end{align}
and hence
\begin{align}
    & E\bigg\{ \frac{\partial U_{t+1} (\psi_0)}{ \partial \psi} \bigg| H_t, A_t = 1 \bigg\} = - E \{ U_{t+1}(\psi_0) \mid H_t, A_t = 1 \} f(H_t) \nonumber \\
    = & - E\{ Y_{t,1}(\bA_{t-1},0) \mid H_t \} f(H_t) = - e^{\mu(H_t)} f(H_t), \label{proofuse-semi-1101} \\
    & E\bigg\{ \frac{\partial U_{t+1} (\psi_0)}{ \partial \psi} \bigg| H_t, A_t = 0 \bigg\} = 0. \label{proofuse-semi-1102}
\end{align}
where the second equality in \eqref{proofuse-semi-1101} follows from Lemma \ref{lem:aux-1}.

Second, we have
\begin{align}
    W_t & = \var\{U_{t+1}(\psi_{0})\mid V_t\}^{-1} = \var(Y_{t,1}\mid V_t)^{-1} e^{2 A_t f(H_t)^T \psi_0} \nonumber \\
        & = \Big[ e^{\mu(H_t) + A_t f(H_t)^T \psi_0} \{1 - e^{\mu(H_t) + A_t f(H_t)^T \psi_0}\} \Big]^{-1} e^{2 A_t f(H_t)^T \psi_0} \nonumber \\
        & = \frac{e^{A_t f(H_t)^T \psi_0}}{e^{\mu(H_t)} \{1 - e^{\mu(H_t) + A_t f(H_t)^T \psi_0}\}}, \label{proofuse-semi-1103}
\end{align}
\begin{align}
    W_t(A_t = 1) & = \frac{e^{f(H_t)^T \psi_0}}{e^{\mu(H_t)} \{1 - e^{\mu(H_t) + f(H_t)^T \psi_0}\}}, \label{proofuse-semi-1103.1} \\
    W_t(A_t = 0) & = \frac{1}{e^{\mu(H_t)} \{1 - e^{\mu(H_t)}\}}, \label{proofuse-semi-1103.2}
\end{align}
and
\begin{align}
    & E(W_t \mid H_t) = E(W_t \mid H_t, A_t = 1) p_t(H_t) + E(W_t \mid H_t, A_t = 0) \{1 - p_t(H_t)\} \nonumber \\
    & = \frac{e^{f(H_t)^T \psi_0}}{e^{\mu(H_t)} \{1 - e^{\mu(H_t) + f(H_t)^T \psi_0}\}} p_t(H_t)
    + \frac{1}{e^{\mu(H_t)} \{1 - e^{\mu(H_t)}\}} \{1 - p_t(H_t)\} \nonumber \\
    & = \frac{1}{e^{\mu(H_t)}} \times \frac{ \{ e^{f(H_t)^T \psi_0} - e^{\mu(H_t) + f(H_t)^T \psi_0} \} p_t(H_t) + \{1 - e^{\mu(H_t) + f(H_t)^T \psi_0} \} \{ 1 - p_t(H_t) \} }{ \{ 1 - e^{\mu(H_t)} \} \{ 1 - e^{\mu(H_t) + f(H_t)^T \psi_0} \} }. \label{proofuse-semi-1104}
\end{align}

Third, it follows from \eqref{proofuse-semi-1100} and \eqref{proofuse-semi-1103} that
\begin{align}
    & E \bigg\{ \frac{\partial U_t (\psi_0)}{ \partial \psi} W_t \bigg| H_t \bigg\} = - E\{ U_t(\psi_0) A_t f(H_t) W_t \mid H_t \}  \nonumber \\
    & = - E \{ U_t(\psi_0) W_t \mid H_t, A_t = 1\} p_t(H_t) f(H_t) \nonumber \\
    & = - E \bigg[ Y_{t,1} e^{-f(H_t)^T \psi_0} \times \frac{e^{f(H_t)^T \psi_0}}{e^{\mu(H_t)} \{1 - e^{\mu(H_t) + f(H_t)^T \psi_0}\}} \bigg| H_t, A_t = 1 \bigg] p_t(H_t) f(H_t) \nonumber \\
    & = - e^{\mu_t(H_t)} \frac{e^{f(H_t)^T \psi_0}}{e^{\mu(H_t)} \{1 - e^{\mu(H_t) + f(H_t)^T \psi_0}\}} p_t(H_t) f(H_t) \label{proofuse-semi-1105} \\
    & = - \frac{e^{f(H_t)^T \psi_0}}{1 - e^{\mu(H_t) + f(H_t)^T \psi_0}} p_t(H_t) f(H_t), \label{proofuse-semi-1106}
\end{align}
where \eqref{proofuse-semi-1105} follows from that $E[Y_{t,1} \exp\{-f(H_t)^T\psi_0\} \mid H_t, A_t] = E[Y_{t,1}(\bA_{t-1}, 0) \mid H_t]$, an implication of Lemma \ref{lem:aux-1}.

Because of \eqref{proofuse-semi-1102}, we have
\begin{align}
    \rho_t(A_t = 1) - \rho_t(A_t = 0) = & E\bigg\{ \frac{\partial U_{t+1} (\psi_0)}{ \partial \psi} \bigg| H_t, A_t = 1 \bigg\} W_t(A_t = 1) \nonumber \\
    & - E\bigg\{ \frac{\partial U_{t+1} (\psi_0)}{ \partial \psi} W_t \bigg| H_t \bigg\} E(W_t \mid H_t)^{-1} \{ W_t(A_t = 1) - W_t(A_t = 0) \}. \label{proofuse-semi-1108}
\end{align}
By \eqref{proofuse-semi-1103.1} and \eqref{proofuse-semi-1103.2} we have
\begin{align}
    W_t(A_t = 1) - W_t(A_t = 0) = \frac{e^{f(H_t)^T\psi_0} - 1}{e^{\mu(H_t)} \{ 1 - e^{\mu(H_t)} \} \{ 1 - e^{\mu(H_t) + f(H_t)^T \psi_0} \} }, 
\end{align}
which, combined with \eqref{proofuse-semi-1104}, yields
\begin{align}
    & E(W_t \mid H_t)^{-1} \{ W_t(A_t = 1) - W_t(A_t = 0) \} \nonumber \\
    & = \frac{e^{f(H_t)^T \psi_0} - 1}{\{ e^{f(H_t)^T \psi_0} - e^{\mu(H_t) + f(H_t)^T \psi_0} \} p_t(H_t) + \{1 - e^{\mu(H_t) + f(H_t)^T \psi_0} \} \{ 1 - p_t(H_t) \}}. \label{proofuse-semi-1109}
\end{align}
Plugging \eqref{proofuse-semi-1101}, \eqref{proofuse-semi-1103.1}, and \eqref{proofuse-semi-1109} into \eqref{proofuse-semi-1108} yields
\begin{align}
    & \rho_t(A_t = 1) - \rho_t(A_t = 0) \nonumber \\
    & = - \frac{e^{f(H_t)^T \psi_0}f(H_t)}{\{ e^{f(H_t)^T \psi_0} - e^{\mu(H_t) + f(H_t)^T \psi_0} \} p_t(H_t) + \{1 - e^{\mu(H_t) + f(H_t)^T \psi_0} \} \{ 1 - p_t(H_t) \}}. \label{proofuse-semi-1110}
\end{align}
Therefore, by plugging \eqref{proofuse-semi-1110} into \eqref{proofuse-semi-1100.1}, we have
\begin{align}
    \rho_t = - K_t \{A_t - p_t(H_t)\} f(H_t). \label{proofuse-semi-1120}
\end{align}

On the other hand, by Lemma \ref{lem:aux-1} we have
\begin{align}
    \dot{U}_{t+1}(\psi_0) = U_{t+1}(\psi_0)-E\{U_{t+1}(\psi_{0})\mid H_t\} = e^{-A_t f(H_t)^T \psi_0} Y_{t,1} - e^{\mu(H_t)}. \label{proofuse-semi-1121}
\end{align}
Plugging \eqref{proofuse-semi-1120} and \eqref{proofuse-semi-1121} into \eqref{proofuse-semi-1100.2} gives the form of $S_{\eff}(\psi_{0})$. This completes the proof.


\end{proof}

\subsection{Assumption and Additional Notation} \label{sec:asmp}

In deriving the semiparametric efficient score, we consider the semiparametric model characterized by the following assumptions:

\begin{assumption}\label{assump:additional-semi-2}
For all $1\leq t\leq T$, $E\{Y_{t,1}(\bar{A}_{t-1},0)\mid H_t,A_t\}=E\{Y_{t,1}(\bar{A}_{t-1},0)\mid H_t\}.$
\end{assumption}

\begin{assumption}\label{assump:additional-semi-1}
Assume that there exists a function $\gamma()$ and a true parameter
value $\psi_{0}\in\mathbb{R}^{p}$, such that for any $1\leq t\leq T$,
\begin{equation}
\log\frac{E\{Y_{t,1}(\bar{a}_t)\mid\bar{z}_t,\bar{a}_t\}}{E\{Y_{t,1}(\bar{a}_{t-1},0)\mid\bar{z}_t,\bar{a}_t\}}=\gamma(t+1,\bar{z}_t,\bar{a}_t;\psi_{0}).\label{eq:assump-model}
\end{equation}
\end{assumption}

In the following, we present additional notation that will be used in the proof. Each will be defined as they appear in the proof. Here we gather the definition of all the terms for ease of reading.

\begin{itemize}
\item The longitudinal data is $L_{1},A_{1},Y_{1,1},L_{2},A_{2},Y_{2,1},\ldots,L_{T},A_{T},Y_{T,1}$, where $L_t$ is a time-varying covariate, $A_t$ is the treatment assignment, and $Y_{t,1}$ is the proximal outcome
\item $Y_{0,1}=\emptyset$, $L_{T+1}=\emptyset$, $A_{T+1}=\emptyset$
\item $Z_{t}=(Y_{t-1,1},L_{t})$
\item $H_{t}=(\bar{A}_{t-1},\bar{Z}_{t})$
\item $V_{t}=(\bar{A}_{t},\bar{Z}_{t})=(H_{t},A_{t})$
\item $U_{t+1}(\psi)=Y_{t,1}\exp\{-\gamma(t+1,\bar{z}_{t},\bar{a}_{t};\psi)\}$
\item $\dot{U}_{t+1}(\psi)=U_{t+1}(\psi)-E\{U_{t+1}(\psi_{0})\mid H_{t}\}$

\smallskip

\item $Q_{t}=E\{U_{t+1}(\psi_{0})\mid V_{t}\}-E\{U_{t+1}(\psi_{0})\mid V_{t-1}\}$
\item $S_{t}=\partial\log f(\sigma_{t+1}\mid V_{t}) / \partial\sigma_{t+1}$
\item $W_{t}=\text{Var}(\sigma_{t+1}\mid V_{t})^{-1}$, which will be shown
to be equal to $\text{Var}\{U_{t+1}(\psi_{0})\mid V_{t}\}^{-1}$
\item $T_{t}=E(W_{t}\mid H_{t})$
\item $T_{t}^{\bullet}=E(T_{t}^{-1}\mid V_{t-1})$
\item $\epsilon_{t}=T_{t}^{-1}W_{t}\sigma_{t+1}+Q_{t}$
\item $W_{t,t-1}=\text{Var}(\epsilon_{t}\mid V_{t-1})^{-1}$

\smallskip

\item $\mathcal{H}$: the Hilbert space of all functions of $V_{T+1}$ that have mean zero finite variance.
\item $\Lambda_{t}^{1}=\{A_{t}^{1}=a_{t}^{1}(V_{T+1}):E(A_{t}^{1}\mid V_{t},Y_{t,1})=0\}$
\item $\Lambda_{t}^{2}=\{A_{t}^{2}=a_{t}^{2}(\sigma_{t+1},V_{t}):E(A_{t}^{2}\mid V_{t})=0,E(A_{t}^{2}\sigma_{t+1}\mid V_{t})=0\}$
\item $\Lambda_{t}^{3}=\sum_{m=1}^{t}\Gamma_{m}^{3}$
\item $\Gamma_{m}^{3}=\{A_{m}^{3}=a_{m}^{3}(V_{m}):E(A_{m}^{3}\mid H_{m})=0\}$
\item $\Lambda_{t}^{4}=\Gamma_{t}^{4}+\sum_{m=1}^{t-1}\Lambda_{m}^{\bullet}$
\item $\Lambda_{m}^{\bullet}=\{A_{m}^{\bullet}=a_{m}^{\bullet}(H_{m}):E(A_{m}^{\bullet}\mid V_{m-1})=0\}$
\item $\Gamma_{t}^{4}=\{A_{t}^{\bullet}+S_{t}E(Q_{t}A_{t}^{\bullet}\mid V_{t-1}):A_{t}^{\bullet}\in\Lambda_{t}^{\bullet}\}$
\item $\tilde{\Gamma}_{t}^{4}=\{A_{t}^{\bullet}-E(Q_{t}A_{t}^{\bullet}\mid V_{t-1})(T_{t}^{\bullet})^{-1}T_{t}^{-1}W_{t}\sigma_{t+1}:A_{t}^{\bullet}\in\Lambda_{t}^{\bullet}\}$
\item $\Lambda_{t}^{5}=\{S_{t}A_{t}^{\bullet}:A_{t}^{\bullet}\in\Lambda_{t}^{\bullet}\}$
\item $\tilde{\Lambda}_{t}^{5}=\{A_{t}^{\bullet}W_{t}\sigma_{t+1}:A_{t}^{\bullet}\in\Lambda_{t}^{\bullet}\}$
\item $\Lambda_{t}^{6}=\{a(V_{t-1})S_{t}:a(V_{t-1})\text{ is any function}\in\mathbb{R}^{p}\}$
\item $\tilde{\Lambda}_{t}^{6}=\{a(V_{t-1})\epsilon_{t}:a(V_{t-1})\text{ is any function}\in\mathbb{R}^{p}\}$

\smallskip

\item $D_{t}=E\{h(\sigma_{t+1},V_{t})W_{t}\sigma_{t+1}\mid H_{t}\}$, for
a given $h(\sigma_{t+1},V_{t})\in\mathcal{H}$
\item $R_{t}=E(B\sigma_{t+1}\mid V_{t})$ and $R_{t-1}=E(R_{t}W_{t}T_{t}^{-1}\mid V_{t-1})$,
for a given $B=b(V_{T+1})\in\mathcal{H}$
\end{itemize}

\subsection{Derivation of the general form of the efficient score} \label{sec:main-lemmas}

\begin{lemma}
\label{lem:semimodel-Mt}Let $\mathcal{M}$ denote the semiparametric model defined by consistency (Assumption \ref{assumption:consistency}), positivity (Assumption \ref{assumption:positivity}), (weak) sequential ignorability (Assumption \ref{assump:additional-semi-2}),
and Assumption \ref{assump:additional-semi-1}. Let $\mathcal{M}_t$ denote the semiparametric model defined by consistency, positivity, and the following $t$-specific
version of (weak) sequential ignorability and  (\ref{eq:assump-model}):
for a fixed $t$,
\begin{align*}
\log\frac{E\{Y_{t,1}(\bar{a}_t)\mid\bar{z}_t,\bar{a}_t\}}{E\{Y_{t,1}(\bar{a}_{t-1},0)\mid\bar{z}_t,\bar{a}_t\}} & =\gamma(t+1,\bar{z}_t,\bar{a}_t;\psi_{0}),\\
E\{Y_{t,1}(\bar{A}_{t-1},0)\mid H_t,A_t\} & =E\{Y_{t,1}(\bar{A}_{t-1},0)\mid H_t\}.
\end{align*}
Let $\Lambda$ and $\Lambda_t$ be the nuisance tangent space for
model $\mathcal{M}$ and model $\mathcal{M}_t$, respectively. Then
we have ${\cal M}=\bigcap_{t=1}^{T}\mathcal{M}_t$ and $\Lambda=\bigcap_{t=1}^{T}\Lambda_t$.
\end{lemma}

\begin{proof}
This follows directly from the definition of nuisance tangent space
(i.e., $L^{2}$-closure of all parametric submodel nuisance scores).
\end{proof}
\begin{lemma}
\label{lem:nuisance-space-decomp}The nuisance tangent space for model
$\mathcal{M}_t$ is $\Lambda_t=\Lambda_t^{1}+\Lambda_t^{2}+\Lambda_t^{3}+\Lambda_t^{4}+\Lambda_t^{5}+\Lambda_t^{6}$,
where
\begin{align*}
\Lambda_t^{1} & =\{A_t^{1}=a_t^{1}(V_{T+1}):E(A_t^{1}\mid V_t,Y_{t,1})=0\},\\
\Lambda_t^{2} & =\{A_t^{2}=a_t^{2}(\sigma_{t+1},V_t):E(A_t^{2}\mid V_t)=0,E(A_t^{2}\sigma_{t+1}\mid V_t)=0\},\\
\Lambda_t^{3} & =\sum_{m=1}^t\Gamma_{m}^{3},\\
\Lambda_t^{4} & =\Gamma_t^{4}+\sum_{m=1}^{t-1}\Lambda_{m}^{\bullet},\\
\Lambda_t^{5} & =\{S_tA_t^{\bullet}:A_t^{\bullet}\in\Lambda_t^{\bullet}\},\\
\Lambda_t^{6} & =\{a(V_{t-1})S_t:a(V_{t-1})\text{ is any function}\in\mathbb{R}^{p}\},
\end{align*}
where
\begin{align*}
\Gamma_{m}^{3} & =\{A_{m}^{3}=a_{m}^{3}(V_{m}):E(A_{m}^{3}\mid H_{m})=0\},\\
\Lambda_{m}^{\bullet} & =\{A_{m}^{\bullet}=a_{m}^{\bullet}(H_{m}):E(A_{m}^{\bullet}\mid V_{m-1})=0\}\\
\Gamma_t^{4} & =\{A_t^{\bullet}+S_tE(Q_tA_t^{\bullet}\mid V_{t-1}):A_t^{\bullet}\in\Lambda_t^{\bullet}\},
\end{align*}
and
\begin{align*}
Q_t & =E\{U_{t+1}(\psi_{0})\mid V_t\}-E\{U_{t+1}(\psi_{0})\mid V_{t-1}\},\\
S_t & =\frac{\partial\log f(\sigma_{t+1}\mid V_t)}{\partial\sigma_{t+1}}.
\end{align*}
Both $Q_t$ and $S_t$ are evaluated at the truth.
\end{lemma}

\begin{proof}
The likelihood for model $\mathcal{M}_t$ is
\begin{align}
L(\psi,\theta) & =f(V_{T+1}\mid V_t,Y_{t,1})f(Y_{t,1}\mid V_t)\prod_{m=1}^t\{f(A_{m}\mid H_{m})f(Z_{m}\mid V_{m-1})\}\nonumber \\
 & =f(V_{T+1}\mid V_t,Y_{t,1};\theta_{1})\nonumber \\
 & \times\frac{\partial\sigma_{t+1}}{\partial Y_{t,1}}\times f(\sigma_{t+1}(\psi,\theta_{4},\theta_{5},\theta_{6})\mid V_t;\theta_{2})\nonumber \\
 & \times\prod_{m=1}^t\{f(A_{m}\mid H_{m};\theta_{3})f(Z_{m}\mid V_{m-1};\theta_{4})\},\label{eq:nuisance-proofuse-1}
\end{align}
where
\begin{align}
\sigma_{t+1}(\psi,\theta_{4},\theta_{5},\theta_{6}) & =Y_{t,1}e^{-\gamma(t+1,H_t,A_t;\psi)}-\beta_t(V_{t-1};\theta_{6})\nonumber \\
 & -\left\{ q_t^{*}(H_t;\theta_{5})-\int q_t^{*}(z_t,V_{t-1};\theta_{5})dF(z_t\mid V_{t-1};\theta_{4})\right\} .\label{eq:nuisance-proofuse-2}
\end{align}
Here, $\theta=(\theta_{1},\theta_{2},\theta_{3},\theta_{4},\theta_{5},\theta_{6})$,
each is an infinite-dimensional nuisance parameter and are variationally
independent of each other. The second equality in (\ref{eq:nuisance-proofuse-1})
follows from the change of variables
\[
(L_{1},A_{1},Y_{1,1},\ldots,Y_{t,1},\ldots,L_{T},A_{T},Y_{t,1})\to(L_{1},A_{1},Y_{2},\ldots,\sigma_{t+1},\ldots,L_{T},A_{T},Y_{T,1})
\]
which has Jacobian $\partial\sigma_{t+1} / \partial Y_{t,1} = e^{-\gamma(t+1,H_t,A_t;\psi)}$.
By Lemma \ref{lem:aux-2-change-of-variable}, the constraints on model
$\mathcal{M}_t$ is equivalent to $E(\sigma_{t+1}\mid H_t,A_t)=0$,
i.e., $\int tdF(t\mid H_t,A_t)=0$.
There is no restrictions on $q_t^{*}(H_t;\theta_{5})$ and $\beta_t(H_{t-1},A_{t-1};\theta_{6})$.
The constraint $E(q_t\mid H_{t-1},A_{t-1})=0$ has been incorporated
because $q_t^{*}$ is centered in (\ref{eq:nuisance-proofuse-2}).

Below we derive the nuisance tangent space for each nuisance parameter
$(\theta_{1},\theta_{2},\theta_{3},\theta_{4},\theta_{5},\theta_{6})$.

\textbf{Nuisance tangent space $\Lambda_t^{1}$ for $\theta_{1}$.}
This follows from Theorem 4.6 in \citet{tsiatis2007semiparametric}.

\textbf{Nuisance tangent space $\Lambda_t^{2}$ for $\theta_{2}$.}
This follows from Theorem 4.7 in \citet{tsiatis2007semiparametric}.

\textbf{Nuisance tangent space $\Lambda_t^{3}$ for $\theta_{3}$.}
This follows from Theorem 4.6 in \citet{tsiatis2007semiparametric}.

\textbf{Nuisance tangent space $\Lambda_t^{4}$ for $\theta_{4}$.}
The score for $\theta_{4}$ equals
\[
\frac{\partial\log L(\psi,\theta)}{\partial\theta_{4}}=\frac{\partial\log f(\sigma_{t+1}(\psi,\theta_{4},\theta_{5},\theta_{6})\mid V_t;\theta_{2})}{\partial\theta_{4}}+\sum_{m=1}^t\frac{\partial\log f(Z_{m}\mid V_{m-1};\theta_{4})}{\partial\theta_{4}}.
\]
The $\frac{\partial\log f(\sigma_{t+1}(\psi,\theta_{4},\theta_{5},\theta_{6})\mid V_t;\theta_{2})}{\partial\theta_{4}}+\frac{\partial\log f(Z_t\mid V_{t-1};\theta_{4})}{\partial\theta_{4}}$
part correspond to $\Gamma_t^{4}$, which is shown in the proof
of Theorem A4.1 in \citet{robins1994snmm}. The $\sum_{m=1}^{t-1}\frac{\partial\log f(Z_{m}\mid V_{m-1};\theta_{4})}{\partial\theta_{4}}$
part correspond to $\sum_{m=1}^{t-1}\Lambda_{m}^{\bullet}$, which
follows from Theorem 4.6 in \citet{tsiatis2007semiparametric}.

\textbf{Nuisance tangent space $\Lambda_t^{5}$ for $\theta_{5}$.
}The score for $\theta_{5}$ equals
\[
\frac{\partial\log L(\psi,\theta)}{\partial\theta_{5}}=\frac{\partial\log f(\sigma_{t+1}(\psi,\theta_{4},\theta_{5},\theta_{6})\mid V_t;\theta_{2})}{\partial\theta_{5}}.
\]
The form of $\Lambda_t^{5}$ is derived in the proof of Theorem
A4.1 in \citet{robins1994snmm}.

\textbf{Nuisance tangent space $\Lambda_t^{6}$ for $\theta_{6}$.
}The score for $\theta_{6}$ equals
\begin{align*}
\frac{\partial\log L(\psi,\theta)}{\partial\theta_{6}} & =\frac{\partial\log f(\sigma_{t+1}(\psi,\theta_{4},\theta_{5},\theta_{6})\mid V_t;\theta_{2})}{\partial\theta_{6}}\\
 & =-\frac{\partial\log f(\sigma_{t+1}(\psi,\theta_{4},\theta_{5},\theta_{6})\mid V_t;\theta_{2})}{\partial\sigma_{t+1}}\times\frac{\partial\beta_t(V_{t-1};\theta_{6})}{\partial\theta_{6}}.
\end{align*}
Because there is no restriction on $\beta_t(V_{t-1};\theta_{6})$,
$\Lambda_t^{6}=\{a(V_{t-1})S_t:a(V_{t-1})\text{ is any function}\in\mathbb{R}^{p}\}$.
\end{proof}
\begin{lemma}
\label{lem:nuisance-space-orthogonalization}The nuisance tangent
space $\Lambda_t$ in Lemma \ref{lem:nuisance-space-decomp} equals
the direct sum of the following spaces: 
\[
\Lambda_t=\Lambda_t^{1}\oplus\Lambda_t^{2}\oplus\bigoplus_{m=1}^t\Gamma_{m}^{3}\oplus\tilde{\Gamma}_t^{4}\oplus\bigoplus_{m=1}^{t-1}\Lambda_{m}^{\bullet}\oplus\tilde{\Lambda}_t^{5}\oplus\tilde{\Lambda}_t^{6},
\]
where
\begin{align*}
\tilde{\Gamma}_t^{4} & =\{A_t^{\bullet}-E(Q_tA_t^{\bullet}\mid V_{t-1})(T_t^{\bullet})^{-1}T_t^{-1}W_t\sigma_{t+1}:A_t^{\bullet}\in\Lambda_t^{\bullet}\},\\
\tilde{\Lambda}_t^{5} & =\{A_t^{\bullet}W_t\sigma_{t+1}:A_t^{\bullet}\in\Lambda_t^{\bullet}\},\\
\tilde{\Lambda}_t^{6} & =\{a(V_{t-1})\epsilon_t:a(V_{t-1})\text{ is any function}\in\mathbb{R}^{p}\},
\end{align*}
and
\begin{align*}
W_t & =\var(\sigma_{t+1}\mid V_t)^{-1},\\
T_t & =E(W_t\mid H_t),\\
T_t^{\bullet} & =E(T_t^{-1}\mid V_{t-1}),\\
\epsilon_t & =T_t^{-1}W_t\sigma_{t+1}+Q_t.
\end{align*}
\end{lemma}

\begin{proof}
In Lemma \ref{lem:aux-3-orthogonal-Lambdas} we show that $\Lambda_t^{1}$,
$\{\Gamma_{m}^{3}\}_{1\le m\leq t}$, $\{\Lambda_{m}^{\bullet}\}_{1\leq m\leq t-1}$
are orthogonal to each other and orthogonal to the rest subspaces
$\Lambda_t^{2},\Gamma_t^{4},\Lambda_t^{5},\Lambda_t^{6}$.
Thus, it suffices to show that $\tilde{\Lambda}_t^{5}=\Pi(\Lambda_t^{5}\mid\Lambda_t^{2,\perp})$,
$\tilde{\Gamma}_t^{4}=\Pi\{\Gamma_t^{4}\mid(\Lambda_t^{2}\oplus\tilde{\Lambda}_t^{5})^{\perp}\}$,
and $\tilde{\Lambda}_t^{6}=\Pi\{\Lambda_t^{6}\mid(\Lambda_t^{2}\oplus\tilde{\Lambda}_t^{5}\oplus\tilde{\Gamma}_t^{4})^{\perp}\}$.

First, we show that $\tilde{\Lambda}_t^{5}=\Pi(\Lambda_t^{5}\mid\Lambda_t^{2,\perp})$.
For any $S_tA_t^{\bullet}\in\Lambda_t^{5}$, because $S_tA_t^{\bullet}$
is a function of $(\sigma_{t+1},V_t)$, Lemma \ref{lem:aux-4-projection-1}
implies that 
\begin{equation}
\Pi(S_tA_t^{\bullet}\mid\Lambda_t^{2})=S_tA_t^{\bullet}-E(S_tA_t^{\bullet}\sigma_{t+1}\mid V_t)\var(\sigma_{t+1}\mid V_t)^{-1}\sigma_{t+1}-E(S_tA_t^{\bullet}\mid V_t).\label{eq:nuisance-orthogonalization-proofuse-1}
\end{equation}
By Lemma \ref{lem:aux-5-integrate-S} we have $E(S_tA_t^{\bullet}\sigma_{t+1}\mid V_t)=A_t^{\bullet}E(S_t\sigma_{t+1}\mid V_t)=-A_t^{\bullet}$
and $E(S_tA_t^{\bullet}\mid V_t)=A_t^{\bullet}E(S_t\mid V_t)=0$,
so (\ref{eq:nuisance-orthogonalization-proofuse-1}) implies
\begin{align*}
\Pi(S_tA_t^{\bullet}\mid\Lambda_t^{2,\perp}) & =S_tA_t^{\bullet}-\Pi(S_tA_t^{\bullet}\mid\Lambda_t^{2})=-A_t^{\bullet}W_t\sigma_{t+1}.
\end{align*}
This gives the form of $\tilde{\Lambda}_t^{5}$.

Second, we show that $\tilde{\Gamma}_t^{4}=\Pi\{\Gamma_t^{4}\mid(\Lambda_t^{2}\oplus\tilde{\Lambda}_t^{5})^{\perp}\}$.
For any $A_t^{\bullet}+S_tE(Q_tA_t^{\bullet}\mid V_{t-1})\equiv g_{1}(\sigma_{t+1},V_t)\in\Gamma_t^{4}$
where $A_t^{\bullet}(H_t)$ satisfies $E(A_t^{\bullet}\mid V_{t-1})=0$.
By Lemma \ref{lem:aux-hilbert-space-sequential-projection}, it suffices
to derive $\Pi\{\Pi(g_{1}\mid\Lambda_t^{2,\perp})\mid\tilde{\Lambda}_t^{5,\perp}\}$.
By Lemma \ref{lem:aux-4-projection-1} we have
\begin{equation}
\Pi\{g_{1}(\sigma_{t+1},V_t)\mid\Lambda_t^{2}\}=g_{1}-E(g_{1}\sigma_{t+1}\mid V_t)\var(\sigma_{t+1}\mid V_t)^{-1}\sigma_{t+1}-E(g_{1}\mid V_t).\label{eq:nuisance-orthogonalization-proofuse-2}
\end{equation}
By Lemma \ref{lem:aux-5-integrate-S} we have $E(g_{1}\mid V_t)=A_t^{\bullet}+E(S_t\mid V_t)E(Q_tA_t^{\bullet}\mid V_{t-1})=A_t^{\bullet}$
and
\[
E(g_{1}\sigma_{t+1}\mid V_t)=A_t^{\bullet}E(\sigma_{t+1}\mid V_t)+E(S_t\sigma_{t+1}\mid V_t)E(Q_tA_t^{\bullet}\mid V_{t-1})=-E(Q_tA_t^{\bullet}\mid V_{t-1}).
\]
These combining with (\ref{eq:nuisance-orthogonalization-proofuse-2})
yields that
\begin{align*}
\Pi\{g_{1}(\sigma_{t+1},V_t)\mid\Lambda_t^{2,\perp}\} & =g_{1}(\sigma_{t+1},V_t)-\Pi\{g_{1}(\sigma_{t+1},V_t)\mid\Lambda_t^{2}\}\\
 & =A_t^{\bullet}-E(Q_tA_t^{\bullet}\mid V_{t-1})W_t\sigma_{t+1}.
\end{align*}
Now, let $g_{2}(\sigma_{t+1},V_t)=A_t^{\bullet}-E(Q_tA_t^{\bullet}\mid V_{t-1})W_t\sigma_{t+1}$.
By Lemma \ref{lem:aux-6-Lambda-tilde-5} we have
\begin{equation}
\Pi\{g_{2}(\sigma_{t+1},V_t)\mid\tilde{\Lambda}_t^{5}\}=\{-E(D_tT_t^{-1}\mid V_{t-1})(T_t^{\bullet}T_t)^{-1}+D_tT_t^{-1}\}W_t\sigma_{t+1},\label{eq:nuisance-orthogonalization-proofuse-3}
\end{equation}
where
\begin{align}
D_t & =E\{g_{2}(\sigma_{t+1},V_t)W_t\sigma_{t+1}\mid H_t\}\nonumber \\
 & =E(A_t^{\bullet}W_t\sigma_{t+1}\mid H_t)-E\{E(Q_tA_t^{\bullet}\mid V_{t-1})W_t^{2}\sigma_{t+1}^{2}\mid H_t\}\nonumber \\
 & =-E(Q_tA_t^{\bullet}\mid V_{t-1})T_t.\label{eq:nuisance-orthogonalization-proofuse-4}
\end{align}
The third equality in (\ref{eq:nuisance-orthogonalization-proofuse-4})
follows from Lemma \ref{lem:aux-comp-useful-result} and the fact
that $E(A_t^{\bullet}W_t\sigma_{t+1}\mid H_t)=E\{A_t^{\bullet}W_tE(\sigma_{t+1}\mid V_t)\mid H_t\}=0$.
Thus, plugging (\ref{eq:nuisance-orthogonalization-proofuse-4}) into
(\ref{eq:nuisance-orthogonalization-proofuse-3}) and we have
\[
\Pi\{g_{2}(\sigma_{t+1},V_t)\mid\tilde{\Lambda}_t^{5}\}=\{E(Q_tA_t^{\bullet}\mid V_{t-1})(T_t^{\bullet}T_t)^{-1}-E(Q_tA_t^{\bullet}\mid V_{t-1})\}W_t\sigma_{t+1}
\]
and
\begin{align*}
\Pi\{g_{2}(\sigma_{t+1},V_t)\mid\tilde{\Lambda}_t^{5,\perp}\} & =g_{2}(\sigma_{t+1},V_t)-\Pi\{g_{2}(\sigma_{t+1},V_t)\mid\tilde{\Lambda}_t^{5}\}\\
 & =A_t^{\bullet}-E(Q_tA_t^{\bullet}\mid V_{t-1})(T_t^{\bullet}T_t)^{-1}W_t\sigma_{t+1}.
\end{align*}
This gives the form of $\tilde{\Gamma}_t^{4}$.

Last, we show that $\tilde{\Lambda}_t^{6}=\Pi\{\Lambda_t^{6}\mid(\Lambda_t^{2}\oplus\tilde{\Lambda}_t^{5}\oplus\tilde{\Gamma}_t^{4})^{\perp}\}$.
For any $a(V_{t-1})S_t\in\Lambda_t^{6}$, by Lemma \ref{lem:aux-hilbert-space-sequential-projection},
it suffices to derive $\Pi(\Pi[\Pi\{a(V_{t-1})S_t\mid\Lambda_t^{2,\perp}\}\mid\tilde{\Lambda}_t^{5,\perp}]\mid\tilde{\Gamma}_t^{4,\perp})$.
By Lemma \ref{lem:aux-4-projection-1} we have 
\begin{align*}
\Pi\{a(V_{t-1})S_t\mid\Lambda_t^{2}\} & =a(V_{t-1})S_t-E\{a(V_{t-1})S_t\sigma_{t+1}\mid V_t\}W_t\sigma_{t+1}-E\{a(V_{t-1})S_t\mid V_t\}\\
 & =a(V_{t-1})S_t+a(V_{t-1})W_t\sigma_{t+1},
\end{align*}
where the second equality follows from Lemma \ref{lem:aux-5-integrate-S}.
Thus $\Pi\{a(V_{t-1})S_t\mid\Lambda_t^{2,\perp}\}=-a(V_{t-1})W_t\sigma_{t+1}\equiv g_{3}(V_t,\sigma_{t+1})$.
By Lemma \ref{lem:aux-6-Lambda-tilde-5} we have 
\[
\Pi\{g_{3}(\sigma_{t+1},V_t)\mid\tilde{\Lambda}_t^{5}\}=\{-E(D_t^{(2)}T_t^{-1}\mid V_{t-1})(T_t^{\bullet}T_t)^{-1}+D_t^{(2)}T_t^{-1}\}W_t\sigma_{t+1},
\]
where, by using Lemma \ref{lem:aux-comp-useful-result}, we have 
\begin{align*}
D_t^{(2)} & =E\{g_{3}(\sigma_{t+1},V_t)W_t\sigma_{t+1}\mid H_t\}\\
 & =-a(V_{t-1})E(W_t^{2}\sigma_{t+1}^{2}\mid H_t)\\
 & =-a(V_{t-1})T_t.
\end{align*}
Thus, 
\[
\Pi\{g_{3}(\sigma_{t+1},V_t)\mid\tilde{\Lambda}_t^{5}\}=\{a(V_{t-1})(T_t^{\bullet}T_t)^{-1}-a(V_{t-1})\}W_t\sigma_{t+1}
\]
and
\[
\Pi\{g_{3}(\sigma_{t+1},V_t)\mid\tilde{\Lambda}_t^{5,\perp}\}=-a(V_{t-1})(T_t^{\bullet}T_t)^{-1}W_t\sigma_{t+1}\equiv g_{4}(\sigma_{t+1},V_t).
\]
Note that $E(g_{4}\mid H_t)=E(g_{4}\mid V_{t-1})=0$, so for $O_{3}^{*}$
and $O_{4}^{*}$ defined in Lemma \ref{lem:aux-6-Lambda-tilde-5}
we have
\begin{align*}
O_{3}^{*}O_{4}^{*}(g_{4}) & =-E\{g(T_t^{\bullet})^{-1}T_t^{-1}W_t\sigma_{t+1}\mid V_{t-1}\}Q_t\\
 & =E\{a(V_{t-1})(T_t^{\bullet})^{-2}T_t^{-2}W_t^{2}\sigma_{t+1}^{2}\mid V_{t-1}\}Q_t\\
 & =a(V_{t-1})(T_t^{\bullet})^{-2}E\{T_t^{-2}E(W_t^{2}\sigma_{t+1}^{2}\mid H_t)\mid V_{t-1}\}Q_t\\
 & =a(V_{t-1})(T_t^{\bullet})^{-2}E(T_t^{-1}\mid V_{t-1})Q_t\\
 & =a(V_{t-1})(T_t^{\bullet})^{-1}Q_t,
\end{align*}
where the third to last equality follows from Lemma \ref{lem:aux-comp-useful-result}.
So by Lemma \ref{lem:aux-6-Lambda-tilde-5} we have
\[
\Pi\{g_{4}(\sigma_{t+1},V_t)\mid\tilde{\Gamma}_t^{4}\}=a(V_{t-1})(T_t^{\bullet})^{-1}Q_t-a(V_{t-1})(T_t^{\bullet})^{-1}E(Q_t^{2}\mid V_{t-1})W_{t,t-1}\epsilon_t,
\]
where $W_{t,t-1}=\var(\epsilon_t\mid V_{t-1})^{-1}$. Hence,
\begin{align*}
 & \Pi\{g_{4}(\sigma_{t+1},V_t)\mid\tilde{\Gamma}_t^{4,\perp}\}\\
= & g_{4}(\sigma_{t+1},V_t)-\Pi\{g_{4}(\sigma_{t+1},V_t)\mid\tilde{\Gamma}_t^{4}\}\\
= & -a(V_{t-1})(T_t^{\bullet})^{-1}(T_t^{-1}W_t\sigma_{t+1}+Q_t)+a(V_{t-1})(T_t^{\bullet})^{-1}E(Q_t^{2}\mid V_{t-1})W_{t,t-1}\epsilon_t\\
= & a(V_{t-1})(T_t^{\bullet})^{-1}\epsilon_t\{W_{t,t-1}E(Q_t^{2}\mid V_{t-1})-1\}\\
= & -a(V_{t-1})W_{t,t-1}\epsilon_t,
\end{align*}
where the last equality follows from Lemma \ref{lem:aux-comp-useful-result}.
Thus, 
\begin{align*}
\tilde{\Lambda}_t^{6} & =\Pi\{\Lambda_t^{6}\mid(\Lambda_t^{2}\oplus\tilde{\Lambda}_t^{5}\oplus\tilde{\Gamma}_t^{4})^{\perp}\}\\
 & =\{-a(V_{t-1})W_{t,t-1}\epsilon_t:a(V_{t-1})\text{ is any function }\in\mathbb{R}^{p}\}\\
 & =\{a(V_{t-1})\epsilon_t:a(V_{t-1})\text{ is any function }\in\mathbb{R}^{p}\},
\end{align*}
where the last equality follows from the fact that $W_{t,t-1}$ is
a function of $V_{t-1}$.

This completes the proof.
\end{proof}
\begin{lemma}
\label{lem:projection-Lambda-t}For any $B=b(V_{T+1})\in\mathcal{H}$,
its projection onto $\Lambda_t^{\perp}$ is
\[
\Pi(B\mid\Lambda_t^{\perp})=\{R_t-T_t^{-1}E(R_tW_t\mid H_t)\}W_t\sigma_{t+1},
\]
where $R_t=E(B\sigma_{t+1}\mid V_t)$, and $W_t,T_t$ are
defined in Lemma \ref{lem:nuisance-space-orthogonalization}.
\end{lemma}

\begin{proof}
For any $B=b(V_{T+1})\in\mathcal{H}$, we have
\begin{align*}
B & =\{B-E(B\mid\sigma_{t+1},V_t)\}+\{E(B\mid\sigma_{t+1},V_t)-E(B\mid V_t)\}\\
 & ~~+\sum_{m=1}^t\{E(B\mid V_{m})-E(B\mid H_{m})\}+\sum_{m=1}^t\{E(B\mid H_{m})-E(B\mid V_{m-1})\}.
\end{align*}
Note that $B-E(B\mid\sigma_{t+1},V_t)\in\Lambda_t^{1}$, and for
all, $1\leq m\leq t$ $E(B\mid V_{m})-E(B\mid H_{m})\in\Gamma_{m}^{3}$
and $E(B\mid H_{m})-E(B\mid V_{m-1})\in\Lambda_{m}^{\bullet}$. Hence,
by Lemma \ref{lem:nuisance-space-orthogonalization} we have
\begin{equation}
\Pi(B\mid\Lambda_t^{\perp})=\Pi\{E(B\mid\sigma_{t+1},V_t)-E(B\mid V_t)\mid\Lambda_t^{\perp}\}+\Pi\{E(B\mid H_t)-E(B\mid V_{t-1})\mid\Lambda_t^{\perp}\}.\label{eq:lem-projection-Lambda-t-proofuse-1}
\end{equation}
By Lemma \ref{lem:aux-9-projection-Lambda-t-part1}, we have
\[
\Pi\{E(B\mid\sigma_{t+1},V_t)-E(B\mid V_t)\mid\Lambda_t^{\perp}\}=\{R_t-T_t^{-1}E(R_tW_t\mid H_t)\}W_t\sigma_{t+1},
\]
where $R_t=E(B\sigma_{t+1}\mid V_t)$. By Lemma \ref{lem:aux-9-projection-Lambda-t-part2},
we have $\Pi\{E(B\mid H_t)-E(B\mid V_{t-1})\mid\Lambda_t^{\perp}\}=0$.
Plugging those into (\ref{eq:lem-projection-Lambda-t-proofuse-1})
completes the proof.
\end{proof}
\begin{lemma}
\label{lem:Lambda-t-alternative-form}The orthogonal complement of
the nuisance tangent space, $\Lambda_t^{\perp}$, is
\[
\Lambda_t^{\perp}=\{d(V_t)\sigma_{t+1}:\text{ any }d(V_t)\in\mathbb{R}^{p}\text{ such that }E[d(V_t)\mid H_t]=0\}.
\]
\end{lemma}

\begin{proof}
Lemma \ref{lem:projection-Lambda-t} implies that
\[
\Lambda_t^{\perp}=\{[R_t-T_t^{-1}E(R_tW_t\mid H_t)]W_t\sigma_{t+1}:R_t=E(h\sigma_{t+1}\mid V_t),h\in\mathcal{H}\}.
\]
Denote by $\Lambda_t^{\perp,\text{conj}}=\{d(V_t)\sigma_{t+1}:\text{ any }d(V_t)\text{ such that }E[d(V_t)\mid H_t]=0\}$.
In the following we show $\Lambda_t^{\perp}=\Lambda_t^{\perp,\text{conj}}$.

First we show $\Lambda_t^{\perp}\subset\Lambda_t^{\perp,\text{conj}}$.
For any $h\in\mathcal{H}$, we have
\begin{align*}
 & E[\{R_t-T_t^{-1}E(R_tW_t\mid X_t)\}W_t\sigma_{t+1}\mid H_t]\\
= & E[\{R_t-T_t^{-1}E(R_tW_t\mid X_t)\}W_tE(\sigma_{t+1}\mid V_t)\mid H_t]=0.
\end{align*}
Hence $\Lambda_t^{\perp}\subset\Lambda_t^{\perp,\text{conj}}$.

Next we show $\Lambda_t^{\perp,\text{conj}}\subset\Lambda_t^{\perp}$.
For any $d(V_t)\sigma_{t+1}\in\Lambda_t^{\perp,\text{conj}}$,
i.e. for any $d(V_t)$ such that $E\{d(V_t)\mid H_t\}=0$, let
$h=d(V_t)\sigma_{t+1}\in\mathcal{H}$, and we have
\[
R_tW_t=E(h\sigma_{t+1}\mid V_t)W_t=d(V_t)E(\sigma_{t+1}^{2}\mid V_t)W_t=d(V_t),
\]
and so $E(R_tW_t\mid H_t)=0$. Therefore,
\[
\{R_t-T_t^{-1}E(R_tW_t\mid H_t)\}W_t\sigma_{t+1}=R_tW_t\sigma_{t+1}=d(V_t)\sigma_{t+1}.
\]
This implies that $d(V_t)\sigma_{t+1}\in\Lambda_t^{\perp}$, and
hence $\Lambda_t^{\perp,\text{conj}}\subset\Lambda_t^{\perp}$.
This completes the proof.
\end{proof}
%

%
\begin{lemma}
\label{lem:Lambda}The orthogonal complement of the nuisance tangent
space for model $\mathcal{M}$ defined in Lemma \ref{lem:semimodel-Mt}
is
\begin{align}
\Lambda^{\perp}=\{\sum_{t=1}^{T}d_t(V_t)\dot{U}_{t+1}(\psi_{0}):\text{ any }d_t(V_t)\in\mathbb{R}^{p}\text{ s.t. }E[d_t(V_t)\mid H_t]=0\}, \label{proofuse-semi-1500}
\end{align}
where $U_{t+1}(\psi)=Y_{t,1}\exp\{-\gamma(t+1,V_t;\psi_{0})\}$
and $\dot{U}_{t+1}(\psi)=U_{t+1}(\psi)-E\{U_{t+1}(\psi_{0})\mid H_t\}$
\end{lemma}

\begin{proof}
Lemma \ref{lem:Lambda-t-alternative-form} implies that 
\begin{align*}
    \Lambda_t^{\perp}=\{d(V_t)[U_{t+1}(\psi_{0})-E\{U_{t+1}(\psi_{0})\mid H_t\}]:\text{ any }d(V_t)\in\mathbb{R}^{p}\text{ such that }E[d(V_t)\mid H_t]=0\}.
\end{align*}
Because $(\bigcap_{t=1}^{T}\Lambda_t)^{\perp}=\sum_{t=1}^{T}\Lambda_t^{\perp}$,
\eqref{proofuse-semi-1500} is an immediate implication of Lemma \ref{lem:semimodel-Mt}.
\end{proof}
\begin{lemma}\label{lem:EIF}
The efficient score $S_{\eff}(\psi_{0})$ is
\[
S_{\eff}(\psi_{0})=-\sum_{t=1}^{T}\bigg[E\bigg\{\frac{\partial U_{t+1}(\psi_{0})}{\partial\psi}\mid V_t\bigg\}-E\bigg\{\frac{\partial U_{t+1}(\psi_{0})}{\partial\psi}W_t\mid H_t\bigg\} E(W_t\mid H_t)^{-1}\bigg]W_t\dot{U}_{t+1}(\psi_{0}),
\]
where $W_t=\var\{U_{t+1}(\psi_{0})\mid V_t\}^{-1}$.
\end{lemma}

\begin{proof}
By definition, the efficient score is the projection of the score
for $\psi$, $S_{\psi}$, onto $\Lambda^{\perp}$. For any $1\leq t<s\leq T$,
for any $d_t(V_t)\dot{U}_{t+1}(\psi_{0})\in\Lambda_t^{\perp}$
and $d_{s}(V_{s})\dot{U}_{s+1}(\psi_{0})\in\Lambda_{s}^{\perp}$,
their inner product is
\[
E\{d_t(V_t)\dot{U}_{t+1}(\psi_{0})d_{s}(V_{s})\dot{U}_{s+1}(\psi_{0})\}=E[d_t(V_t)\dot{U}_{t+1}(\psi_{0})d_{s}(V_{s})E\{\dot{U}_{s+1}(\psi_{0})\mid V_{s}\}]=0,
\]
where the last equality follows from Lemma \ref{lem:aux-1}. This
implies that $\Lambda_t^{\perp}\perp\Lambda_{s}^{\perp}$ for any
$1\leq t<s\leq T$. Therefore, $\Lambda^{\perp}=\bigoplus_{t=1}^{T}\Lambda_t^{\perp}$,
and $\Pi(S_{\psi}\mid\Lambda^{\perp})=\sum_{t=1}^{T}\Pi(S_{\psi}\mid\Lambda_t^{\perp})$.
By Lemma \ref{lem:projection-Lambda-t}, we have
\begin{align}
\Pi(S_{\psi}\mid\Lambda_t^{\perp}) & =\{R_t-T_t^{-1}E(R_tW_t\mid H_t)\}W_t\sigma_{t+1}\nonumber \\
 & =\{E(S_{\psi}\sigma_{t+1}\mid V_t)-E(S_{\psi}\sigma_{t+1}W_t\mid H_t)E(W_t\mid H_t)^{-1}\}W_t\sigma_{t+1}.\label{eq:lem-eif-proofuse-1}
\end{align}
We have $\sigma_{t+1}=U_{t+1}(\psi_{0})-E_{P}\{U_{t+1}(\psi_{0})\mid V_t\}=\dot{U}_{t+1}(\psi_{0})$
as in Lemma \ref{lem:aux-2-change-of-variable}, so $W_t=\var(\sigma_{t+1}\mid V_t)^{-1}=\var\{U_{t+1}(\psi_{0})\mid V_t\}^{-1}$.
By the generalized information equality \citep{newey1990}
\[
E(S_{\psi}\dot{U}_{t+1}(\psi_{0})\mid V_t)=-E\bigg\{\frac{\partial\dot{U}_{t+1}(\psi_{0})}{\partial\psi}\mid V_t\bigg\}=-E\bigg\{\frac{\partial U_{t+1}(\psi_{0})}{\partial\psi}\mid V_t\bigg\}.
\]
So (\ref{eq:lem-eif-proofuse-1}) becomes
\begin{align*}
\Pi(S_{\psi}\mid\Lambda_t^{\perp}) & =\{E(S_{\psi}\sigma_{t+1}\mid V_t)-E(S_{\psi}\sigma_{t+1}W_t\mid H_t)E(W_t\mid H_t)^{-1}\}W_t\sigma_{t+1}\\
 & =-\bigg[E\bigg\{\frac{\partial U_{t+1}(\psi_{0})}{\partial\psi}\mid V_t\bigg\}-E\bigg\{\frac{\partial U_{t+1}(\psi_{0})}{\partial\psi}W_t\mid H_t\bigg\} E(W_t\mid H_t)^{-1}\bigg]W_t\dot{U}_{t+1}(\psi_{0}).
\end{align*}
Thus, the form of $S_{\eff}(\psi_{0})$ follows from the fact
that $S_{\eff}(\psi_{0})=\Pi(S_{\psi}\mid\Lambda^{\perp})=\sum_{t=1}^{T}\Pi(S_{\psi}\mid\Lambda_t^{\perp})$.
\end{proof}

\section{Supporting lemmas used in Section \ref{appen:proof-efficient-score}} \label{sec:supporting-lemmas}
\begin{lemma}
\label{lem:aux-1}$E\{U_{t+1}(\psi_{0})\mid H_t,A_t\}=E\{U_{t+1}(\psi_{0})\mid H_t\}$.
\end{lemma}

\begin{proof}
We have
\begin{align*}
 & E\{H_{t+1}(\psi)\mid H_t,A_t\}\\
= & E\{Y_{t,1}\mid H_t,A_t\}\\
\text{(consistency)}= & E\{Y_{t,1}(\bar{A}_t)\mid H_t,A_t\}e^{-\gamma(t+1,\bar{Z}_t,\bar{A}_t;\psi)}\\
\text{(by (\ref{eq:assump-model}))}= & E\{Y_{t,1}(\bar{A}_{t-1},0)\mid H_t,A_t\}e^{\gamma(t+1,\bar{Z}_t,\bar{A}_t;\psi_{0})-\gamma(t+1,\bar{Z}_t,\bar{A}_t;\psi)}\\
\text{(sequential ignorability)}= & E\{Y_{t,1}(\bar{A}_{t-1},0)\mid H_t\}e^{\gamma(t+1,\bar{Z}_t,\bar{A}_t;\psi_{0})-\gamma(t+1,\bar{Z}_t,\bar{A}_t;\psi)}.
\end{align*}
Therefore $E\{U_{t+1}(\psi_{0})\mid H_t,A_t\}=E\{U_{t+1}(\psi_{0})\mid H_t\}$.
\end{proof}
\begin{lemma}
\label{lem:aux-2-change-of-variable}Let $\sigma_{t+1}$ be a random
variable that is defined on the same sample space as $V_{T+1}$. Consider
a tuple $(P^{'},q_t(H_t),\beta_t(H_{t-1},A_{t-1})),$ where
$P'$ is a probability distribution of $V_{T+1}\cup\sigma_{t+1}\backslash Y_{t,1}$,
$q_t$ is a (deterministic) function of $H_t$, and $\beta_t$
is a (deterministic) function of $H_{t-1},A_{t-1}$. Define $\mathcal{M}_t^{'}$
the collection of $(P^{'},q_t(H_t),\beta_t(H_{t-1},A_{t-1}))$
tuples such that positivity holds for $P^{'}$ and that
\begin{align}
E(\sigma_{t+1}\mid H_t,A_t) & =0,\label{eq:model-Mt-1}\\
q_t(H_t,A_t) & =q_t(H_t)\text{ is a function of \ensuremath{H_t}},\label{eq:model-Mt-2}\\
E(q_t\mid H_{t-1},A_{t-1}) & =0.\label{eq:model-Mt-3}
\end{align}
Then there is a 1-1 mapping $g$ between $\mathcal{M}_t$ and $\mathcal{M}_t^{'}$
given by:
\[
g:P\mapsto(P^{'},q_t(H_t),\beta_t(H_{t-1},A_{t-1}))
\]
where $P'$ is induced by $P$ and $\sigma_{t+1}=U_{t+1}(\psi_{0})-E_{P}\{U_{t+1}(\psi_{0})\mid H_t,A_t\}$,
$q_t(H_t,A_t)=E_{P}\{U_{t+1}(\psi_{0})\mid H_t,A_t\}-E_{P}\{U_{t+1}(\psi_{0})\mid H_{t-1},A_{t-1}\}$
and $,\beta_t(H_{t-1},A_{t-1})=E_{P}\{U_{t+1}(\psi_{0})\mid H_{t-1},A_{t-1}\}$.
The inverse mapping is
\[
g^{-1}:(P^{'},q_t(H_t),\beta_t(H_{t-1},A_{t-1}))\mapsto P
\]
where $P$ is induced by $P'$ and $Y_{t,1}=e^{\gamma(t+1,H_t,A_t;\psi_{0})}(\sigma_{t+1}+q_t+\beta_t)$.
\end{lemma}

\begin{proof}
First, we show that the $g(P)\in\mathcal{M}_t^{'}$. Let $q_t=E\{U_{t+1}(\psi_{0})\mid H_t,A_t\}-E\{U_{t+1}(\psi_{0})\mid H_{t-1},A_{t-1}\}$,
$\beta_t=E\{U_{t+1}(\psi_{0})\mid H_{t-1},A_{t-1}\}$, and $\sigma_{t+1}=Y_{t,1}-q_t-\beta_t$.
Let $P'$ be the probability distribution of $V_{T+1}\cup\sigma_{t+1}\backslash Y_{t,1}$
induced by $\sigma_{t+1}=Y_{t,1}e^{-\gamma(t+1,H_t,A_t;\psi_{0})}-q_t-\beta_t$
and $P$. Trivially we have $E(\sigma_{t+1}\mid H_t,A_t)=0$ and
$E(q_t\mid H_{t-1},A_{t-1})=0$. Because $P\in\mathcal{M}_t$,
Lemma \ref{lem:aux-1} implies $q_t=q_t(H_t)$. Therefore, $(P^{'},q_t,\beta_t)\in\mathcal{M}_t^{'}$.

Then we show that $g^{-1}\{P^{'},q_t(H_t),\beta_t(H_{t-1},A_{t-1})\}\in\mathcal{M}_t$.
We have
\begin{align*}
 & E\{Y_{t,1}e^{-\gamma(t+1,H_t,A_t;\psi_{0})}\mid H_t,A_t\}\\
= & E\{\sigma_{t+1}+q_t+\beta_t\mid H_t,A_t\}\\
= & 0+q_t(H_t)+\beta_t(H_{t-1},A_{t-1}).
\end{align*}
Set $A_t=0$, we have
weak sequential ignorability
\[
E\{Y_{t,1}(\bar{A}_{t-1},0)\mid H_t,A_t=0\}=E\{Y_{t,1}(\bar{A}_{t-1},0)\mid H_t\}.
\]
 Taking the ratio between $A_t=a_t$ and $A_t=0$, we have (\ref{eq:assump-model})
\[
\frac{E\{Y_{t,1}(\bar{A}_{t-1},a_t)\mid H_t,A_t\}}{E\{Y_{t,1}(\bar{A}_{t-1},0)\mid H_t,A_t\}}=e^{\gamma(t+1,H_t,a_t;\psi_{0})}.
\]
Therefore $P\in\mathcal{M}_t$. 
\end{proof}
\begin{lemma}
\label{lem:aux-3-orthogonal-Lambdas}$\Lambda_t^{1}$, $\{\Gamma_{m}^{3}\}_{1\le m\leq t}$,
$\{\Lambda_{m}^{\bullet}\}_{1\leq m\leq t-1}$ are orthogonal to each
other and orthogonal to the rest subspaces $\Lambda_t^{2},\Gamma_t^{4},\Lambda_t^{5},\Lambda_t^{6}$.
\end{lemma}

\begin{proof}
Using the definition in Lemma \ref{lem:nuisance-space-decomp}, we
have the following. (We will repeatedly use the fact $E(S_t\mid V_t)=0$,
which is shown in Lemma \ref{lem:aux-5-integrate-S}.)
\begin{itemize}
\item $\Lambda_t^{1}\perp\Lambda_t^{2}$: $\forall A_t^{1}\in\Lambda_t^{1},A_t^{2}\in\Lambda_t^{2}$,
we have
\[
E(A_t^{1}A_t^{2})=E\{E(A_t^{1}A_t^{2}\mid V_t,Y_{t,1})\}=E\{A_t^{2}E(A_t^{1}\mid V_t,Y_{t,1})\}=0.
\]
Similarly, we can show $\Lambda_t^{1}\perp\Gamma_{m}^{3}$, $\Lambda_t^{1}\perp\Gamma_t^{4}$,
$\Lambda_t^{1}\perp\Lambda_{m}^{\bullet}$, $\Lambda_t^{1}\perp\Lambda_t^{5}$,
$\Lambda_t^{1}\perp\Lambda_t^{6}$.
\item $\Gamma_{m}^{3}\perp\Lambda_t^{2}$ for all $1\leq m\leq t$: $\forall A_{m}^{3}\in\Gamma_{m}^{3},A_t^{2}\in\Lambda_t^{2}$,
we have
\[
E(A_{m}^{3}A_t^{2})=E\{E(A_{m}^{3}A_t^{2}\mid V_t)\}=E\{A_{m}^{3}E(A_t^{2}\mid V_t)\}=0.
\]
\item $\Gamma_{m}^{3}\perp\Gamma_{k}^{3}$ for all $1\leq m<k\leq t$: $\forall A_{m}^{3}\in\Gamma_{m}^{3},A_{k}^{3}\in\Gamma_{k}^{3}$,
we have
\[
E(A_{m}^{3}A_{k}^{3})=E\{E(A_{m}^{3}A_{k}^{3}\mid H_{k})\}=E\{A_{m}^{3}E(A_{k}^{3}\mid H_{k})\}=0.
\]
\item $\Gamma_{m}^{3}\perp\Gamma_t^{4}$ for all $1\leq m\leq t-1$: $\forall A_{m}^{3}\in\Gamma_{m}^{3},A_t^{\bullet}+S_tE(Q_tA_t^{\bullet}\mid V_{t-1})\in\Gamma_t^{4}$,
we have
\begin{align*}
E[A_{m}^{3}\{A_t^{\bullet}+S_tE(Q_tA_t^{\bullet}\mid V_{t-1})\}] & =E(A_{m}^{3}A_t^{\bullet})+E\{A_{m}^{3}S_tE(Q_tA_t^{\bullet}\mid V_{t-1})\}\\
 & =E\{E(A_{m}^{3}A_t^{\bullet}\mid V_{t-1})\}+E[E\{A_{m}^{3}S_tE(Q_tA_t^{\bullet}\mid V_{t-1})\mid V_t\}]\\
 & =E\{A_{m}^{3}E(A_t^{\bullet}\mid V_{t-1})\}+E[A_{m}^{3}E(Q_tA_t^{\bullet}\mid V_{t-1})E(S_t\mid V_t)]\\
 & =0
\end{align*}
\item $\Gamma_t^{3}\perp\Gamma_t^{4}$: $\forall A_t^{3}\in\Gamma_t^{3},A_t^{\bullet}+S_tE(Q_tA_t^{\bullet}\mid V_{t-1})\in\Gamma_t^{4}$,
we have
\begin{align*}
E[A_t^{3}\{A_t^{\bullet}+S_tE(Q_tA_t^{\bullet}\mid V_{t-1})\}] & =E(A_t^{3}A_t^{\bullet})+E\{A_t^{3}S_tE(Q_tA_t^{\bullet}\mid V_{t-1})\}\\
 & =E\{E(A_t^{3}A_t^{\bullet}\mid H_t)\}+E[E\{A_t^{3}S_tE(Q_tA_t^{\bullet}\mid V_{t-1})\mid V_t\}]\\
 & =E\{A_t^{\bullet}E(A_t^{3}\mid H_t)\}+E[A_t^{3}E(Q_tA_t^{\bullet}\mid V_{t-1})E(S_t\mid V_t)]\\
 & =0
\end{align*}
\item $\Gamma_{m}^{3}\perp\Lambda_{k}^{\bullet}$ for all $1\leq m\leq t$
and $1\leq k\leq t-1$: $\forall A_{m}^{3}\in\Gamma_{m}^{3},A_{k}^{\bullet}\in\Lambda_{k}^{\bullet}$,
if $m<k$ we have
\[
E(A_{m}^{3}A_{k}^{\bullet})=E\{E(A_{m}^{3}A_{k}^{\bullet}\mid V_{k-1})\}=E\{A_{m}^{3}E(A_{k}^{\bullet}\mid V_{k-1})\}=0;
\]
if $m\geq k$ we have
\[
E(A_{m}^{3}A_{k}^{\bullet})=E\{E(A_{m}^{3}A_{k}^{\bullet}\mid H_{m})\}=E\{A_{k}^{\bullet}E(A_{m}^{3}\mid H_{m})\}=0.
\]
\item $\Gamma_{m}^{3}\perp\Lambda_t^{5}$ for all $1\leq m\leq t$: $\forall A_{m}^{3}\in\Gamma_{m}^{3},S_tA_t^{\bullet}\in\Lambda_t^{5}$,
we have
\[
E(A_{m}^{3}S_tA_t^{\bullet})=E\{E(A_{m}^{3}S_tA_t^{\bullet}\mid V_t)\}=E\{A_{m}^{3}A_t^{\bullet}E(S_t\mid V_t)\}=0.
\]
\item $\Gamma_{m}^{3}\perp\Lambda_t^{6}$ for all $1\leq m\leq t$: $\forall A_{m}^{3}\in\Gamma_{m}^{3},a(V_{t-1})S_t\in\Lambda_t^{6}$,
we have
\[
E\{A_{m}^{3}a(V_{t-1})S_t\}=E[E\{A_{m}^{3}a(V_{t-1})S_t\mid V_t\}]=E\{A_{m}^{3}a(V_{t-1})E(S_t\mid V_t)\}=0.
\]
\item $\Lambda_{m}^{\bullet}\perp\Lambda_t^{2}$ for all $1\leq m\leq t-1$:
$\forall A_{m}^{\bullet}\in\Lambda_{m}^{\bullet},A_t^{2}\in\Lambda_t^{2}$,
we have
\[
E(A_{m}^{\bullet}A_t^{2})=E\{E(A_{m}^{\bullet}A_t^{2}\mid V_t)\}=E\{A_{m}^{\bullet}E(A_t^{2}\mid V_t)\}=0.
\]
\item $\Lambda_{m}^{\bullet}\perp\Gamma_t^{4}$ for all $1\leq m\leq t-1$:
$\forall A_{m}^{\bullet}\in\Lambda_{m}^{\bullet},A_t^{\bullet}+S_tE(Q_tA_t^{\bullet}\mid V_{t-1})\in\Gamma_t^{4}$,
we have
\begin{align*}
E[A_{m}^{\bullet}\{A_t^{\bullet}+S_tE(Q_tA_t^{\bullet}\mid V_{t-1})\}] & =E(A_{m}^{\bullet}A_t^{\bullet})+E\{A_{m}^{\bullet}S_tE(Q_tA_t^{\bullet}\mid V_{t-1})\}\\
 & =E\{E(A_{m}^{\bullet}A_t^{\bullet}\mid V_{t-1})\}+E[E\{A_{m}^{\bullet}S_tE(Q_tA_t^{\bullet}\mid V_{t-1})\mid V_t\}]\\
 & =E\{A_{m}^{\bullet}E(A_t^{\bullet}\mid V_{t-1})\}+E\{A_{m}^{\bullet}E(Q_tA_t^{\bullet}\mid V_{t-1})E(S_t\mid V_t)\}\\
 & =0
\end{align*}
\item $\Lambda_{m}^{\bullet}\perp\Lambda_{k}^{\bullet}$ for all $1\leq m<k\leq t-1$:
$\forall A_{m}^{\bullet}\in\Lambda_{m}^{\bullet},A_{k}^{\bullet}\in\Lambda_{k}^{\bullet}$,
we have
\[
E(A_{m}^{\bullet}A_{k}^{\bullet})=E\{E(A_{m}^{\bullet}A_{k}^{\bullet}\mid V_{k-1})\}=E\{A_{m}^{\bullet}E(A_{k}^{\bullet}\mid V_{k-1})\}=0.
\]
\item $\Lambda_{m}^{\bullet}\perp\Lambda_t^{5}$ for all $1\leq m\leq t-1$:
$\forall A_{m}^{\bullet}\in\Lambda_{m}^{\bullet},S_tA_t^{\bullet}\in\Lambda_t^{5}$,
we have
\[
E(A_{m}^{\bullet}S_tA_t^{\bullet})=E\{E(A_{m}^{\bullet}S_tA_t^{\bullet}\mid V_t)\}=E\{A_{m}^{\bullet}A_t^{\bullet}E(S_t\mid V_t)\}=0.
\]
\item $\Lambda_{m}^{\bullet}\perp\Lambda_t^{6}$ for all $1\leq m\leq t-1$:
$\forall A_{m}^{\bullet}\in\Lambda_{m}^{\bullet},a(V_{t-1})S_t\in\Lambda_t^{6}$,
we have
\[
E\{A_{m}^{\bullet}a(V_{t-1})S_t\}=E[E\{A_{m}^{\bullet}a(V_{t-1})S_t\mid V_t\}]=E\{A_{m}^{\bullet}a(V_{t-1})E(S_t\mid V_t)\}=0.
\]
\end{itemize}
\end{proof}
\begin{lemma}[Projection onto $\tilde{\Lambda}_t^{2}$]
\label{lem:aux-4-projection-1}Let $\mathcal{G}$ be the Hilbert
space of all mean-zero finite-variance functions of $(X,Y)$, where
$(X,Y)$ follows some unknown distribution $P$ and the only restriction
is $E(X\mid Y)=0$. Let
\[
\Lambda=\{h(X,Y)\in\mathcal{G}:E(h\mid Y)=0,E(hX\mid Y)=0\},
\]
then we have $\Lambda=\{O(h):h\in\mathcal{G}\}=O(\mathcal{G})$, where
the operator $O=O_{2}\circ O_{1}$ with
\begin{align*}
O_{1}(h) & =h-E(h\mid Y),\\
O_{2}(h) & =h-E(hX\mid Y)\var(X\mid Y)^{-1}X.
\end{align*}
Both $O_{1}$ and $O_{2}$ are self-adjoint, i.e., $O_{1}^{*}=O_{1}$
and $O_{2}^{*}=O_{2}$. For any $h(X,Y)\in\mathcal{G}$, its projection
onto $\Lambda$ equals
\[
\Pi\{h(X,Y)\mid\Lambda\}=h-E(hX\mid Y)\var(X\mid Y)^{-1}X-E(h\mid Y).
\]
\end{lemma}

\begin{proof}
The first statment in the Lemma is that $\Lambda=\Lambda^{'}=O(\mathcal{G})$.
To show this, we first show $\Lambda^{'}\subset\Lambda$. For any
$h\in\mathcal{G}$, i.e. for any $O(h)\in\Lambda^{'}$, we have
\begin{align*}
E\{O(h)\mid Y\} & =E[O_{1}(h)-E\{O_{1}(h)X\mid Y\}\var(X\mid Y)^{-1}X\mid Y]\\
 & =E\{O_{1}(h)\mid Y\}-E\{O_{1}(h)X\mid Y\}\var(X\mid Y)^{-1}E(X\mid Y)\\
 & =0-0=0,
\end{align*}
and
\begin{align*}
E\{O(h)X\mid Y\} & =E[O_{1}(h)X-E\{O_{1}(h)X\mid Y\}\var(X\mid Y)^{-1}X^{2}\mid Y]\\
 & =E\{O_{1}(h)X\mid Y\}-E\{O_{1}(h)X\mid Y\}=0,
\end{align*}
so $O(h)\in\Lambda$ and $\Lambda^{'}\subset\Lambda$. Next, we show
$\Lambda\subset\Lambda^{'}$, i.e., for any $h\in\Lambda$, there
exists $g\in\mathcal{G}$ such that $O(g)=h$. We claim that 
\begin{align}
O(g) & =g-E(g\mid Y)-E[\{g-E(g\mid Y)\}X\mid Y]\var(X\mid Y)^{-1}X\nonumber \\
 & =g-E(g\mid Y)-E(gX\mid Y)\var(X\mid Y)^{-1}X+E(g\mid Y)E(X\mid Y)\var(X\mid Y)^{-1}X\label{eq:lem-aux-4-proofuse-1}
\end{align}
Because $h\in\Lambda$, we have $E(h\mid Y)=0$ and $E(hX\mid Y)=0$.
Therefore, let $g=h$ in (\ref{eq:lem-aux-4-proofuse-1}) and it becomes
$O(h)=h$, and thus $\Lambda\subset\Lambda^{'}$.

Next we show that $O_{1}$ and $O_{2}$ are both self-adjoint. For
any $h,g\in\mathcal{G}$, we have
\begin{align*}
<O_{1}(h),g> & =E[\{h-E(h\mid Y)\}g]=E(hg)-E\{E(h\mid Y)g\}\\
 & =E(hg)-E\{hE(g\mid Y)\}=E[h\{g-E(g\mid Y)\}],
\end{align*}
and
\begin{align*}
<O_{2}(h),g> & =E[\{h-E(hX\mid Y)\var(X\mid Y)^{-1}X\}g]\\
 & =E(hg)-E[E\{E(hX\mid Y)\var(X\mid Y)^{-1}Xg\mid Y\}]\\
 & =E(hg)-E[E(hX\mid Y)\var(X\mid Y)^{-1}E\{Xg\mid Y\}]\\
 & =E(hg)-E[E\{hX\var(X\mid Y)^{-1}E(Xg\mid Y)\mid Y\}]\\
 & =E[h\{g-X\var(X\mid Y)^{-1}E(Xg\mid Y)\}].
\end{align*}
Hence both $O_{1}$ and $O_{2}$ are both self-adjoint. 

The adjoint operator for $O$ is
\begin{align*}
O^{*}(h) & =O_{1}^{*}\circ O_{2}^{*}(h)\\
 & =h-E(hX\mid Y)\var(X\mid Y)^{-1}X-E\{h-E(hX\mid Y)\var(X\mid Y)^{-1}X\mid Y\}\\
 & =h-E(hX\mid Y)\var(X\mid Y)^{-1}X-E(h\mid Y).
\end{align*}
By a functional analysis result, for any $h\in\mathcal{G}$, its projection
$\Pi(h\mid\Lambda)$ satisfies $O^{*}\{\Pi(h\mid\Lambda)\}=O^{*}(h)$,
i.e.,
\begin{align}
 & \Pi(h\mid\Lambda)-E\{\Pi(h\mid\Lambda)X\mid Y\}\var(X\mid Y)^{-1}X-E\{\Pi(h\mid\Lambda)\mid Y\}\nonumber \\
= & h-E(hX\mid Y)\var(X\mid Y)^{-1}X-E(h\mid Y).\label{eq:lem-aux-4-proofuse-2}
\end{align}
Because $\Pi(h\mid\Lambda)\in\Lambda$, $E\{\Pi(h\mid\Lambda)X\mid Y\}=E\{\Pi(h\mid\Lambda)\mid Y\}=0$,
so (\ref{eq:lem-aux-4-proofuse-2}) yields
\[
\Pi(h\mid\Lambda)=h-E(hX\mid Y)\var(X\mid Y)^{-1}X-E(h\mid Y).
\]
This completes the proof.
\end{proof}
\begin{lemma}
\label{lem:aux-5-integrate-S}Consider a random variable $X$ with
$E(X)=0$. Let $S(x)=\partial\log p(x)/\partial x$ where $p(x)$
is the density of $X$. Then under regularity conditions, $E\{S(X)\}=0$
and $E\{S(X)X\}=-1$.
\end{lemma}

\begin{proof}
We have
\begin{align*}
E\{S(X)\} & =\int_{-\infty}^{\infty}p(x)\frac{\partial\log p(x)}{\partial x}dx=\int_{-\infty}^{\infty}\frac{\partial p(x)}{\partial x}dx\\
 & =\int_{-\infty}^{\infty}\frac{\partial p(x+u)}{\partial u}\Big|_{u=0}dx=\frac{\partial}{\partial u}\left\{ \int_{-\infty}^{\infty}p(x+u)dx\right\} \Big|_{u=0}=0
\end{align*}
and
\begin{align*}
E\{S(X)X\} & =\int_{-\infty}^{\infty}p(x)\frac{\partial\log p(x)}{\partial x}xdx=\int_{-\infty}^{\infty}\frac{\partial p(x)}{\partial x}xdx\\
 & =\int_{-\infty}^{\infty}\frac{\partial p(x+u)}{\partial u}x\Big|_{u=0}dx=\frac{\partial}{\partial u}\left\{ \int_{-\infty}^{\infty}p(x+u)xdx\right\} \Big|_{u=0}\\
 & =\frac{\partial}{\partial u}\left\{ \int_{-\infty}^{\infty}p(t)(t-u)dt\right\} \Big|_{u=0}=\frac{\partial}{\partial u}(-u)\Big|_{u=0}=-1.
\end{align*}
\end{proof}
\begin{lemma}[Projection onto $\tilde{\Lambda}_t^{5}$]
\label{lem:aux-6-Lambda-tilde-5}Consider $\tilde{\Lambda}_t^{5}=\{A_t^{\bullet}W_t\sigma_{t+1}:A_t^{\bullet}\in\Lambda_t^{\bullet}\}$
as defined in Lemma \ref{lem:nuisance-space-orthogonalization}. For
any mean zero function $h(\sigma_{t+1},V_t)\in\mathcal{H}$, its
projection onto $\tilde{\Lambda}_t^{5}$ is
\[
\Pi\{h(\sigma_{t+1},V_t)\mid\tilde{\Lambda}_t^{5}\}=\tilde{A}_t^{\bullet}W_t\sigma_{t+1},
\]
where $\tilde{A}_t^{\bullet}=-E(D_tT_t^{-1}\mid V_{t-1})(T_t^{\bullet}T_t)^{-1}+D_tT_t^{-1}$,
with $D_t=E\{h(\sigma_{t+1},V_t)W_t\sigma_{t+1}\mid H_t\}$
and $W_t,T_t,T_t^{\bullet}$ as defined in Lemma \ref{lem:nuisance-space-orthogonalization}.
\end{lemma}

\begin{proof}
We can express $\tilde{\Lambda}_t^{5}$ as the image of the operator
$O_{2}\circ O_{1}$: $\tilde{\Lambda}_t^{5}=O_{2}\circ O_{1}(\mathcal{H})$,
where for any $g(V_{T+1})\in\mathcal{H}$ define $O_{1}(g)=E(g\mid H_t)-E(g\mid V_{t-1})$
and $O_{2}(g)=gW_t\sigma_{t+1}$. It follows that both $O_{1}$
and $O_{2}$ are self-adjoint. Suppose $\Pi\{h(\sigma_{t+1},V_t)\mid\tilde{\Lambda}_t^{5}\}=\tilde{A}_t^{\bullet}W_t\sigma_{t+1}$
for some $\tilde{A}_t^{\bullet}(H_t)\in\Lambda_t^{\bullet}$.
By a functional analysis result, we have $O_{1}^{*}\circ O_{2}^{*}(\tilde{A}_t^{\bullet}W_t\sigma_{t+1})=O_{1}^{*}\circ O_{2}^{*}\{h(\sigma_{t+1},V_t)\}$,
i.e.,
\begin{equation}
E(\tilde{A}_t^{\bullet}W_t^{2}\sigma_{t+1}^{2}\mid H_t)-E(\tilde{A}_t^{\bullet}W_t^{2}\sigma_{t+1}^{2}\mid V_{t-1})=E(hW_t\sigma_{t+1}\mid H_t)-E(hW_t\sigma_{t+1}\mid V_{t-1}).\label{eq:lem-aux-6-proofuse-1}
\end{equation}
By Lemma \ref{lem:aux-comp-useful-result}, $E(\tilde{A}_t^{\bullet}W_t^{2}\sigma_{t+1}^{2}\mid H_t)=\tilde{A}_t^{\bullet}T_t$
and $E(\tilde{A}_t^{\bullet}W_t^{2}\sigma_{t+1}^{2}\mid V_{t-1})=E\{E(\tilde{A}_t^{\bullet}W_t^{2}\sigma_{t+1}^{2}\mid H_t)\mid V_{t-1}\}=E(\tilde{A}_t^{\bullet}T_t\mid V_{t-1})$.
So by the definition of $D_t$, (\ref{eq:lem-aux-6-proofuse-1})
becomes $\tilde{A}_t^{\bullet}T_t-E(\tilde{A}_t^{\bullet}T_t\mid V_{t-1})=D_t-E(D_t\mid V_{t-1})$,
i.e.,
\begin{equation}
\tilde{A}_t^{\bullet}=T_t^{-1}\{E(\tilde{A}_t^{\bullet}T_t\mid V_{t-1})+D_t-E(D_t\mid V_{t-1})\}.\label{eq:lem-aux-6-proofuse-2}
\end{equation}
Because $\tilde{A}_t^{\bullet}\in\Lambda_t^{\bullet}$, $E(\tilde{A}_t^{\bullet}\mid V_{t-1})=0$,
so taking $E(\cdot\mid V_{t-1})$ on both sides of (\ref{eq:lem-aux-6-proofuse-2})
gives
\[
0=E(T_t^{-1}\mid V_{t-1})E(\tilde{A}_t^{\bullet}T_t\mid V_{t-1})+E(T_t^{-1}D_t\mid V_{t-1})-E(T_t^{-1}\mid V_{t-1})E(D_t\mid V_{t-1}).
\]
By the definition of $T_t^{\bullet}$ this becomes
\begin{equation}
E(\tilde{A}_t^{\bullet}T_t\mid V_{t-1})=E(D_t\mid V_{t-1})-(T_t^{\bullet})^{-1}E(T_t^{-1}D_t\mid V_{t-1}).\label{eq:lem-aux-6-proofuse-3}
\end{equation}
Plugging (\ref{eq:lem-aux-6-proofuse-3}) into (\ref{eq:lem-aux-6-proofuse-2})
yields
\begin{align*}
\tilde{A}_t^{\bullet} & =T_t^{-1}\{E(D_t\mid V_{t-1})-(T_t^{\bullet})^{-1}E(T_t^{-1}D_t\mid V_{t-1})+D_t-E(D_t\mid V_{t-1})\}\\
 & =-(T_tT_t^{\bullet})^{-1}E(T_t^{-1}D_t\mid V_{t-1})+T_t^{-1}D_t.
\end{align*}
This completes the proof.
\end{proof}
\begin{lemma}[Projection onto $\tilde{\Gamma}_t^{4}$]
\label{lem:aux-7-Gamma-tilde-4}Consider $\tilde{\Gamma}_t^{4}=\{A_t^{\bullet}-E(Q_tA_t^{\bullet}\mid V_{t-1})(T_t^{\bullet})^{-1}T_t^{-1}W_t\sigma_{t+1}:A_t^{\bullet}\in\Lambda_t^{\bullet}\}$
as defined in Lemma \ref{lem:nuisance-space-orthogonalization}. For
any $h(\sigma_{t+1},V_t)\in\mathcal{H}$, its projection onto $\tilde{\Gamma}_t^{4}$
is
\begin{equation}
\Pi\{h(\sigma_{t+1},V_t)\mid\tilde{\Gamma}_t^{4}\}=O_{3}^{*}O_{4}^{*}(h)-E\{O_{3}^{*}O_{4}^{*}(h)Q_t\mid V_{t-1}\}W_{t,t-1}\epsilon_t,\label{eq:lem-aux-7-statement-1}
\end{equation}
where
\begin{align*}
\epsilon_t & =T_t^{-1}W_t\sigma_{t+1}+Q_t,\\
W_{t,t-1} & =\var(\epsilon_t\mid V_{t-1})^{-1},\\
O_{3}(h) & =E(h\mid H_t)-E(h\mid V_{t-1}),\\
O_{3}^{*}(h) & =O_{3}(h),\\
O_{4}(h) & =h-E(hQ_t\mid V_{t-1})(T_t^{\bullet})^{-1}T_t^{-1}W_t\sigma_{t+1},\\
O_{4}^{*}(h) & =h-E\{h(T_t^{\bullet})^{-1}T_t^{-1}W_t\sigma_{t+1}\mid V_{t-1}\}Q_t,
\end{align*}
and
\begin{equation}
O_{3}^{*}O_{4}^{*}(h)=E(h\mid H_t)-E(h\mid V_{t-1})-E\{h(T_t^{\bullet})^{-1}T_t^{-1}W_t\sigma_{t+1}\mid V_{t-1}\}Q_t.\label{eq:lem-aux-7-statement-2}
\end{equation}
In particular, if $h=h(V_t)$, then
\begin{equation}
\Pi\{h(V_t)\mid\tilde{\Gamma}_t^{4}\}=E(h\mid H_t)-E(h\mid V_{t-1})-E(hQ_t\mid V_{t-1})W_{t,t-1}\epsilon_t;\label{eq:lem-aux-7-statement-3}
\end{equation}
if $h=h(V_{t-1})$, then $\Pi\{h(V_{t-1})\mid\tilde{\Gamma}_t^{4}\}=0$.
Here, $W_t,T_t,T_t^{\bullet}$ are defined in Lemma \ref{lem:nuisance-space-orthogonalization}.
\end{lemma}

\begin{proof}
By definition it is straightforward that $\tilde{\Gamma}_t^{4}=O_{4}\circ O_{3}(\mathcal{H})$
and that $O_{3}^{*}=O_{3}$. To derive $O_{4}^{*}$, for any $h,g\in\mathcal{H}$
we have
\begin{align*}
<O_{4}(h),g> & =E[\{h-E(hQ_t\mid V_{t-1})(T_t^{\bullet})^{-1}T_t^{-1}W_t\sigma_{t+1}\}g]\\
 & =E(hg)-E\{E(hQ_t\mid V_{t-1})(T_t^{\bullet})^{-1}T_t^{-1}W_t\sigma_{t+1}g\}\\
 & =E(hg)-E[hQ_tE\{(T_t^{\bullet})^{-1}T_t^{-1}W_t\sigma_{t+1}g\mid V_{t-1}\}]\\
 & =E(h[g-Q_tE\{(T_t^{\bullet})^{-1}T_t^{-1}W_t\sigma_{t+1}g\mid V_{t-1}\}]),
\end{align*}
so $O_{4}^{*}(h)=h-E\{h(T_t^{\bullet})^{-1}T_t^{-1}W_t\sigma_{t+1}\mid V_{t-1}\}Q_t$.
Using the fact that $Q_t=Q_t(H_t)$ and $E(Q_t\mid V_{t-1})=0$,
we have
\begin{align*}
O_{3}^{*}O_{4}^{*}(h) & =E\{O_{4}^{*}(h)\mid H_t\}-E\{O_{4}^{*}(h)\mid V_{t-1}\}\\
 & =E\{h-E\{h(T_t^{\bullet})^{-1}T_t^{-1}W_t\sigma_{t+1}\mid V_{t-1}\}Q_t\mid H_t\}\\
 & ~~-E\{h-E\{h(T_t^{\bullet})^{-1}T_t^{-1}W_t\sigma_{t+1}\mid V_{t-1}\}Q_t\mid V_{t-1}\}\\
 & =E(h\mid H_t)-E(h\mid V_{t-1})-E\{h(T_t^{\bullet})^{-1}T_t^{-1}W_t\sigma_{t+1}\mid V_{t-1}\}Q_t.
\end{align*}

To derive the projection, for a given $h(\sigma_{t+1},V_t)\in\mathcal{H}$,
suppose
\begin{equation}
\Pi\{h(\sigma_{t+1},V_t)\mid\tilde{\Gamma}_t^{4}\}=A_t^{\bullet}-E(Q_tA_t^{\bullet}\mid V_{t-1})(T_t^{\bullet})^{-1}T_t^{-1}W_t\sigma_{t+1}\equiv h_{p}\label{eq:lem-aux-7-proofuse-0.5}
\end{equation}
for some $A_t^{\bullet}\in\Lambda_t^{\bullet}$. We calculate
a few terms:
\begin{align*}
E(h_{p}\mid H_t) & =E\{A_t^{\bullet}-E(Q_tA_t^{\bullet}\mid V_{t-1})(T_t^{\bullet})^{-1}T_t^{-1}W_t\sigma_{t+1}\mid H_t\}\\
 & =A_t^{\bullet}-E\{E(Q_tA_t^{\bullet}\mid V_{t-1})(T_t^{\bullet})^{-1}T_t^{-1}W_tE(\sigma_{t+1}\mid V_t)\mid H_t\}\\
 & =A_t^{\bullet},\\
E(h_{p}\mid V_{t-1}) & =E\{E(h_{p}\mid H_t)\mid V_{t-1}\}=E(A_t^{\bullet}\mid V_{t-1})=0,
\end{align*}
and
\begin{align}
 & E\{h_{p}(T_t^{\bullet})^{-1}T_t^{-1}W_t\sigma_{t+1}\mid V_{t-1}\}\nonumber \\
= & E\{A_t^{\bullet}(T_t^{\bullet})^{-1}T_t^{-1}W_t\sigma_{t+1}\mid V_{t-1}\} -E\{E(Q_tA_t^{\bullet}\mid V_{t-1})(T_t^{\bullet})^{-2}T_t^{-2}W_t^{2}\sigma_{t+1}^{2}\mid V_{t-1}\}\nonumber \\
= & E\{A_t^{\bullet}(T_t^{\bullet})^{-1}T_t^{-1}W_tE(\sigma_{t+1}\mid H_t)\mid V_{t-1}\} -E(Q_tA_t^{\bullet}\mid V_{t-1})(T_t^{\bullet})^{-2}E\{T_t^{-2}E(W_t^{2}\sigma_{t+1}^{2}\mid H_t)\mid V_{t-1}\}\nonumber \\
= & 0-E(Q_tA_t^{\bullet}\mid V_{t-1})(T_t^{\bullet})^{-2}E(T_t^{-2}T_t\mid V_{t-1})\nonumber \\
= & -E(Q_tA_t^{\bullet}\mid V_{t-1})(T_t^{\bullet})^{-1},\label{eq:lem-aux-7-proofuse-1}
\end{align}
where the second to last equality in (\ref{eq:lem-aux-7-proofuse-1})
follows from (\ref{lem:aux-comp-useful-result}). Plugging them into
(\ref{eq:lem-aux-7-statement-2}) yields
\begin{align*}
O_{3}^{*}O_{4}^{*}(h_{p}) & =E(h_{p}\mid H_t)-E(h_{p}\mid V_{t-1})-E\{h_{p}(T_t^{\bullet})^{-1}T_t^{-1}W_t\sigma_{t+1}\mid V_{t-1}\}Q_t\\
 & =A_t^{\bullet}+E(Q_tA_t^{\bullet}\mid V_{t-1})(T_t^{\bullet})^{-1}Q_t.
\end{align*}
A functional analysis result implies that $O_{3}^{*}O_{4}^{*}\{h(\sigma_{t+1},V_t)\}=O_{3}^{*}O_{4}^{*}(h_{p})$,
i.e., 
\begin{equation}
O_{3}^{*}O_{4}^{*}(h)=A_t^{\bullet}+E(Q_tA_t^{\bullet}\mid V_{t-1})(T_t^{\bullet})^{-1}Q_t.\label{eq:lem-aux-7-proofuse-2}
\end{equation}
Multiply both sides by $Q_t$ then take $E(\cdot\mid V_{t-1})$,
(\ref{eq:lem-aux-7-proofuse-2}) becomes
\begin{equation}
E\{O_{3}^{*}O_{4}^{*}(h)Q_t\mid V_{t-1}\}=E(A_t^{\bullet}Q_t\mid V_{t-1})+E(A_t^{\bullet}Q_t\mid V_{t-1})(T_t^{\bullet})^{-1}E(Q_t^{2}\mid V_{t-1}).\label{eq:lem-aux-7-proofuse-3}
\end{equation}
Noting that $E(Q_t\mid V_{t-1})=0$, (\ref{eq:lem-aux-7-proofuse-3})
implies that
\[
E(A_t^{\bullet}Q_t\mid V_{t-1})=\frac{E\{O_{3}^{*}O_{4}^{*}(h)Q_t\mid V_{t-1}\}}{1+(T_t^{\bullet})^{-1}\var(Q_t\mid V_{t-1})}.
\]
Plugging into (\ref{eq:lem-aux-7-proofuse-2}) and we have
\[
A_t^{\bullet}=O_{3}^{*}O_{4}^{*}(h)-\frac{E\{O_{3}^{*}O_{4}^{*}(h)Q_t\mid V_{t-1}\}Q_t}{T_t^{\bullet}+\var(Q_t\mid V_{t-1})}.
\]
This implies that
\begin{align*}
E(Q_tA_t^{\bullet}\mid V_{t-1}) & =E\{Q_tO_{3}^{*}O_{4}^{*}(h)\mid V_{t-1}\}-\frac{E\{O_{3}^{*}O_{4}^{*}(h)Q_t\mid V_{t-1}\}E(Q_t^{2}\mid V_{t-1})}{T_t^{\bullet}+\var(Q_t\mid V_{t-1})}\\
 & =\frac{E\{O_{3}^{*}O_{4}^{*}(h)Q_t\mid V_{t-1}\}T_t^{\bullet}}{T_t^{\bullet}+\var(Q_t\mid V_{t-1})}.
\end{align*}
Therefore, by the definition of $h_{p}$ in (\ref{eq:lem-aux-7-proofuse-0.5})
we have
\begin{align*}
h_{p} & =O_{3}^{*}O_{4}^{*}(h)-\frac{E\{O_{3}^{*}O_{4}^{*}(h)Q_t\mid V_{t-1}\}Q_t}{T_t^{\bullet}+\var(Q_t\mid V_{t-1})}\\
 & ~~-\frac{E\{O_{3}^{*}O_{4}^{*}(h)Q_t\mid V_{t-1}\}T_t^{-1}W_t\sigma_{t+1}}{T_t^{\bullet}+\var(Q_t\mid V_{t-1})}\\
 & =O_{3}^{*}O_{4}^{*}(h)-E\{O_{3}^{*}O_{4}^{*}(h)Q_t\mid V_{t-1}\}\frac{T_t^{-1}W_t\sigma_{t+1}+Q_t}{\var(T_t^{-1}W_t\sigma_{t+1}+Q_t\mid V_{t-1})}\\
 & =O_{3}^{*}O_{4}^{*}(h)-E\{O_{3}^{*}O_{4}^{*}(h)Q_t\mid V_{t-1}\}W_{t,t-1}\epsilon_t,
\end{align*}
where the second to last equality follows from Lemma \ref{lem:aux-comp-useful-result}.
This proves (\ref{eq:lem-aux-7-statement-1}).

If $h=h(V_t)$, then
\begin{align*}
 & E\{h(V_t)(T_t^{\bullet})^{-1}T_t^{-1}W_t\sigma_{t+1}\mid V_{t-1}\}Q_t\\
= & E\{h(V_t)(T_t^{\bullet})^{-1}T_t^{-1}W_tE(\sigma_{t+1}\mid V_t)\mid V_{t-1}\}Q_t=0,
\end{align*}
so $O_{3}^{*}O_{4}^{*}(h)=E(h\mid H_t)-E(h\mid V_{t-1})$ and 
\begin{align*}
 & E\{O_{3}^{*}O_{4}^{*}(h)Q_t\mid V_{t-1}\}\\
= & E\{E(h\mid H_t)Q_t\mid V_{t-1}\}-E\{E(h\mid V_{t-1})Q_t\mid V_{t-1}\}\\
= & E(hQ_t\mid V_{t-1}).
\end{align*}
This proves (\ref{eq:lem-aux-7-statement-3}). If $h=h(V_{t-1})$,
then $O_{3}^{*}O_{4}^{*}(h)=0$ and hence $\Pi\{h(V_{t-1})\mid\tilde{\Gamma}_t^{4}\}=0$.
This completes the proof.
\end{proof}
\begin{lemma}[Projection onto $\tilde{\Lambda}_t^{6}$]
\label{lem:aux-8-Lambda-tilde-6}For any $h\in\mathcal{H}$, $\Pi(h\mid\tilde{\Lambda}_t^{6})=E(h\epsilon_t\mid V_{t-1})W_{t,t-1}\epsilon_t$,
where $\epsilon_t=T_t^{-1}W_t\sigma_{t+1}+Q_t$ and $W_{t,t-1}=\var(\epsilon_t\mid V_{t-1})^{-1}$.
\end{lemma}

\begin{proof}
Because $\tilde{\Lambda}_t^{6}=\{a(V_{t-1})\epsilon_t:a(V_{t-1})\text{ is any function}\in\mathbb{R}^{p}\}$,
we have $E(h\epsilon_t\mid V_{t-1})W_{t,t-1}\epsilon_t\in\tilde{\Lambda}_t^{6}$.
So it suffices to show that for any $a(V_{t-1})\epsilon_t\in\tilde{\Lambda}_t^{6}$,
$h-E(h\epsilon_t\mid V_{t-1})W_{t,t-1}\epsilon_t\perp a(V_{t-1})\epsilon_t$.
Because $E(\epsilon_t^{2}\mid V_{t-1})=W_{t,t-1}^{-1}$, we have
\begin{align*}
E\{E(h\epsilon_t\mid V_{t-1})W_{t,t-1}\epsilon_t^{2}a(V_{t-1})\} & =E\{E(h\epsilon_t\mid V_{t-1})W_{t,t-1}E(\epsilon_t^{2}\mid V_{t-1})a(V_{t-1})\}\\
 & =E\{E(h\epsilon_t\mid V_{t-1})a(V_{t-1})\}\\
 & =E\{h\epsilon_ta(V_{t-1})\},
\end{align*}
thus
\[
E[\{h-E(h\epsilon_t\mid V_{t-1})W_{t,t-1}\epsilon_t\}a(V_{t-1})\epsilon_t]=0.
\]
This completes the proof.
\end{proof}
\begin{remark*}
$\tilde{\Lambda}_t^{6}$ may not be the image of a linear operator
on $\mathcal{H}$. So instead of directly deriving the adjoing operator
for $\tilde{\Lambda}_t^{6}$, the form of the projection in Lemma
\ref{lem:aux-8-Lambda-tilde-6} is obtained by first considering the
projection onto the following subspace of $\tilde{\Lambda}_t^{6}$:
\[
\hat{\Lambda}_t^{6}=O_{5}(\mathcal{H})=\{E(h\mid V_{t-1})\epsilon_t:h\in\mathcal{H}\}.
\]
The adjoint operator $O_{5}^{*}$ for $O_{5}$ can be derived as follows.
For any $h,g\in\mathcal{H}$,
\begin{align*}
<O_{5}(h),g> & =E\{E(h\mid V_{t-1})\epsilon_tg\}\\
 & =E[E\{E(h\mid V_{t-1})\epsilon_tg\mid V_{t-1}\}]\\
 & =E\{E(h\mid V_{t-1})E(\epsilon_tg\mid V_{t-1})\}\\
 & =E\{hE(\epsilon_tg\mid V_{t-1})\},
\end{align*}
so $O_{5}^{*}(g)=E(\epsilon_tg\mid V_{t-1})$. Now, suppose for
a given $h\in\mathcal{H}$, $\Pi(h\mid\hat{\Lambda}_t^{6})=h_{p}(V_{t-1})\epsilon_t$
for some $h_{p}(V_{t-1})$ satisfying $E(h_{p})=0$. By a functional
analysis result we have
\[
O_{5}^{*}\{h_{p}(V_{t-1})\epsilon_t\}=O_{5}^{*}(h),
\]
i.e.,
\[
E\{\epsilon_t^{2}h_{p}(V_{t-1})\mid V_{t-1}\}=E(\epsilon_th\mid V_{t-1}).
\]
Since $E(\epsilon_t^{2}\mid V_{t-1})=W_{t,t-1}^{-1}$, the above
display implies $h_{p}(V_{t-1})=W_{t,t-1}E(\epsilon_th\mid V_{t-1})$,
and thus $\Pi(h\mid\hat{\Lambda}_t^{6})=E(\epsilon_th\mid V_{t-1})W_{t,t-1}\epsilon_t$.
Having obtained this, we then verified by definition that it is also
the projection onto $\tilde{\Lambda}_t^{6}$ in the proof of Lemma
\ref{lem:aux-8-Lambda-tilde-6}.
\end{remark*}
\begin{lemma}[Projection of $h(\sigma_{t+1},V_t)$ onto $\Lambda_t^{\perp}$]
\label{lem:aux-9-projection-Lambda-t-part1}For any $B=b(V_{T+1})\in\mathcal{H}$,
let $h(\sigma_{t+1},V_t)=E(B\mid\sigma_{t+1},V_t)-E(B\mid V_t)$.
Then the projection of $h(\sigma_{t+1},V_t)$ onto $\Lambda_t^{\perp}$
is
\[
\Pi\{h(\sigma_{t+1},V_t)\mid\Lambda_t^{\perp}\}=\{R_t-T_t^{-1}E(R_tW_t\mid X_t)\}W_t\sigma_{t+1},
\]
where $R_t=E(B\sigma_{t+1}\mid V_t)$, and $W_t,T_t$ are
defined in Lemma \ref{lem:nuisance-space-orthogonalization}.
\end{lemma}

\begin{proof}
By Lemma \ref{lem:nuisance-space-orthogonalization} and Lemma \ref{lem:aux-hilbert-space-sequential-projection},
to compute $\Pi\{h(\sigma_{t+1},V_t)\mid\Lambda_t^{\perp}\}$
it suffices to calculate sequentially the projection of $h(\sigma_{t+1},V_t)$
onto $\Lambda_t^{1,\perp},\Lambda_t^{2,\perp},(\bigoplus_{m=1}^t\Gamma_{m}^{3})^{\perp},(\bigoplus_{m=1}^{t-1}\Lambda_{m}^{\bullet})^{\perp},\tilde{\Lambda}_t^{5,\perp},\tilde{\Gamma}_t^{4,\perp}$
and $\tilde{\Lambda}_t^{6,\perp}$.
\begin{enumerate}
\item[(i)] For any $A_t^{1}(V_{T+1})\in\Lambda_t^{1}$, $E\{h(\sigma_{t+1},V_t)A_t^{1}\}=E\{h(\sigma_{t+1},V_t)E(A_t^{1}\mid V_t,Y_{t,1})\}=0$.
So $h(\sigma_{t+1},V_t)\in\Lambda_t^{1,\perp}$ and $\Pi(h\mid\Lambda_t^{1,\perp})=h$.
\item[(ii)] By Lemma \ref{lem:aux-4-projection-1} and the fact that $E(h\mid V_t)=0$,
we have $\Pi(h\mid\Lambda_t^{2})=h-E(h\sigma_{t+1}\mid V_t)W_t\sigma_{t+1}$.
Note that
\begin{align*}
E(h\sigma_{t+1}\mid V_t) & =E\{E(B\mid\sigma_{t+1},V_t)\sigma_{t+1}\mid V_t\}-E\{E(B\mid V_t)\sigma_{t+1}\mid V_t\}\\
 & =E(B\sigma_{t+1}\mid V_t)-0=R_t,
\end{align*}
so we have
\[
\Pi(h\mid\Lambda_t^{2,\perp})=h-\Pi(h\mid\Lambda_t^{2})=R_tW_t\sigma_{t+1}.
\]
\item[(iii)] For any $g(V_t)\in\mathcal{H}$, we have $E\{g(V_t)R_tW_t\sigma_{t+1}\}=E\{g(V_t)R_tW_tE(\sigma_{t+1}\mid V_t)\}=0$,
so $R_tW_t\sigma_{t+1}\in(\bigoplus_{m=1}^t\Gamma_{m}^{3})^{\perp}$
and $R_tW_t\sigma_{t+1}\in(\bigoplus_{m=1}^{t-1}\Lambda_{m}^{\bullet})^{\perp}$.
Therefore, $\Pi\{R_tW_t\sigma_{t+1}\mid(\bigoplus_{m=1}^t\Gamma_{m}^{3})^{\perp}\}=R_tW_t\sigma_{t+1}$
and $\Pi\{R_tW_t\sigma_{t+1}\mid(\bigoplus_{m=1}^{t-1}\Lambda_{m}^{\bullet})^{\perp}\}=R_tW_t\sigma_{t+1}$.
\item[(vi)] To use Lemma \ref{lem:aux-6-Lambda-tilde-5} to compute $\Pi(R_tW_t\sigma_{t+1}\mid\tilde{\Lambda}_t^{5,\perp})$,
we first calculate a few terms:
\[
D_t=E(R_tW_t^{2}\sigma_{t+1}^{2}\mid H_t)=E\{R_tW_t^{2}E(\sigma_{t+1}^{2}\mid V_t)\mid H_t\}=E(R_tW_t\mid H_t),
\]
\begin{align*}
\tilde{A}_t^{\bullet}: & =-E(D_tT_t^{-1}\mid V_{t-1})(T_t^{\bullet}T_t)^{-1}+D_tT_t^{-1}\\
 & =-E\{E(R_tW_t\mid H_t)T_t^{-1}\mid V_{t-1}\}(T_t^{\bullet}T_t)^{-1}+E(R_tW_t\mid H_t)T_t^{-1}\\
 & =-E(R_tW_tT_t^{-1}\mid V_{t-1})(T_t^{\bullet}T_t)^{-1}+E(R_tW_t\mid H_t)T_t^{-1}\\
 & =-R_{t-1}(T_t^{\bullet})^{-1}T_t^{-1}+E(R_tW_t\mid H_t)T_t^{-1},
\end{align*}
where we define $R_{t-1}=E(R_tW_tT_t^{-1}\mid V_{t-1})$. So
by Lemma \ref{lem:aux-6-Lambda-tilde-5} we have
\[
\Pi(R_tW_t\sigma_{t+1}\mid\tilde{\Lambda}_t^{5})=\tilde{A}_t^{\bullet}W_t\sigma_{t+1}=\{E(R_tW_t\mid H_t)-R_{t-1}(T_t^{\bullet})^{-1}\}T_t^{-1}W_t\sigma_{t+1},
\]
and
\begin{align*}
\Pi(R_tW_t\sigma_{t+1}\mid\tilde{\Lambda}_t^{5,\perp}) & =R_tW_t\sigma_{t+1}-\Pi(R_tW_t\sigma_{t+1}\mid\tilde{\Lambda}_t^{5})\\
 & =\{R_tT_t+R_{t-1}(T_t^{\bullet})^{-1}-E(R_tW_t\mid H_t)\}T_t^{-1}W_t\sigma_{t+1}\\
 & \equiv h_{1}(V_t)T_t^{-1}W_t\sigma_{t+1},
\end{align*}
where we define $h_{1}(V_t)=R_tT_t+R_{t-1}(T_t^{\bullet})^{-1}-E(R_tW_t\mid H_t)$.
\item[(v)] Now we will use Lemma \ref{lem:aux-7-Gamma-tilde-4} to compute $\Pi\{h_{1}(V_t)T_t^{-1}W_t\sigma_{t+1}\mid\tilde{\Gamma}_t^{4,\perp}\}$.
Since $E(\sigma_{t+1}\mid V_t)=0$, we have 
\[
E\{h_{1}(V_t)T_t^{-1}W_t\sigma_{t+1}\mid H_t\}=E\{h_{1}(V_t)T_t^{-1}W_t\sigma_{t+1}\mid V_{t-1}\}=0.
\]
We also have (using $E(\sigma_{t+1}^{2}\mid V_t)=W_t^{-1}$)
\begin{align}
 & E\{h_{1}(V_t)T_t^{-1}W_t\sigma_{t+1}\times(T_t^{\bullet})^{-1}T_t^{-1}W_t\sigma_{t+1}\mid V_{t-1}\}\nonumber \\
= & E\{h_{1}(V_t)(T_t^{\bullet})^{-1}T_t^{-2}W_t^{2}\sigma_{t+1}^{2}\mid V_{t-1}\}\nonumber \\
= & E\{h_{1}(V_t)(T_t^{\bullet})^{-1}T_t^{-2}W_t\mid V_{t-1}\}\nonumber \\
= & E\{R_t(T_t^{\bullet})^{-1}T_t^{-1}W_t\mid V_{t-1}\}+E\{R_{t-1}(T_t^{\bullet})^{-2}T_t^{-2}W_t\mid V_{t-1}\}\nonumber \\
 & -E\{E(R_tW_t\mid H_t)(T_t^{\bullet})^{-1}T_t^{-2}W_t\mid V_{t-1}\}.\label{eq:lem-aux-8-proofuse-1}
\end{align}
We compute out each term in (\ref{eq:lem-aux-8-proofuse-1}):
\begin{align*}
E\{R_t(T_t^{\bullet})^{-1}T_t^{-1}W_t\mid V_{t-1}\} & =E\{R_tT_t^{-1}W_t\mid V_{t-1}\}(T_t^{\bullet})^{-1}=R_{t-1}(T_t^{\bullet})^{-1},\\
E\{R_{t-1}(T_t^{\bullet})^{-2}T_t^{-2}W_t\mid V_{t-1}\} & =R_{t-1}(T_t^{\bullet})^{-2}E\{T_t^{-2}W_t\mid V_{t-1}\}\\
 & =R_{t-1}(T_t^{\bullet})^{-2}E\{T_t^{-2}E(W_t\mid H_t)\mid V_{t-1}\}\\
 & =R_{t-1}(T_t^{\bullet})^{-2}E(T_t^{-1}\mid V_{t-1})=R_{t-1}(T_t^{\bullet})^{-1},\\
E\{E(R_tW_t\mid H_t)(T_t^{\bullet})^{-1}T_t^{-2}W_t\mid V_{t-1}\} & =E\{E(R_tW_tT_t^{-2}\mid H_t)E(W_t\mid H_t)\mid V_{t-1}\}(T_t^{\bullet})^{-1}\\
 & =E\{E(R_tW_tT_t^{-1}\mid H_t)\mid V_{t-1}\}(T_t^{\bullet})^{-1}\\
 & =R_{t-1}(T_t^{\bullet})^{-1}.
\end{align*}
Hence, (\ref{eq:lem-aux-8-proofuse-1}) becomes
\[
E\{h_{1}(V_t)T_t^{-1}W_t\sigma_{t+1}\times(T_t^{\bullet})^{-1}T_t^{-1}W_t\sigma_{t+1}\mid V_{t-1}\}=R_{t-1}(T_t^{\bullet})^{-1}.
\]
By the definition of $O_{3}^{*}$ and $O_{4}^{*}$ in Lemma \ref{lem:aux-7-Gamma-tilde-4},
we have
\[
O_{3}^{*}O_{4}^{*}\{h_{1}(V_t)T_t^{-1}W_t\sigma_{t+1}\}=-R_{t-1}(T_t^{\bullet})^{-1}Q_t,
\]
and
\[
E[O_{3}^{*}O_{4}^{*}\{h_{1}(V_t)T_t^{-1}W_t\sigma_{t+1}\}Q_t\mid V_{t-1}]=-R_{t-1}(T_t^{\bullet})^{-1}\var(Q_t\mid V_{t-1}).
\]
With the above computation, Lemma \ref{lem:aux-7-Gamma-tilde-4} implies
that
\begin{align*}
& \phantom{=} \Pi\{h_{1}(V_t)T_t^{-1}W_t\sigma_{t+1}\mid\tilde{\Gamma}_t^{4}\} \\
& =O_{3}^{*}O_{4}^{*}\{h_{1}(V_t)T_t^{-1}W_t\sigma_{t+1}\}-E[O_{3}^{*}O_{4}^{*}\{h_{1}(V_t)T_t^{-1}W_t\sigma_{t+1}\}Q_t\mid V_{t-1}]W_{t,t-1}\epsilon_t\\
 & =-R_{t-1}(T_t^{\bullet})^{-1}Q_t+R_{t-1}(T_t^{\bullet})^{-1}\var(Q_t\mid V_{t-1})W_{t,t-1}\epsilon_t.
\end{align*}
Thus, the projection $\Pi\{h_{1}(V_t)T_t^{-1}W_t\sigma_{t+1}\mid\tilde{\Gamma}_t^{4,\perp}\}$
equals
\begin{align*}
 & \Pi\{h_{1}(V_t)T_t^{-1}W_t\sigma_{t+1}\mid\tilde{\Gamma}_t^{4,\perp}\}\\
= & h_{1}(V_t)T_t^{-1}W_t\sigma_{t+1}-\Pi\{h_{1}(V_t)T_t^{-1}W_t\sigma_{t+1}\mid\tilde{\Gamma}_t^{4}\}\\
= & \{R_tT_t+R_{t-1}(T_t^{\bullet})^{-1}-E(R_tW_t\mid H_t)\}T_t^{-1}W_t\sigma_{t+1} \\
&+R_{t-1}(T_t^{\bullet})^{-1}Q_t-R_{t-1}(T_t^{\bullet})^{-1}\var(Q_t\mid V_{t-1})W_{t,t-1}\epsilon_t\\
= & \{R_t-T_t^{-1}E(R_tW_t\mid H_t)\}W_t\sigma_{t+1}+R_{t-1}(T_t^{\bullet})^{-1}\epsilon_t-R_{t-1}(T_t^{\bullet})^{-1}\var(Q_t\mid V_{t-1})W_{t,t-1}\epsilon_t\\
= & \{R_t-T_t^{-1}E(R_tW_t\mid H_t)\}W_t\sigma_{t+1}+R_{t-1}W_{t,t-1}\epsilon_t,
\end{align*}
where the last equality follows from Lemma \ref{lem:aux-comp-useful-result}.
\item[(vi)] Denote by $h_{2}(\sigma_{t+1},V_t)=\{R_t-T_t^{-1}E(R_tW_t\mid H_t)\}W_t\sigma_{t+1}$
and $h_{3}(\sigma_{t+1},V_t)=R_{t-1}W_{t,t-1}\epsilon_t$. Now
we will use Lemma \ref{lem:aux-8-Lambda-tilde-6} to compute $\Pi(h_{2}+h_{3}\mid\tilde{\Lambda}_t^{6,\perp})$.
We first calculate a few terms:
\begin{equation}
E\{h_{2}(\sigma_{t+1},V_t)\epsilon_t\mid V_{t-1}\}=E(R_tW_t\sigma_{t+1}\epsilon_t\mid V_{t-1})-E\{E(R_tW_t\mid H_t)T_t^{-1}W_t\sigma_{t+1}\epsilon_t\mid V_{t-1}\}.\label{eq:lem-aux-8-proofuse-2}
\end{equation}
Using the fact that $E(\sigma_{t+1}\mid V_t)=0$ and $E(\sigma_{t+1}^{2}\mid V_t)=W_t^{-1}$,
we have. By Lemma \ref{lem:aux-comp-useful-result}(iv) we have
\[
E(R_tW_t\sigma_{t+1}\epsilon_t\mid V_{t-1})=E\{R_tW_tE(\sigma_{t+1}\epsilon_t\mid V_t)\mid V_{t-1}\}=E(R_tW_tT_t^{-1}\mid V_{t-1})=R_{t-1},
\]
and
\begin{align*}
 & E\{E(R_tW_t\mid H_t)T_t^{-1}W_t\sigma_{t+1}\epsilon_t\mid V_{t-1}\}\\
= & E\{E(R_tW_t\mid H_t)T_t^{-1}W_tE(\sigma_{t+1}\epsilon_t\mid V_t)\mid V_{t-1}\}\\
= & E\{E(R_tW_tT_t^{-2}\mid H_t)W_t\mid V_{t-1}\}\\
= & E\{R_tW_tT_t^{-2}E(W_t\mid H_t)\mid V_{t-1}\}\\
= & E(R_tW_tT_t^{-1}\mid V_{t-1})=R_{t-1}.
\end{align*}
By plugging these into (\ref{eq:lem-aux-8-proofuse-2}), we can use Lemma \ref{lem:aux-8-Lambda-tilde-6}
to derive that
\[
\Pi\{h_{2}(\sigma_{t+1},V_t)\mid\tilde{\Lambda}_t^{6}\}=E\{h_{2}(\sigma_{t+1},V_t)\epsilon_t\mid V_{t-1}\}W_{t,t-1}\epsilon_t=0,
\]
i.e., $\Pi\{h_{2}(\sigma_{t+1},V_t)\mid\tilde{\Lambda}_t^{6,\perp}\}=h_{2}(\sigma_{t+1},V_t)$.
On the other hand, by definition we have $h_{3}(\sigma_{t+1},V_t)=R_{t-1}W_{t,t-1}\epsilon_t\in\tilde{\Lambda}_t^{6}$.
So
\[
\Pi(h_{2}+h_{3}\mid\tilde{\Lambda}_t^{6,\perp})=h_{2}(\sigma_{t+1},V_t)=\{R_t-T_t^{-1}E(R_tW_t\mid H_t)\}W_t\sigma_{t+1}.
\]
This completes the proof.
\end{enumerate}
\end{proof}
\begin{lemma}[Projection of $h(H_t)$ onto $\Lambda_t^{\perp}$]
\label{lem:aux-9-projection-Lambda-t-part2}For any $B=b(V_{T+1})\in\mathcal{H}$,
let $h(H_t)=E(B\mid H_t)-E(B\mid V_{t-1})$. Then $h(H_t)\in\Lambda_t$,
i.e., $\Pi\{h(H_t)\mid\Lambda_t^{\perp}\}=0$.
\end{lemma}

\begin{proof}
By Lemma \ref{lem:nuisance-space-orthogonalization} and Lemma \ref{lem:aux-hilbert-space-sequential-projection},
to compute $\Pi\{h(H_t)\mid\Lambda_t^{\perp}\}$ it suffices to
calculate sequentially the projection of $h(H_t)$ onto $\Lambda_t^{1,\perp},\Lambda_t^{2,\perp},(\bigoplus_{m=1}^t\Gamma_{m}^{3})^{\perp},(\bigoplus_{m=1}^{t-1}\Lambda_{m}^{\bullet})^{\perp},\tilde{\Lambda}_t^{5,\perp},\tilde{\Gamma}_t^{4,\perp}$
and $\tilde{\Lambda}_t^{6,\perp}$.
\begin{enumerate}
\item[(i)] For any $A_t^{1}(V_{T+1})\in\Lambda_t^{1}$, $E\{h(H_t)\Lambda_t^{1}\}=E\{h(H_t)E(\Lambda_t^{1}\mid V_t,Y_{t,1})\}=0$.
So $h(H_t)\in\Lambda_t^{1,\perp}$ and $\Pi(h\mid\Lambda_t^{1,\perp})=h$.
\item[(ii)] For any $A_t^{2}(\sigma_{t+1},V_t)\in\Lambda_t^{2}$, $E\{h(H_t)A_t^{2}\}=E\{h(H_t)E(A_t^{2}\mid V_t)\}=0$.
So $h(H_t)\in\Lambda_t^{2,\perp}$ and $\Pi(h\mid\Lambda_t^{2,\perp})=h$.
\item[(iii)] For any $A_{m}^{3}(V_{m})\in\Gamma_{m}^{3}$ with $1\leq m\leq t-1$,
we have $E\{h(H_t)A_{m}^{3}(V_{m})\}=E[E\{h(H_t)\mid V_{t-1}\}A_{m}^{3}(V_{m})]=0$.
For any $A_t^{3}(V_t)\in\Gamma_t^{3}$, we have $E\{h(H_t)A_t^{3}(V_t)\}=E[h(H_t)E\{A_{m}^{3}(V_{m})\mid H_t\}]=0$.
So $h(H_t)\in(\bigoplus_{m=1}^t\Gamma_{m}^{3})^{\perp}$ and $\Pi\{h\mid(\bigoplus_{m=1}^t\Gamma_{m}^{3})^{\perp}\}=h$.
\item[(iv)] For any $A_{m}^{\bullet}(H_{m})\in\Lambda_{m}^{\bullet}$ with $1\leq m\leq t-1$,
we have $E\{h(H_t)A_{m}^{\bullet}(H_{m})\}=E[E\{h(H_t)\mid V_{t-1}\}A_{m}^{\bullet}(H_{m})]=0$.
So $h(H_t)\in(\bigoplus_{m=1}^{t-1}\Lambda_{m}^{\bullet})^{\perp}$
and $\Pi\{h\mid(\bigoplus_{m=1}^{t-1}\Lambda_{m}^{\bullet})^{\perp}\}=h$.
\item[(v)] We have $D_t=E\{h(H_t)W_t\sigma_{t+1}\mid H_t\}=E\{h(H_t)W_tE(\sigma_{t+1}\mid V_t)\mid H_t\}=0$,
so by Lemma \ref{lem:aux-6-Lambda-tilde-5} $\Pi\{h(H_t)\mid\tilde{\Lambda}_t^{5}\}=0$
and $\Pi\{h(H_t)\mid\tilde{\Lambda}_t^{5,\perp}\}=h(H_t)$.
\item[(vi)] By Lemma \ref{lem:aux-7-Gamma-tilde-4}, using the fact that $E(Q_t\mid V_{t-1})=0$,
we have
\begin{align*}
\Pi\{h(H_t)\mid\tilde{\Gamma}_t^{4}\} & =E(h\mid H_t)-E(h\mid V_{t-1})-E(hQ_t\mid V_{t-1})W_{t,t-1}\epsilon_t\\
 & =h(H_t)-E\{E(B\mid H_t)Q_t-E(B\mid V_{t-1})Q_t\mid V_{t-1}\}W_{t,t-1}\epsilon_t\\
 & =h(H_t)-E(BQ_t\mid V_{t-1})W_{t,t-1}\epsilon_t,
\end{align*}
so $\Pi\{h(H_t)\mid\tilde{\Gamma}_t^{4,\perp}\}=h(H_t)-\Pi\{h(H_t)\mid\tilde{\Gamma}_t^{4}\}=E(BQ_t\mid V_{t-1})W_{t,t-1}\epsilon_t$.
\item[(vii)] By definition we have $E(BQ_t\mid V_{t-1})W_{t,t-1}\epsilon_t\in\tilde{\Lambda}_t^{6}$,
so $\Pi\{E(BQ_t\mid V_{t-1})W_{t,t-1}\epsilon_t\mid\tilde{\Lambda}_t^{6,\perp}\}=0$.
This completes the proof.
\end{enumerate}
\end{proof}
\begin{lemma}
\label{lem:aux-comp-useful-result}With $W_t=\var(\sigma_{t+1}\mid V_t)^{-1}$,
$T_t=E(W_t\mid H_t)$, $T_t^{\bullet}=E(T_t^{-1}\mid V_{t-1})$,
$W_{t,t-1}=\var(T_t^{-1}W_t\sigma_{t+1}+Q_t\mid V_{t-1})^{-1}$,
we have
\begin{enumerate}
\item[(i)] $E(W_t^{2}\sigma_{t+1}^{2}\mid H_t)=T_t$.
\item[(ii)] $W_{t,t-1}^{-1}=\var(T_t^{-1}W_t\sigma_{t+1}\mid V_{t-1})+\var(Q_t\mid V_{t-1})=T_t^{\bullet}+\var(Q_t\mid V_{t-1})$.
\item[(iii)] $1-\var(Q_t\mid V_{t-1})W_{t,t-1}=T_t^{\bullet}$.
\item[(iv)] $E(\sigma_{t+1}\epsilon_t\mid V_t)=T_t^{-1}$.
\end{enumerate}
\end{lemma}

\begin{proof}
For (i), because $E(\sigma_{t+1}\mid V_t)=0$, we have
\begin{align*}
E(W_t^{2}\sigma_{t+1}^{2}\mid H_t) & =E\{\var(\sigma_{t+1}\mid V_t)^{-2}\sigma_{t+1}^{2}\mid H_t\}\\
 & =E[E\{\var(\sigma_{t+1}\mid V_t)^{-2}\sigma_{t+1}^{2}\mid V_t\}\mid H_t]\\
 & =E\{\var(\sigma_{t+1}\mid V_t)^{-1}\mid H_t\}=T_t.
\end{align*}

For (ii), we have

\begin{align*}
 & \var(T_t^{-1}W_t\sigma_{t+1}+Q_t\mid V_{t-1})=E\{(T_t^{-1}W_t\sigma_{t+1}+Q_t)^{2}\mid V_{t-1}\}\\
= & E(T_t^{-2}W_t^{2}\sigma_{t+1}^{2}\mid V_{t-1})+E(Q_t^{2}\mid V_{t-1})\\
= & E(T_t^{-1}\mid V_{t-1})+\var(Q_t\mid V_{t-1})=T_t^{\bullet}+\var(Q_t\mid V_{t-1}).
\end{align*}

(iii) is an immediate implication of (ii).

For (iv), we have
\begin{align*}
E(\sigma_{t+1}\epsilon_t\mid V_t) & =E\{\sigma_{t+1}(T_t^{-1}W_t\sigma_{t+1}+Q_t)\mid V_t\}\\
 & =T_t^{-1}W_tE(\sigma_{t+1}^{2}\mid V_t)+E(\sigma_{t+1}Q_t\mid V_t)\\
 & =T_t^{-1}W_tW_t^{-1}+0\\
 & =T_t^{-1}.
\end{align*}

This completes the proof.
\end{proof}
\begin{lemma}
\label{lem:aux-hilbert-space-sequential-projection}Suppose $\Lambda_{1}$
and $\Lambda_{2}$ are two subspaces of the Hilbert space $\mathcal{H}$,
and they are orthogonal to each other. Then for any $h\in\mathcal{H}$,
we have
\[
\Pi\{h\mid(\Lambda_{1}\oplus\Lambda_{2})^{\perp}\}=\Pi\{\Pi(h\mid\Lambda_{1}^{\perp})\mid\Lambda_{2}^{\perp}\}.
\]
\end{lemma}

\begin{proof}
This is a standard Hilbert space result. See, for example, \citet{akhiezer2013theory}.
\end{proof}

\end{document}